\documentclass[12pt]{article}
\usepackage{mheck}

\usepackage{blkarray}

\usepackage{tikz} 
\usetikzlibrary{matrix}

\theoremstyle{theorem}
\newtheorem{fact}{Fact}
\newtheorem{corl}{Corollary}
\newtheorem{prop}{Proposition}

\newtheorem*{lem}{Lemma}
\newtheorem*{php}{Physical expectation}

\theoremstyle{definition}
\newtheorem{defn}{Definition}
\newtheorem{rem}{Remark}
\newtheorem{exe}{Example}
\newtheorem{ass}{Assumption}

\date{December, 2016}

\title{Homological $S$--Duality in 4d $\cn=2$ QFTs
}

\authors{Matteo Caorsi\footnote{e-mail: {\tt matteocao@gmail.com}} and Sergio Cecotti\footnote{e-mail: {\tt cecotti@sissa.it}}\vskip 9pt

\centerline{SISSA, via Bonomea 265, I-34100 Trieste, ITALY}}

\abstract{The $S$--duality group $\mathbb{S}(\cf)$ of a 4d $\cn=2$ supersymmetric theory $\cf$ is identified with 
the group of triangle equivalences of its cluster category $\mathscr{C}(\cf)$ modulo the subgroup acting trivially on the physical quantities. $\mathbb{S}(\cf)$ is a discrete group commensurable to a subgroup of the Siegel modular group $Sp(2g,\Z)$ ($g$ being the dimension of the Coulomb branch).  This identification reduces the determination of the $S$-duality group of a given $\cn=2$ theory to a problem in homological algebra.
In this paper we describe the techniques which make the computation straightforward for a large class of $\cn=2$ QFTs. The group $\mathbb{S}(\cf)$ is naturally presented as a generalized braid group.

The $S$-duality groups are often larger than expected. In some models the  enhancement of $S$-duality is quite spectacular. For instance, a QFT with a huge $S$-duality group is the Lagrangian SCFT with gauge group $SO(8)\times SO(5)^3\times SO(3)^6$ and half-hypermultiplets in the bi- and tri-spinor representations.

We focus on four families of examples: the $\cn=2$ SCFTs of the form $(G,G^\prime)$, $D_p(G)$, and $E_r^{(1,1)}(G)$, as well as the asymptotically-free theories $(G,\widehat{H})$ (which contain $\cn=2$ SQCD as a special case).
For the $E_r^{(1,1)}(G)$ models we confirm the presence of the $PSL(2,\Z)$ $S$-duality group predicted by Del Zotto, Vafa and Xie, but for most models in this class $S$-duality gets enhanced to a larger group.}

\begin{document}

\maketitle

\newpage

\tableofcontents

\section{Introduction}

Dualities between quantum theories (with enough supersymmetry) are most conveniently understood as 
exact equivalences between the linear triangle categories which describe their  BPS objects \cite{book}. For instance, mirror symmetry is best described as the equivalence of the bounded derived category of coherent sheaves on one manifold $X$ and the bounded derived Fukaya category of its mirror space $X^\vee$ (homological mirror symmetry \cite{homokon}). $S$-duality of 
a 4d $\cn=2$ QFT, being an internal duality, is described by the auto-equivalences of a single BPS category rather than by the comparison of two \emph{a priori} different categories as in mirror symmetry. The basic example is the $SL(2,\Z)$ duality of $\cn=4$ SYM: through its relation to $T$-duality \cite{Vafa:1997mh}, it gets identified with the  group of auto-equivalences of the derived category of coherent sheaves over an elliptic curve, which has an explicit realization in terms of Fourier-Mukai transforms \cite{FM}.  

For a general $\cn=2$ QFT, it is natural to \emph{define} the group $\mathbb{S}$ of (generalized) 
$S$-dualities as the group $\mathrm{Aut}(\mathscr{C})$ of triangle auto-equivalences of the category $\mathscr{C}$ describing its BPS objects,
modulo the physically trivial ones
\be
\mathbb{S}=\mathrm{Aut}(\mathscr{C})/\text{(physically trivial)}.
\ee
The Kontsevich-Soibelman
wall-crossing formula \cite{kswcf}, and related physical arguments
\cite{Gaiotto:2008cd,Gaiotto:2009hg,framed,Cecotti:2010fi}, show that the appropriate triangle category $\mathscr{C}$ to describe the BPS sector of a $\cn=2$ QFT is the \textit{cluster category}\footnote{\ For nice introductions to cluster categories (and algebras) see refs.\cite{clustercat,clustercat2,kellerrev}.} associated to the mutation class of its BPS quivers \cite{Alim:2011kw}.
The categorical viewpoint reduces the problem of determining the dualities of a given 4d $\cn=2$ theory to a well-posed mathematical problem, which may be tackled with standard methods of homological algebra. 

In this paper we lay down a general framework for homological $S$-duality,
developing ideas and techniques which allow to determine the $S$-duality group $\mathbb{S}$ very explicitly in a large class of rather complicated $\cn=2$ models. The homological  viewpoint leads to a presentation of $\mathbb{S}$ in the form of a higher braid group. As expected, modulo commensurability,
 $\mathbb{S}$ is an arithmetic subgroup of $Sp(2g,\R)$,
where $g$ is the dimension of the Coulomb branch.


In this paper homological $S$-duality is worked out in detail for the large class of 4d $\cn=2$ models  (superconformal or asymptotically-free) which, under the 4d/2d correspondence of ref.\!\cite{Cecotti:2010fi}, are related to well-behaved 2d (2,2) theories. However the idea of homological $S$-duality is more general, and some of our examples
actually do not belong to this  class. 
The advantage of having a nice 2d $(2,2)$ correspondent, is that the
4d BPS category $\mathscr{C}$ may be constructed as the cluster category of the 2d brane category
$\mathscr{B}$ which in many cases is  well understood \cite{orlov1,orlov2,orlov3}.

The class of 4d $\cn=2$ theories whose category
$\mathscr{C}$ we can read from 2d contains, for instance,
$\cn=2$ SQCD 
 with simply-laced gauge groups and matter in representations which are ``nice'' in the sense of \cite{Tachikawa:2011yr}, quiver gauge theories, SCFTs engineered by polynomial singularities (such as the Arnold ones \cite{arnold}),  \emph{etc.} To keep the paper of finite length, we shall focus mainly on four groups of $\cn=2$ models:
 \begin{itemize}
 \item[a)] The $(G,G^\prime)$ SCFTs \cite{Cecotti:2010fi} labelled by two
 simply-laced Dynkin graphs\footnote{ In this paper we abuse notations, and use the same symbol $G$ to denote the Dynkin graph, the corresponding Lie algebra, and its (simply-connected) Lie group.} 
 $G,G^\prime\in ADE$;
  \item[b)] the $(\widehat{H},G)$ asymptotically-free theories \cite{Cecotti:2012jx,Cecotti:2013lda} labelled by a
 simply-laced Dynkin graph $G\in ADE$ and a mutation class of acyclic affine quivers
 \begin{equation}
 \widehat{H}= \widehat{A}(p_1,p_2)\ (p_1\geq p_2\geq 1),\ \text{or }\widehat{D}_r,\ (r\geq 4),\ \text{or } \widehat{E}_6,\ \widehat{E}_7,\ \widehat{E}_8. 
 \end{equation}
In these theories, $G$ is always a factor of the gauge group, and the $\beta$-function of the corresponding gauge coupling is strictly negative; 
 \item[c)] the $D_p(G)$ SCFTs \cite{Cecotti:2012jx,Cecotti:2013lda} labelled by $G\in ADE$ and the period\footnote{\ It is convenient to extend the definition to $p=1$ by declaring $D_1(G)$ to be the empty SCFT for all $G$.} $p\geq 2$. The flavor group of $D_p(G)$ contains $G$ as a subgroup;
 \item[d)] the DZVX models \cite{DelZotto:2015rca,Cecotti:2013lda} labelled by one of the four affine stars ($\widehat{D}_4$, $\widehat{E}_6$, $\widehat{E}_7$, or $\widehat{E}_8$) and a Lie algebra $G\in ADE$.
 Again $G$ is a factor of the gauge group, but now its gauge coupling is exactly marginal.
 \end{itemize}
Many SQCD models and quiver gauge theories are recovered as special cases of b), c), and d) \cite{DelZotto:2015rca,Cecotti:2013lda}. Our methods may be  extended to other classes of theories.

Models b) and d) have the physical interpretation of $\cn=2$ SYM with group $G$ gauging the diagonal $G$-symmetry of a number of $D_{p_i}(G)$
SCFTs (see table \ref{mmmatt}). The interplay between the dualities of the sub-constituents and of the full theory allows to perform many crossed checks between the various models.

\begin{table}
\begin{center}
\begin{tabular}{cc||cc}\hline\hline
model & matter sector & model & matter sector
\\\hline
$(\widehat{A}(p_1,p_2),G)$ &
$D_{p_1}(G)$, $D_{p_2}(G)$ &
$D^{(1,1)}_4(G)$ &
$D_2(G)$, $D_2(G)$, $D_2(G)$, $D_2(G)$\\
$(\widehat{D}_r,G)$ &
$D_{r-2}(G)$, $D_2(G)$, $D_2(G)$
&$E^{(1,1)}_6(G)$ &$D_3(G)$, $D_3(G)$, $D_3(G)$\\
$(\widehat{E}_r,G)$ &
$D_{r-3}(G)$, $D_3(G)$, $D_2(G)$ &
$E^{(1,1)}_7(G)$ &$D_4(G)$,
$D_4(G)$, $D_2(G)$\\
&& $E^{(1,1)}_8(G)$ &$D_6(G)$,
$D_3(G)$, $D_2(G)$\\\hline\hline
\end{tabular}
\caption{\label{mmmatt} \textsc{Left:} the QFT in item b); \textsc{Right:} the SCFT in item c). All theories are written as $G$ SYM coupled to a ``matter'' sector. The subscript $p$ of $D_p(G)$ is called the \emph{period} of the sub-constituent.
By the \emph{set of periods} of a QFT of class b) or d) we mean the list $\{p_1,p_2,\cdots,p_s\}$ of the periods of its ``matter'' sub-systems. }
\end{center}
\end{table}

The homological approach to $S$-duality has been used in ref.\!\cite{shepard} to address the question of the action of the $S$-duality group
$SL(2,\Z)$ on the observables of the four elliptic (\emph{a.k.a.}\! tubular) complete $\cn=2$ SCFT\footnote{ The elliptic complete $\cn=2$ SCFTs \eqref{completetubular} are in one-to-one correspondence with the affine Dynkin graphs which are also stars
\cite{Cecotti:2011rv}.} 
\be\label{completetubular}
D^{(1,1)}_4,\quad E^{(1,1)}_6,\quad
E^{(1,1)}_7,\quad E^{(1,1)}_8,
\ee
 building on \cite{Cecotti:2012va} and the mathematical literature \cite{ringelDER,lenzingmeltzer,meltzer,BKL} on the relevant categories  (i.e.\! the cluster categories of the four tubular canonical algebras). The family of SCFTs \eqref{completetubular} has been generalized by DZVX in ref.\!\cite{DelZotto:2015rca}\footnote{\ For previous work see \cite{Cecotti:2012jx,Cecotti:2013lda}.} : the DZVX models are labelled by one of the four affine stars \eqref{completetubular} together with a simply--laced Lie algebra $G$,
\begin{equation}\label{delvafa}
D^{(1,1)}_4(G),\quad E^{(1,1)}_6(G),\quad
E^{(1,1)}_7(G), \quad E^{(1,1)}_8(G),\qquad G\in ADE.
\end{equation}  
For $G=A_1$ they reduce to the  complete SCFTs \eqref{completetubular}. In ref.\!\cite{DelZotto:2015rca} the SCFTs \eqref{delvafa} were geometrically engineered in $F$--theory:
 the geometry contains an elliptic curve, and it was predicted that
all these SCFTs have (at least)
a $SL(2,\Z)$ group of $S$--dualities which however should act on the physical observables in a counter-intuitive way. The initial motivation of the present work was to check the prediction of \cite{DelZotto:2015rca}, understand the action of $SL(2,\Z)$ on the observables, and discuss the possible enhancements of $S$-duality to a group strictly bigger than $SL(2,\Z)$. 

A simple example may illustrate why
the action of the $S$-duality group
in these models looks rather puzzling. Consider the simplest sequence of such theories,
the $D^{(1,1)}_4(A_{2N-1})$ ones ($N\in\mathbb{N}$). They are the quiver gauge theories in figure \ref{figquiv}; 
in some corner of their parameter spaces they have a weakly coupled Lagrangian formulation.
For $N=1$ the four bifundamentals reduce to fundamental, and we recover $SU(2)$ SQCD with $N_f=4$
which is known\footnote{ See also \cite{shepard} for the corresponding homological analysis.} \cite{SW2} to have a $SL(2,\Z)$ $S$-duality which acts on the $SO(8)$ flavor weights by $Spin(8)$ triality; only the congruence subgroup $\Gamma(2)\subset SL(2,\Z)$ commutes with the flavor. Going to higher $N$'s,
 the $SU(2)$ electric/magnetic
 charges get replaced by the $SU(2N)$ electric/magnetic charges of the central node, while the role of the flavor charges are played by the four $U(1)$ flavor charges of the bifundamentals together with the electric/magnetic charges of the peripheral $SU(N)$ gauge groups. It is expected \cite{DelZotto:2015rca} that we still have a $SL(2,\Z)$
 group rotating the electric and magnetic charges of the central node as before, but then $SL(2,\Z)$ has to act by $Spin(8)$ on all other charges, which means that gauge charges of the distinct peripheral gauge groups should mix together. This conclusion is quite counter-intuitive from the weak-coupling physics, but seems forced on us from geometric engineering. Of course, there is no contradiction, since $S$-duality is quite a strong-coupling property, yet the picture deserves a more detailed analysis which
 may be performed in the homological approach after having developed all the necessary tools.

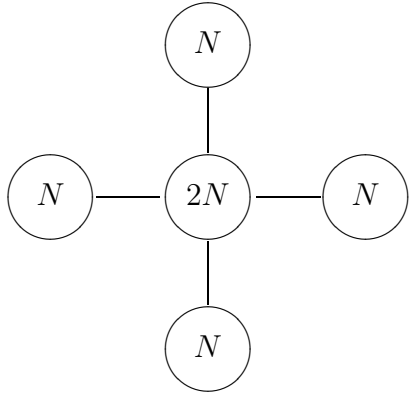
\begin{figure}
$$
\xymatrix{& *++[o][F-]{\,\phantom{\Big|}N
\phantom{\Big|}\,}\ar@{-}[d]\\
*++[o][F-]{\,\phantom{\Big|}N
\phantom{\Big|}\,}
\ar@{-}[r]&
*++[o][F-]{\phantom{\Big|}2N
\phantom{\Big|}}\ar@{-}[r]&
*++[o][F-]{\,\phantom{\Big|}N
\phantom{\Big|}\,}\\
& *++[o][F-]{\,\phantom{\Big|}N
\phantom{\Big|}\,}\ar@{-}[u]
}
$$
\caption{\label{figquiv} The
$D^{(1,1)}(A_{2N-1})$ models as quiver gauge theories. A circle with a $N$ (resp.\! $2N$) stands for a $SU(N)$ (resp.\! $SU(2N)$) $\cn=2$ SYM sector and an edge connecting two circles to a $SU(2N)\times SU(N)$ bifundamental hypermultiplet.}
\end{figure}

%
%
%

\subsection{Notations/definitions/conventions} 
We systematically abuse notation, and use the same symbol $G$ to denote a Dynkin graph, the corresponding Lie algebra, and its Lie group. The same symbol also denotes a generic quiver obtained by orienting the graph $G$, except when we choose a reference orientation, in which case we write $\vec G$ for the chosen orientation.
If $\psi$ is an arrow in a quiver $Q$, we write $s(\psi)$ (resp.\! $t(\psi)$) for its source (resp.\! target) node.
$h(G)$ and $r(G)$ denote the Coxeter number and (respectively) the rank of $G$.

\begin{ass} In this paper, \emph{all} categories are assumed to be $\C$-linear with finite-dimensional Hom/Ext spaces and split idempotents (hence Krull-Schimdt).
In particular, all algebras and modules are finite-dimensional and defined over $\C$. \end{ass}

The identity functor is written $\mathrm{Id}$. In all triangle categories, we write $[1]$ for the suspension functor and $[m]$ for its $m$-fold iteration. We write $\mathsf{vect}$ for the semi-simple category of finite-dimensional vector $\C$-spaces and linear maps.

\begin{defn} For the benefit of the reader, we recall here some standard definitions we shall use throughly:
\begin{itemize}
\item[a)] If $T$ is an object (or, more generally, a class of objects) in a category $\cx$ (satisfying our general assumptions) we write $\mathsf{add}\,T$ for its additive closure, i.e.\! the full subcategory of $\cx$ formed by direct summands of direct sums of copies of $T$. 
\item[b)] If $\cx$ is an Abelian (resp.\! a triangle) category, its Grothendieck group $K_0(\cx)$ is the free Abelian group over the isoclasses $[X]$ of its objects $X$ modulo the relations $[Z]=[X]+[Y]$ whenever $X\to Z\to Y$ is a short exact sequence (resp.\! a distinguished triangle). 
\item[c)] An Abelian category $\ch$ is called \emph{hereditary} iff $\mathrm{Ext}^k(X,Y)=0$ for all $k\geq 2$ and all objects $X$, $Y\in\ch$. 
\item[d)] A triangle category is \textit{$n$-periodic} ($n\in\mathbb{N}$) if, for all objects $X$, $X[n]\simeq X$.
\item[e)] A triangulated category $\ct$
 has \emph{Serre duality} if there is an exact functor $S\colon \ct\to\ct$
such that we have the bi-functorial isomorphism
\be\label{Seerre}
\mathrm{Hom}_\ct(X,Y)\simeq D\,\mathrm{Hom}_\ct(Y, SX), \qquad \forall\; X,Y \in \ct.
\ee
where $D(-)\equiv\mathrm{Hom}_\mathsf{vect}(-,\C)$ is the standard duality in the category of $\C$-spaces.
When it exists, $S$ is unique (up to natural isomorphism).
\item[f)]
 A triangulated category $\ct$ with Serre functor $S$ is said to be
 \emph{$n$-Calabi-Yau} ($n\in\mathbb{N})$
 iff $S\simeq [n]$, that is,
 \be
 \mathrm{Hom}_\ct(X,Y)\simeq D\,\mathrm{Hom}_\ct(Y,X[n]).
 \ee 
 \item[g)] an object $X\in\cx$ is a \emph{brick} iff $\mathrm{End}_\cx X=\C$. 
 \end{itemize}\end{defn}
  
%
%

\paragraph{Organization of the paper.} The rest of this paper is organized as follows. In section 2 we discuss the relation between cluster categories and $S$-duality groups. In section 3 we describe the 4d/2d correspondence of \cite{Cecotti:2010fi}  in the categorical language, introduce the root category of the 2d model, and relate the auto-equivalence group of the 2d root category to the auto-equivalence group of the 4d cluster category. In section 4 we discuss the auto-equivalences of triangle categories, introducing the Thomas-Seidel twists, the telescopic functors, and their braid relations. In section 5 we study two interesting classes of models whose $S$-duality groups have simple descriptions, namely the $(G,G^\prime)$ SCFTs \cite{Cecotti:2010fi}  and the $(G,\widehat{H})$ QFTs
\cite{Cecotti:2012jx,Cecotti:2013lda}. In section 6 we discuss the relation of the duality group of the fully interacting theory with the ones for its decoupled constituents.
In section 7 we introduce a more general framework which allows to study the homological $S$-duality of  the DZVX models. In section 8 we give  some additional
details on special models.
Some side material is presented in the appendices.
   

\section{Cluster categories and
$S$-duality groups}

\subsection{BPS objects vs.\! cluster categories}\label{bpsvscluster}

The BPS objects of a 4d $\cn=2$ QFT $\cf$ are encoded in its cluster-category
$\mathscr{C}(\cf)$. The cluster-categorical framework automatically incorporates the Kontsevich-Soibelman
wall-crossing formula \cite{kswcf}, and hence is the right language to formulate the BPS spectrum problem in an intrinsic and global way over the full deformation space of the theory $\cf$, that is, independently of the duality-frame and the particular BPS chamber. 
To fix the ideas, we review how the connection
\be
\mathscr{C}(\cf)\longleftrightarrow \text{(the BPS sector of $\cf$)}
\ee
works and, in particular, how we read the BPS spectrum along the Coulomb branch from $\mathscr{C}(\cf)$, referring to the existing literature \cite{Cecotti:2010fi,Alim:2011kw,bonn} for further details.
\medskip

In general terms, a cluster category is a gadget of the following form:  

 \begin{defn}\label{whatclucat} (see e.g.\! \cite{palu1})
 A triangle category\footnote{\ Recall that in this paper ``category'' always stands for ``$\C$-linear category with finite-dimensional Hom spaces and split idempotents''. } $\mathscr{C}$ is called a \emph{cluster category} iff it is 2-Calabi-Yau, and admits a \emph{cluster-tilting} object $\ct$, that is, an object such that:
 \begin{itemize}
 \item[\textit{i)}]$\mathrm{Hom}_\mathscr{C}(\ct,\ct[1])=0$,
  \item[\textit{ii)}] $\forall\,X\in\mathscr{C}$, $\mathrm{Hom}(X,\ct[1])=0$ $\Rightarrow$ $X\in\mathsf{add}\,\ct$.
  \end{itemize}
 \end{defn}

 \begin{fact} [\!\!\cite{clustercat,palu1}] If $\mathscr{C}$ is a cluster category and $\ct$ a cluster-tilting object, we have an equivalence of categories
 \be\label{uttrsd}
 \boldsymbol{J}_\ct\colon \mathscr{C}/\langle\mathsf{add}\,\ct[1]\rangle\xrightarrow{\;\sim\;} \mathsf{mod}\,\mathrm{End}_\mathscr{C}(\ct),\qquad 
 X\mapsto \mathrm{Hom}_\mathscr{C}(\ct,X).
 \ee
 where $\langle\mathsf{add}\,\ct[1]\rangle$ denotes the ideal of morphisms which factor through objects in $\mathsf{add}\,\ct[1]$. 
 \end{fact}
 
The module category $\mathsf{mod}\,\mathrm{End}_\mathscr{C}(\ct)$ comes with a skew-symmetric integral form
\cite{palu1}
   \begin{gather}
  \langle X,Y\rangle_D=
 -\langle Y, X\rangle_D\in\Z, \qquad X,Y\in\mathsf{mod}\,\mathrm{End}_\mathscr{C}(\ct),\\
 \langle X, Y\rangle_D\equiv\dim \mathrm{Hom}(X,Y)-\dim \mathrm{Ext}^1(X,Y)-
\dim \mathrm{Hom}(Y,X)+\dim \mathrm{Ext}^1(Y,X)
 \end{gather}
  which has the physical interpretation of the Dirac electro-magnetic pairing between the BPS objects corresponding to the two objects $X$, $Y$ \cite{Alim:2011kw}. 
The Dirac form $\langle-,-\rangle_D$ is well-defined on the Grothendieck group $K_0(\mathsf{mod}\,\,\mathrm{End}_\mathscr{C}(\ct))$ \cite{palu1}. Objects are \emph{mutually local} if their Dirac pairing is zero.
\medskip

To recover the BPS spectrum of $\cf$  from $\mathscr{C}(\cf)$, one works through the following steps.

\subparagraph{Step 1.} One chooses a tilting object $\ct$ which is appropriate for the physical regime we are interested in, and consider the algebra
$\cb_\ct\equiv \mathrm{End}_{\mathscr{C}(\cf)}\,\ct$ (which we assume to be finite-dimensional, this being guaranteed in the physical context). Let $S_i$, $i=1,2,\dots,r$ be the simple modules of $\cb_\ct$. Construct the 2-acyclic\footnote{\ A quiver is \emph{2-acyclic} iff there are no loops (arrows with starting and ending at the same node) nor pairs of opposite arrows $i\leftrightarrows j$. } quiver $Q_\ct$ whose nodes are in one-to-one correspondence with the $S_i$'s and connect nodes $i$ and $j$ by a signed\footnote{\ That is, a negative number of arrows means arrows in the opposite direction.} number of arrows from $i$ to $j$ equal to 
\be\label{lllashtb}
\#\text{(arrows from $i$ to $j$)}\equiv\langle S_i,S_j\rangle_D.
\ee
The algebra $\cb_\ct$ is then the path algebra $\C Q_\ct$ bounded by the ideal $(\partial\cw_\ct)$ generated by the cyclic derivatives of a certain non-degenerate superpotential\footnote{\ In this abstract sense, a superpotential $\cw_\ct$ is a  complex linear combination of closed oriented cycles in $Q_\ct$. In physical terms, this translates in the statement that $\cw_\ct$ is a gauge-invariant, single-trace, 1d chiral operator.} $\cw_\ct$
\cite{DWZ}. We call this algebra the \emph{Jacobian algebra}
$\mathsf{Jac}(Q_\ct,\cw_\ct)$ of the quiver with superpotential $(Q_\ct,\cw_\ct)$.  The Grothendieck group
$K_0(\mathsf{mod}\,\,\mathrm{End}_\mathscr{C}(\ct))\equiv\Gamma$ gets identified with the lattice generated by the isoclasses of the $S_i$'s, $\Gamma\equiv \bigoplus_{i=1}^r \Z [S_i]$. The \emph{positive cone} $\Gamma_+\subset \Gamma$ of actual modules
is $\Gamma_+=\bigoplus_{i=1}^r\Z_{\geq0}[S_i]$. If $X$ is a module, and
$[X]=\sum_{i=1}^r x_i[S_i]$, the non-negative integer $x_i$ is called the \textit{dimension of $X$ at the $i$-th node,}
and $\boldsymbol{x}\equiv(x_1,\dots,x_r)$ its \textit{dimension vector.} As we shall see momentarily, $\Gamma$ is the lattice of conserved charges\footnote{\ More precisely: the conserved charges of an effective IR description of $\cf$ which is specified by the choice of $\ct$). See discussion below.} of $\cf$, and we shall use the terms Grothendieck class, dimension vector, and charge vector interchangeably.

\subparagraph{Step 2.} One needs to introduce a further datum which specifies the values of the various parameters of the physical theory: 
couplings, masses, and the specific vacuum in the Coulomb branch we are considering. These data are encoded in the central charge of the $\cn=2$ \textsc{susy} algebra, $Z$, which is a linear combination of the IR internal charges with complex coefficients which depend on the physical parameters. 
In the categoric language $Z$ becomes the 
 stability function: a group homomorphism 
\be
Z\colon \Gamma\equiv K_0(\mathsf{mod}\,\mathsf{Jac}(Q_\ct,\cw_\ct))\to \C
\ee
 such that the positive cone $\Gamma_+\subset \Gamma$ is mapped in the upper half-plane \cite{stab}. Then we have a well-defined map $\arg Z\colon\Gamma_+\to [0,\pi)$. A Jacobian module
$X\in\mathsf{mod}\,\mathsf{Jac}(Q_\ct,\cw_\ct)$ is \emph{stable} iff for all proper submodules $Y$ one has
\be\label{ssstableee}
\arg Z(Y)<\arg Z(X).
\ee
We note that a stable object is necessarily a brick.

\begin{rem}\label{ooooxb} Let $\ct=\bigoplus_{i=1}^r T_i$ with $T_i$ indecomposable and pairwise non-isomorphic. Since $\Gamma\otimes \C\simeq \bigoplus_{i=1}^r \C [T_i]$, 
to specify $Z(-)$ it is enough to give the $r$ complex numbers $Z([T_i])$.\end{rem}

\subparagraph{Step 3.} The data $(\mathscr{C}(\cf),\ct,Z)$
 define a 1d quantum system with 4 supercharges\footnote{\ We may think of this class of 1d \textsc{susy} theories as obtained by dimensional reduction of 4d $\cn=1$ \textsc{susy}. We use the 4d language in describing the various couplings.} for each
$\boldsymbol{x}\in \Gamma_+$ \cite{Alim:2011kw,halo}, namely the 1d quiver gauge theory over $Q_\ct$ having gauge group
$\prod_{i=1}^r U(x_i)$, with one bi-fundamental Higgs chiral superfield $\phi_\alpha$ in the representation $(\boldsymbol{\overline{x}_i},\boldsymbol{x}_j)_{-1,1}$ of $U(x_i)\times U(x_j)$ per arrow $i\xrightarrow{\alpha} j$ of $Q_\ct$, and the gauge invariant superpotential $\cw_\ct$. $Z$ specifies all other couplings, such as the FI terms of the Abelian gauge groups \cite{halo}. This 1d system is physically interpreted as the world-line description of a 4d particle.
A state of the particle is BPS in the 4d sense iff it preserves 4 supercharges;
from the 1d perspective, a state preserving 4 supersymmetries is a \textsc{susy} vacuum. Hence we have one 4d BPS ultra-short supermultiplet with charges $\boldsymbol{x}$ per vacuum of the corresponding 1d system. To get the 1d quantum vacua one proceeds through two steps: first determines the manifold $\cm_{\boldsymbol{x}}$ of the classical vacua, and then quantizes the 1d \text{susy} $\sigma$-model with target space $\cm_{\boldsymbol{x}}$. 
Classical vacua are determined by solving the $F$-condition, $\partial \cw_\ct=0$, and the $D$-condition
\cite{Alim:2011kw,king}\footnote{\ In eqn.\eqref{uuuujjkka} $e_i\in \mathsf{Jac}(Q_\ct,\cw_\ct)$ is the minimal idempotent associated to the $i$--th node of $Q_\ct$.}
\be\label{uuuujjkka}
\sum_{t(\alpha)=i}\phi_\alpha\phi^\ast_\alpha-\sum_{s(\alpha)=i}\phi_\alpha^\ast\phi_\alpha =\theta_i\cdot e_i\quad \text{for all }i\in Q_\ct,
\ee
(with $\theta_i$ the FI term of the $i$--th $U(1)$ gauge factor).
A Higg field configuration (modulo gauge transformation) is then identified with an isoclass of representations of the quiver $Q_\ct$ which is actually a Jacobian module by the $F$-condition.
One shows that the $D$-condition is equivalent to the requirement that the module is stable in the sense of 
\eqref{ssstableee}. 

Had we started with a different cluster-tilting object $\ct^\prime\not\simeq \ct$,
we would get a different 1d model.
The two quantum models are however equivalent under 1d Seiberg duality \cite{seidu}. Therefore, a choice of $\ct$ is just a choice of 1d Seiberg duality-frame. This is the precise sense in which physics is independent of the choice of $\ct$. At the level of quivers with superpotential, Seiberg duality is mutation in the sense of DWZ \cite{DWZ}.

\subparagraph{Step 4.}
The stable objects of $\mathsf{mod}\,\mathsf{Jac}(Q_\ct,\cw_\ct)$
organize themselves into continuous families $\{\co_\lambda\}_{\lambda\in \cm_{\boldsymbol{x}}}$ (with $\co_{\lambda}\not\simeq \co_{\lambda^\prime}$ for $\lambda\neq\lambda^\prime$),  parametrized by irreducible complex
projective varieties $\cm_{\boldsymbol{x}}$ (so $\cm_{\boldsymbol{x}}$ are, in particular, compact K\"ahler). An object $\co$ is said to be \emph{rigid} iff its index variety $\cm_{\boldsymbol{x}}$ reduces to a point. 
\textbf{Step 3} associates 
to each (generically) stable family $\{\co_\lambda\}_{\lambda\in \cm_{\boldsymbol{x}}}$ a quantum system, namely the 1d supersymmetric $\sigma$-model with target space $\cm_{\boldsymbol{x}}$, which is an irreducible component of the space of classical \textsc{susy} vacua of the 1d quiver gauge theory. The supersymmetric vacua of this 1d $\sigma$-model then correspond to quantum states of a 4d BPS particle with charges $\boldsymbol{x}$. 
It is well-known \cite{Witten:1982df} that, for a 1d $\sigma$-model, the vector space of vacua is isomorphic to the
Dolbeault cohomology $H^{\bullet\bullet}(\cm_{\boldsymbol{x}})$ of $\cm_{\boldsymbol{x}}$. On $H^{\bullet\bullet}(\cm_{\boldsymbol{x}})$ there is a natural action of the 1d $R$-symmetry
\be\label{mmmkkklq}
U(1)_R\times SU(2)_L 
\ee
where the $SU(2)_L$ action is induced by the Lefshetz decomposition of harmonic forms \cite{GH}, while $U(1)_R$ acts as $e^{i\alpha(p-q)}$ on $(p,q)$ forms. Physically, $U(1)_R$ is interpreted as the torus of the (unbroken) 4d  $R$-symmetry group $SU(2)_R$, while $SU(2)_L$ is the Lorentz little group of a massive particle.
The space $H^{\bullet\bullet}(\cm_{\boldsymbol{x}})$ then decomposes into a set $I(\cm_{\boldsymbol{x}})$ of irreducible representations of $SU(2)_R\times SU(2)_L$ labelled by their respective dimensions $(\boldsymbol{r},\boldsymbol{s})\in \mathbb{N}^2$. A celebrated conjecture \cite{framed,DelZotto:2014bga} states that only the representations $(\boldsymbol{1},\boldsymbol{s})$ actually appear. In conclusion, the stable family $\{\co_\lambda\}_{\lambda\in \cm_{\boldsymbol{x}}}$ yields 4d BPS states with charge $\boldsymbol{x}$ in the following representation of $SU(2)_R\times SU(2)_L$
\be\label{ptqz}
\Big((\boldsymbol{2},\boldsymbol{1})\oplus 
(\boldsymbol{1},\boldsymbol{2})\Big)\otimes\bigoplus_{(\boldsymbol{r},\boldsymbol{s})\in I(\cm_{\boldsymbol{x}})} (\boldsymbol{r},\boldsymbol{s}).
\ee
To complete the BPS multiplet
one needs to add to \eqref{ptqz} the PCT conjugate states which correspond to the shifted family $\{\co_\lambda[1]\}_{\lambda\in\cm_{\boldsymbol{x}}}$, which obviously produce the same $SU(2)_R\times SU(2)_L$ content. 

\begin{rem} The square of PCT is the identity, and one would naively expect that the category describing the BPS objects of a QFT should be 2-periodic.
In general this is not the case:
the double shift $[2]$ needs only to be quasi-isomorphic to a functor $\mathbb{M}\colon \mathscr{C}(\cf)\to \mathscr{C}(\cf)$ called the
 \emph{quantum monodromy}
\cite{Cecotti:2010fi}. The Kontsevich-Soibelman wall-crossing formula \cite{kswcf} is equivalent to the requirement that $\mathbb{M}$ is well-defined (up to conjugacy) \cite{Cecotti:2010fi}. In a $\cn=2$ theory with a weakly-coupled Lagrangian formulation, $\mathbb{M}$ acts as the identity on the microscopic degrees of freedom, in agreement with the perturbative analysis. 
If $\cf$ is a SCFT with a weakly-coupled Lagrangian,
$\mathbb{M}\simeq \mathrm{Id}$ and the category
$\mathscr{C}(\cf)$ is 2-periodic as naively expected.  
 \end{rem}

From eqn.\eqref{ptqz} we see that the maximal spin produced by a family $\{\co_\lambda\}_{\lambda\in\cm_{\boldsymbol{x}}}$ is equal to 
\be\label{m-spin}
\text{max-spin}=\frac{1}{2}\big(1+\dim_\C \cm_{\boldsymbol{x}}\big).
\ee
The BPS states arising from a rigid stable object $X$ are said to form a \emph{hypermultiplet}, while a $\mathbb{P}^1$-family of stable objects
$\{W_\lambda\}_{\lambda\in\mathbb{P}^1}$ produces a \emph{vector} multiplet (maximal spin 1). 

Interaction vertices between BPS particles correspond to exact sequences of $\mathsf{mod}\,\mathsf{Jac}(Q_\ct,\cw_\ct)$ whose objects $\co^{(i)}_{\lambda_i}\in\mathsf{mod}\,\mathsf{Jac}(Q_\ct,\cw_\ct)$ ($i=1,2,3$) are all generically stable
\be\label{vertices}
0\to \co^{(1)}_{\lambda_1}\to \co^{(2)}_{\lambda_2}\to \co^{(3)}_{\lambda_3}\to 0\qquad\longleftrightarrow\qquad \begin{gathered}\xymatrix{&& \co^{(1)}_{\lambda_1}\\
\co^{(2)}_{\lambda_2}\ar@{~}[r]&\bullet\ar@{-}[ru]
\ar@{-}[rd]\\
&&\co^{(3)}_{\lambda_3}}\end{gathered}
\ee
From this identification it is obvious that the module category needs to be ``completed'' to a triangle category in order to implement crossing symmetry
(rotation of the triangles).
 From \eqref{vertices} it is also clear that all conserved quantities factorize through the Grothendieck group $\Gamma\equiv \mathsf{mod}\,\mathsf{Jac}(Q_\ct,\cw_\ct)$, which is then the \emph{universal} group of conserved (additive) charges. Its rank $r$ is equal to the total number of the electric, magnetic, and flavor charges of the $\cn=2$ QFT, that is,
 \be
 \begin{gathered}
 r=2\dim_\C\text{(Coulomb branch)}+\mathrm{rank}\, F,\\
 g\equiv \dim_\C\text{(Coulomb branch)}=\frac{1}{2}\mathrm{rank}\,\langle -,-\rangle_D.
 \end{gathered}
 \ee 
Here $F$ is the flavor group (a compact Lie group). The sub-lattice of the flavor charges, $\Gamma_\text{flavor}\subset \Gamma$ is the \emph{radical} of the Dirac form  
\be
\Gamma_\text{flavor}=\Big\{\boldsymbol{x}\in \Gamma\;\Big|\;
\langle \boldsymbol{y},\boldsymbol{x}\rangle_D=0\ \ \forall\,\boldsymbol{y}\in \Gamma\Big\}.
\ee
The electro-magnetic lattice $\Gamma_\text{e.m.}$ is defined as
\be
\Gamma_\text{e.m.}=\Gamma\big/\Gamma_\text{flavor}.
\ee
The Dirac pairing $\langle -,-\rangle_D$ induces a \emph{non-degenerate} skew-symmetric, integral pairing in
$\Gamma_\text{e.m.}$ denoted by the same symbol. 

\begin{rem} We stress that the lattice 
$\Gamma$ is the Grothendieck group $K_0(\mathsf{mod}\,\mathsf{Jac}(Q_\ct,\cw_\ct))$ which is not isomorphic (in general) to the cluster category Grothendieck group $K_0(\mathscr{C})$. See \cite{kusss} for examples of 
$K_0(\mathscr{C})$ groups.
The physical meaning of the relation between these Abelian groups is as follows:
the cluster category $\mathscr{C}$ describes the UV microscopic theory in all its  physical regimes and phases.
For a generic Coulomb regime, we may
describe the theory in the IR   as an effective Abelian gauge theory with conserved electric, magnetic, and flavor charges.
Such an IR regime selects (non uniquely in general) a pair $(\ct,Z)$, leading to the category
$\mathsf{mod}\,\mathsf{Jac}(Q_\ct,\cw_\ct)$. The equivalence \eqref{uttrsd} should be though of as the ``dictionary'' between the UV and IR viewpoints. From the UV viewpoint, however, the gauge group is non-Abelian, and the electric-magnetic charges do not take value in a lattice. We shall discuss this issue in more detail elsewhere \cite{toappear}. For the present purposes it suffices to remark that $\Gamma$ is not an intrinsic property of the UV category $\mathscr{C}$. 
\end{rem}

\subsection{$\mathrm{Aut}(\mathscr{C}(\cf))$ and the $S$-duality group $\mathbb{S}(\cf)$}

We write $\mathrm{Aut}(\mathscr{C}(\cf))$ for the group of the triangle auto-equivalences of the cluster category $\mathscr{C}(\cf)$. Let $\mu\in\mathrm{Aut}(\mathscr{C}(\cf))$. If $\ct$ is a cluster-tilting object,  so is $\mu \ct$. If the family of objects $\{\co_\lambda\}_{\lambda\in\cm_{\boldsymbol{x}}}\subset \mathscr{C}(\cf)$
had the property that its Jacobian images $\{\boldsymbol{J}_\ct\,\co_\lambda\}_{\lambda\in\cm_{\boldsymbol{x}}}$ are stable for the given central charge $Z(-)$,
the family $\{\mu \co_\lambda\}_{\lambda\in\cm_{\boldsymbol{x}}}$ has Jacobian images $\{\boldsymbol{J}_{\mu \ct}(\mu \co_\lambda)\}_{\lambda\in\cm_{\boldsymbol{x}}}$ which are stable for the pushed-forward central charge
$\mu_\ast Z(-)$ defined by
(cfr.\! \textbf{Remark \ref{ooooxb}})
 \be
\mu_\ast Z([\mu T_i])=Z([T_i]).
\ee
The new stable family has the same index variety $\cm_{\boldsymbol{x}}$ as the old one, and hence the same 1d $\sigma$-model in \textbf{Step 4} of the procedure in \S.\,\ref{bpsvscluster}.
Therefore, the family $\{\mu \co_\lambda\}_{\lambda\in\cm_{\boldsymbol{x}}}$ produces the same $SU(2)_R\times SU_L(2)$ representation content \eqref{ptqz} as the original family.
Since the pair $(\ct,Z)$ encodes the physical regime of the QFT $\cf$, we learn that the physics looks identical in the original point in parameter space, $(\ct,Z)$, and in the pushed-forward one, $(\mu \ct,\mu_\ast Z)$.

The statement that the (BPS) physical observables of a QFT $\cf$ look identical in two distinct regimes specified by $(\ct,Z)$ and $(\mu \ct,\mu_\ast Z)$  (e.g.\! at strong and weak coupling) is what we mean by a (generalized) $S$-duality.
We see that all $\mu\in\mathrm{Aut}(\mathscr{C}(\cf))$ produce a duality in this broad sense.

However not all elements of $\mu\in\mathrm{Aut}(\mathscr{C}(\cf))$ induce non-trivial dualities.

\begin{defn}
An element $\sigma\in\mathrm{Aut}(\mathscr{C}(\cf))$ is said to be \emph{physically trivial} if for all families of objects
$\{X_\lambda\}_{\lambda\in\cm_X}$,
 there exist maps $f_X\colon \cm_X\to \cm_X$ such that
 \be
 \sigma X_\lambda\simeq X_{f_X(\lambda)},
 \ee
that is, the only effect of $\sigma$ is to re-parametrize the moduli varieties $\cm_X$. The re-parametrization has a trivial effect on the 1d $\sigma$-models on the particle world-lines, and the physics remains unchanged.
We write $\mathrm{Aut}(\mathscr{C}(f))^0\subset \mathrm{Aut}(\mathscr{C}(f))$ for the normal subgroup of physically trivial auto-equivalences.
\end{defn}

Let us describe the trivial subgroup
$\mathrm{Aut}(\mathscr{C}(f))^0$ a bit more explicitly. Let $\ct=\bigoplus_i T_i$ a cluster-tilting object of $\mathscr{C}(\cf)$. An element $\mu\in \mathrm{Aut}(\mathscr{C}(f))^0$ fixes all $T_i$ since they are rigid, and induces linear maps
\be\label{arrqaun}
\mu_{ij}\colon\mathrm{Hom}_{\mathscr{C}(\cf)}(T_i,T_j)\to \mathrm{Hom}_{\mathscr{C}(\cf)}(T_i,T_j).
\ee
At the level of the endo-quiver of $\ct$, $\mu_{ij}$ yields a family of linear redefinitions of the arrows between nodes $i$ and $j$ which leave the endo-algebra invariant. The action of $\mu$ on a generic representation $X$ of the Jacobian quiver then produces a representation $\mu X$ with the same vector spaces at the nodes and arrows redefined by the above linear maps $\mu_{ij}$.

\begin{exe} Let $\cf$ be pure $SU(2)$ SYM. The endo-quiver of the canonical tilting object is the Kronecker quiver $1\rightrightarrows 2$. Then $\mathrm{Aut}(\mathscr{C}(f))^0=PSL(2,\C)$.
\end{exe}

\begin{defn}
The quotient group
\be
\mathbb{S}(\cf)\equiv \mathrm{Aut}(\mathscr{C}(\cf))/\mathrm{Aut}(\mathscr{C}(\cf))^0,
\ee
is the \textit{$S$-duality group} of the QFT $\cf$.
\end{defn}

On general physical grounds \cite{zumn}\footnote{\ For a review of electro-magnetic dualities in Lagrangian field theory, see \S.1.4 of the book \cite{mybook}.} the group of dualities of any 4d $\cn=2$ theory should act on the electro-magnetic charges by an arithmetic  subgroup\footnote{\ As before, $g$ is the complex dimension of the Coulomb branch.} 
\be\cg_\Z\subset Sp(2g,\R).\ee 
The $S$-duality $\mathbb{S}(\cf)$ is slightly more general than plain electro-magnetic duality, since it  acts non-trivially on the flavor charges, as it happens in $SU(2)$ $N_f=4$ SQCD, where the action is through $Spin(8)$ triality
\cite{SW2}. 
 $\mathbb{S}(\cf)$ needs only to be \emph{commensurable}\footnote{\ We say that two discrete groups $\cg_1$, $\cg_2$ are \emph{commensurable} if they have finite-indices subgroup $\ch_a\subset \cg_a$ ($a=1,2$) with $\ch_1$, $\ch_2$ isomorphic. Commensurability
 equivalence will be written as $\thickapprox$.} to a discrete group of the form $\cg_\Z$,
 \be\label{physexp}
 \mathbb{S}(\cf)\thickapprox \cg_\Z.
 \ee
At this point property \eqref{physexp} is far from obvious. Indeed, the physical motivation for eqn.\eqref{physexp} is the preservation of the Dirac pairing $\langle -,-\rangle_D$ under duality;  
but 
the Dirac pairing is
defined on $\Gamma\equiv K_0(\mathsf{mod}\,\mathsf{End}_{\mathscr{C}(\cf)}(\ct))$
rather than on $\mathscr{C}(\cf)$ itself.
In section \ref{spsp2g}
we shall show that property
\eqref{physexp} holds (in the  appropriate sense) at least for all models mentioned in the \textbf{Introduction}, those listed there under letters a), b), c), and d) as well as the ones which may be reduced to these cases. The group $\cg_\Z$ may be explicitly realized as a  concrete group of symplectic integral $2g\times 2g$ matrices, and the full $\mathbb{S}(\cf)$ group as an explicit group of
$r\times r$ integral unimodular matrices ($r=2g+\mathrm{rank}\,F$).

To get the explicit realization of the $S$-duality group we introduce an \emph{auxiliary} triangle category $\mathscr{R}(\cf)$ (the root category of $\cf$), which has essentially the same auto-equivalence group as the cluster category $\mathscr{C}(\cf)$, but is more amenable for direct computation of this group. The root category will be introduced \emph{via} the 4d/2d correspondence \cite{Cecotti:2010fi} that we now review in a language suited for our present purposes.

\section{Categorical 4d/2d correspondence
and branes}\label{4d2d3ct}

In ref.\!\cite{Cecotti:2010fi} the relation between cluster algebras and BPS objects of a $\cn=2$ QFT was also explained in terms of a correspondence between the 2d (2,2) theories $\tilde\cf$ with $\hat c< 2$ and certain 4d $\cn=2$ models $\cf$.
The cluster category $\mathscr{C}(\cf)$ of a 4d theory $\cf$ with a ``nice'' 2d correspondent $\tilde\cf$ is easier to visualize, and the analysis of its $S$-duality becomes simpler.

We refer to \cite{Cecotti:2010fi} for the physical motivations behind the 4d/2d correspondence. 

The discussion in this section is adequate (at least) for all theories in classes a), b), c) and d) listed in the \textbf{Introduction}, and, more generally, for all models which have a ``good'' 4d/2d correspondence.

\subsection{The correspondence  $\tilde\cf\rightarrow \cf$}\label{lllksdv}
In the categorical language, the relation $\tilde\cf\rightarrow \cf$ may be (roughly) summarized as follows.
There is an Abelian category $\mathscr{A}(\tilde\cf)$, of global dimension $n\leq 2$ with tilting object (in the sense of \cite{miya}) $\ct=\bigoplus_{i=1}^rT_i$, whose bounded derived category $\mathscr{B}(\tilde\cf)\equiv D^b\mathscr{A}(\tilde\cf)$ has Serre functor $S$, and is the \emph{brane category} \cite{orlov1,orlov2,orlov3} of the 2d (2,2) theory $\tilde\cf$.
The endo-quivers of $\ct$ are interpreted as the BPS quivers of the 2d theory $\tilde\cf$, see \cite{Cecotti:2010fi}.

Consider the orbit category
\be\label{ororbt}
\mathscr{B}(\tilde \cf)/\langle S^{-1}[2]\rangle^\Z,
\ee
whose objects are the branes in $\mathscr{B}(\tilde\cf)$ and morphism spaces
\be
\mathrm{Hom}_\text{orb}(X,Y)=\bigoplus_{k\in\Z} \mathrm{Hom}_{\mathscr{B}(\tilde\cf)}(X, S^k Y[-2k]).
\ee
The orbit category is 2-Calabi-Yau by construction, and has tilting objects inherited from $\mathscr{A}(\tilde F)$, but it is not a cluster category (in general) because it is not triangulated.
However, there is a canonical ``smallest'' triangle category which contains the orbit category, its \textit{triangular hull} \cite{kellerorb,kellerrev}. We shall write $\mathsf{Hu}_\triangle$ for the operation of completing an orbit category to its triangular hull. The embedding of the orbit category in its triangular hull  
\be
\mathscr{B}(\tilde \cf)/\langle S^{-1}[2]\rangle^\Z \to \mathsf{Hu}_\triangle\!\!\left(\mathscr{B}(\tilde \cf)/\langle S^{-1}[2]\rangle^\Z\right)
\ee
is fully faithful \cite{kellerorb,kellerrev}.  The category in the \textsc{rhs} is now triangulated, $2$-Calabi-Yau, and has a tilting objet, so, according to \textbf{Definition \ref{whatclucat}}, it is a cluster category.
We say that $\cf$ is the 4d correspondent of the 2d theory $\tilde \cf$ iff the cluster category constructed in this way out of $D^b\mathscr{A}(\tilde\cf)$ is the cluster category of $\cf$, that is, iff
\be\label{jjjjasd}
\mathscr{C}(\cf)\simeq\mathsf{Hu}_\triangle\!\left(D^b\mathscr{A}(\tilde\cf)/\langle S^{-1}[2]\rangle^\Z\right).
\ee 
From this equivalence we deduce the relation between the 4d Jacobian quiver and the 2d BPS quiver \cite{Cecotti:2010fi}. For $\mathscr{A}(\tilde\cf)$ hereditary these two quivers are just equal, otherwise the Jacobian 4d quiver is the ``completion'' of the 2d one. It was this graphical connection which suggested the 4d/2d correspondence in the first place \cite{Cecotti:2010fi}.
\medskip

An important property of the triangular hull \eqref{jjjjasd} is the following:

\begin{fact}[\!\!\cite{imaorbit} \textbf{Theorem 5.4}]\label{norigg} 
The objects $X$ of the triangular hull which do not belong to the orbit category appear in families of dimension at least 1 and are never rigid, $\mathrm{Ext}^1(X,X)\neq0$.
\end{fact}
Comparing with eqn.\eqref{m-spin}, we see that the ``extra'' BPS particles 
 arising from the triangle completion of an orbit category have higher spin and are never hypermultiplets.
 \medskip
 
 In order to understand the properties of the ``good'' brane categories $\mathscr{B}(\tilde \cf)$, we introduce a definition.
 
\begin{defn} 1) An object $X$ in a triangle category $\mathscr{T}$, with Serre functor $S$, is said to be \textit{fractional Calabi-Yau} if there are positive integers $a$, $b$ such that
\be\label{cycycyc}
S^bX=X[a].
\ee
2) A triangle category with Serre functor $S$ is said to have \textit{fractional Calabi-Yau dimension $\tfrac{a}{b}$} if $a$, $b$ are the smallest positive integers such that \eqref{cycycyc} holds for all objects.
\end{defn}

If we have a triangle category $\mathscr{T}$ with Serre functor $S$, we write $\mathsf{CY}(\mathscr{T})$ for the additive closure in $\mathscr{T}$ of the class of its fractional Calabi-Yau objects. $\mathscr{T}$ is fractionally Calabi-Yau iff $\mathscr{T}=\mathsf{CY}(\mathscr{T})$. 
If the brane category $\mathscr{B}(\tilde\cf)$ has fractional Calabi-Yau dimension $\tfrac{a}{b}$ we say that the 2d theory $\tilde\cf$ has a conformal UV fixed point with Virasoro central charge $\hat c=\tfrac{a}{b}$.
Otherwise we say that $\tilde \cf$ is asymptotically-free (logarithmic violation of scaling \cite{Cecotti:1992rm}).

\begin{rem} We note that  we are not allowed to simplify the common factors of the integers $a$, $b$ in the fractional Calabi-Yau dimension $\tfrac{a}{b}$. For instance, the $A_{2n-1}$ (2,2) minimal model has fractional CY dimension $\tfrac{2(n-1)}{2n}\neq \tfrac{n-1}{n}$. Indeed, the statement that the Calabi-Yau dimension is $\tfrac{a}{b}$ corresponds to the following \emph{two} properties of the UV fixed point of the 2d QFT
$\tilde\cf$
\be
\text{fractional CY dimension}=\frac{a}{b}\quad\Leftrightarrow\quad
\begin{cases}
\hat c=\frac{a}{b}\ \text{in }\mathbb{Q}\\
\text{for chiral primaries }h=\frac{k}{b}\ \text{with }k\in\mathbb{N}.
\end{cases}
\ee
\end{rem} 

For the 2d theory which have a 4d correspondent, the fractional Calabi-Yau dimension of all CY objects is bounded by 2
\be
\frac{a}{b}<2.
\ee
Rougly speaking, Calabi-Yau objects in a brane category correspond to operators which are protected with respect to mixing under RG flow. In all ``good'' 2d theory, we have at least one such operator, namely the identity.

Thus we shall assume that our brane category
$\mathscr{B}(\tilde\cf)$, has the following property:

\begin{ass}\label{mmmmhas}
The 2d brane triangle category $\mathscr{B}(\tilde\cf)\equiv D^b\mathscr{A}(\tilde\cf)$,
in addition to tilting objects and Serre duality, has some CY object with $\tfrac{a}{b}<2$.
\end{ass}

Of course, this condition is satisfied by all models listed in the \textbf{Introduction}.
It is shown in \cite{kellerorb} that, under this additional assumption, the cluster category $\mathsf{Hu}_\triangle\!\!\left(\mathscr{B}(\tilde \cf)/\langle S^{-1}[2]\rangle^\Z\right)$ is equal to the orbit category if and only if $\mathscr{A}(\tilde\cf)$ is derived equivalent to a hereditary category; comparing with ref.\!\cite{Cecotti:2011rv}, we see that this happen only when the 2d QFT
$\tilde\cf$ is either minimal or affine or elliptic 
(in particular,  this requires $\tfrac{a}{b}\leq 1$).

\subsection{The root category $\mathscr{R}(\tilde\cf)$}\label{rrrorott}

If $X\in\mathscr{B}(\tilde\cf)$ is a brane, its anti-brane is $X[1]$.
Since we think of the anti-brane of the anti-brane as the original brane, it seems more appropriate to replace $\mathscr{B}(\tilde\cf)$ with the orbit category
\be
\mathscr{B}(\tilde\cf)\big/[2\Z],
\ee
in order to implement the equivalence
$[2]\sim \mathrm{Id}$. However, again, the orbit category is not triangulated in general, and we have to take its triangular hull
\be
\mathscr{R}(\tilde\cf)\simeq\mathsf{Hu}_\triangle\!\!\left(D^b\mathscr{A}(\tilde\cf)/[2\Z]\right).
\ee

We call $\mathscr{R}(\tilde\cf)$ the \emph{root category} since
for $\mathscr{A}(\tilde\cf)$ the module category of a Dynkin algebra we recover the
 Peng-Xiao root category \cite{pengxiao} of the associated Dynkin graph.  
 
 By construction, $\mathscr{R}(\tilde\cf)$ is triangulated, 2-periodic, and has tilting object $\ct$. We write $S$ for the image of the Serre functor in the root category.
 
\subsection{The 2d and 4d quantum monodromies} 

In the categorical framework, the quantum monodromy is just the Serre functor $S$.
However, $S$ describes different monodromies in 2d and 4d. For the 2d model, $S$ acts on the root category $\mathscr{R}(\tilde\cf)$, while for the 4d one $S$ acts on the cluster category $\mathscr{C}(\cf)$.
If the 2d theory flows in the UV to a non-degenerate SCFT, that is, if $\mathscr{B}(\tilde\cf)$ has fractional Calabi-Yau dimension $\tfrac{a}{b}$, both quantum monodromies have finite order.

For the 2d monodromy $H$ of a UV SCFT, we have two possibilities: either $a$ is even or it is odd.
For $a$ even $S^b\sim [a]\sim \mathrm{Id}$ in $\mathscr{R}(\tilde \cf)$
and the order is $b$. Otherwise the order of $H$ is $2b$. Then
\be
\text{order }H=\frac{2b}{\gcd(a,2)}.
\ee

The 4d monodromy $\mathbb{M}$ for $a$ even satisfies
\be
S^{b}\sim [a]\sim S^{a/2},
\ee
so that\footnote{\ We used that $\hat c=a/b<2$.}
\be
\text{order }\mathbb{M}\equiv \mathfrak{m}=\frac{2b-a}{\gcd(a,2)}>0.
\ee

Now $(S[-2])^b=[a-2b]$ so that
we have the intersection of groups of auto-equivalences
\be\label{ttttiixvv}
\big\langle S[-2]\big\rangle^\Z\bigcap\, \big[2\Z\big]=\big[2 \mathfrak{m}\Z\big],\quad\text{where }0<2\mathfrak{m}=\begin{cases} 2b-a &a\ \text{even}\\
4b-2a &a\ \text{odd.} \end{cases}
\ee
We define the auxiliary category
\be
\mathscr{CR}(\tilde\cf)=\mathsf{Hu}_\triangle\!\!\left(D^b\mathscr{A}(\tilde\cf)/[2\mathfrak{m}\Z]\right),
\ee
which is fractional Calabi-Yau of dimension $\langle a\rangle_\mathfrak{m}/b$,
where $\langle a\rangle_\mathfrak{m}$ stands for the smallest non-negative integer congruent to $a$ mod $2\mathfrak{m}$.

\subsection{Comparing $\mathrm{Aut}(\mathscr{R}(\tilde\cf))$ and $\mathrm{Aut}(\mathscr{C}(\cf))$}\label{comppparing}

For simplicity, we assume the 2d theory $\tilde\cf$ to have a non-degenerate UV fixed point with $\hat c=\tfrac{a}{b}<2$, so that $\mathscr{B}(\tilde\cf)\equiv D^b\mathscr{A}(\tilde\cf)$ is fractional Calabi-Yau of dimension $\tfrac{a}{b}$. 
We consider the Abelian subgroup $\mathfrak{A}\subset\mathrm{Aut}(\mathscr{B}(\tilde\cf))$ generated by $T\equiv S[-2]$ and $[1]$ i.e.
\be\label{qqqq1z}
\mathfrak{A}=\Big\langle T, [1]\;\Big|\; T^b=[a-2b]\Big\rangle=\Big\{ T^k[\ell],\ \ell\in\Z,\ k=0,1,\dots, b-1\Big\}.
\ee
Since $\mathscr{A}(\tilde\cf)$ is Hom-finite with finite global dimension,
for all objects $X,Y\in\mathscr{B}(\tilde\cf)$
\be\label{lljazfrr}
\sum_{\rho\in\mathfrak{A}} \dim \mathrm{Hom}_{\mathscr{B}(\tilde\cf)}(X, \rho Y)<\infty.
\ee
We define the four subgroups of $\mathfrak{A}$: 
\be
\mathfrak{A}_\emptyset=1,\qquad\mathfrak{A}_c=T^\Z,\qquad\mathfrak{A}_r=[2\Z],\qquad
\mathfrak{A}_{cr}=[2\mathfrak{m}\Z]=\mathfrak{A}_c\cap\mathfrak{A}_{r}.
\ee

From eqn.\eqref{ttttiixvv} 
one infers the following diagram of functors
\be
\begin{gathered}\label{rrrqwwzaX}
\xymatrix{&& \mathscr{B}(\tilde\cf)\equiv \mathscr{B}(\tilde\cf)/\mathfrak{A}_\emptyset\ar[dll]_{\pi_r}\ar[d]^{\pi_{cr}}\ar[drr]^{\pi_c}\\
\mathscr{B}(\tilde\cf)/\mathfrak{A}_r\ar[d]_{\iota_r} &&\mathscr{B}(\tilde\cf)/\mathfrak{A}_{cr}\ar[d]_{\iota_{cr}}\ar[ll]_{\pi_1}\ar[rr]^{\pi_2}&& \mathscr{B}(\tilde\cf)/\mathfrak{A}_c\ar[d]_{\iota_c}\\
\mathscr{R}(\tilde\cf)\ar@/_2pc/@{<..>}[rrrr]^\thickapprox_\rho
&&\mathscr{CR}(\tilde\cf)\ar[ll]_{\varpi_1}\ar[rr]^{\varpi_2}&&\mathscr{C}(\cf)}\end{gathered}
\ee
where $\pi_\star$ are canonical projections and $\iota_\star$ fully faithful embeddings in the respective triangular hulls. The orbit categories
$\mathscr{B}(\tilde\cf)/\mathfrak{A}_\star$ (with $\star=\emptyset, c,r,cr$)
have the same objects as $\mathscr{B}(\tilde\cf)$ and morphism spaces
\be\label{jjjjkasqw0}
\mathrm{Hom}_{\mathscr{B}(\tilde\cf)/\mathfrak{A}_\star}\!(A,B)=\bigoplus_{\rho\in\mathfrak{A}_\star}\mathrm{Hom}_{\mathscr{B}(\tilde\cf)}(A,\rho B).
\ee
The images of the tilting summands, $\{T_i\}_{i=1}^r$ of $\mathscr{A}(\tilde\cf)$,
  generate homologically
  the various linear categories appearing in the diagram \eqref{rrrqwwzaX}.
Then an auto-equivalence of one of these categories is uniquely determined by its action on the corresponding  full subcategory image of 
$\mathsf{add}\,\ct$.

We write $\varpi_\star=\iota_\star\circ \pi_\star$ for $\star=r,c,cr$.
 $\varpi_\star$ send auto-equivalences into auto-equivalences
\begin{align}
\mathrm{Aut}\,{\varpi_r}\colon& \mathrm{Aut}(\mathscr{B}(\tilde\cf))\to \mathrm{Aut}(\mathscr{R}(\tilde\cf)),\\
\mathrm{Aut}\,{\varpi_c}\colon& \mathrm{Aut}(\mathscr{B}(\tilde\cf))\to \mathrm{Aut}(\mathscr{C}(\cf)).
\end{align}
Since all auto-equivalences commute with $T$ and $[1]$, $\mathfrak{A}$-orbits are sent into $\mathfrak{A}$-orbits. 
An auto-equivalence permutes the $\mathfrak{A}$-orbits $\{\co_\xi\}_{\xi\in\Xi}$ between themselves and acts on each orbit by a functor of the form $S^{k_\xi}[\ell_\xi]$ (the integers $k_\xi$, $\ell_\xi$ depending on the orbit $\xi$).

We claim that 
$\ker \mathrm{Aut}(\varpi_1)=\Z/\mathfrak{m}\Z$. Indeed, let $\ct=\oplus_i T_i$ be the tilting object of $\mathscr{A}(\tilde\cf)$. One has
\be
\dim \mathrm{Hom}_{\mathscr{CR}}(T_i,T_j[k])= c_{ij}\,\delta^{(2\mathfrak{m})}_{k,0}
\ee
where $\delta^{(p)}_{i,j}$ is the mod $p$ Kronecker delta:
\be\label{modelta}
\delta^{(p)}_{i,j}=\begin{cases}
1 & i-j=0\mod p\\
0 &\text{otherwise.}
\end{cases}
\ee
If $\mu\in\ker \mathrm{Aut}(\varpi_1)$, $\mu T_i=T_i[2k_i]$, for certain integers $k_i$. Then
\be
\begin{split}
c_{ij}\,\delta^{(2\mathfrak{m})}_{k,0}&=\dim \mathrm{Hom}_{\mathscr{CR}}(\mu T_i,\mu T_j[k])=\\
&=\dim\mathrm{Hom}_{\mathscr{CR}}(T_i,T_j[k+2k_j-2k_i])=c_{ij}\,\delta^{(2\mathfrak{m})}_{k,2k_i-2k_j},
\end{split} 
\ee
hence $k_i=k_j\equiv\kappa$ (mod. $m$) for all $i$,$j$ (since the category is connected) and $\mu=[2\kappa]$.
The same argument applied to $\mathrm{Aut}(\varpi_2)$ shows that $\ker \mathrm{Aut}(\varpi_2)$ is finite\footnote{\ In many models $\ker \mathrm{Aut}(\varpi_2)$ is as small as $(S[-2])^k$ where $k\in\Z/(2b/\gcd(a,2))\Z$.}. 
Hence
\begin{fact}\label{amam36t} If the 2d theory $\tilde\cf$ flows in the UV to a non-degenerate (2,2) SCFT, the groups $\mathrm{Aut}(\mathscr{CR}(\tilde\cf))$, $\mathrm{im}\,\mathrm{Aut}(\varpi_1)$, and $\mathrm{im}\,\mathrm{Aut}(\varpi_2)$,
 are \emph{commensurable.}
 Then so are $\mathrm{im}\,\mathrm{Aut}(\varpi_c)$ and $\mathrm{im}\,\mathrm{Aut}(\varpi_r)$.
\end{fact}

It remains to determine the images of $\mathrm{Aut}(\varpi_r)$, $\mathrm{Aut}(\varpi_c)$. We shall see in \S.\ref{ujjujuja} that the maps 
$\mathrm{Aut}(\varpi_1)$, $\mathrm{Aut}(\varpi_2)$ are onto at least for the subgroups
\be\label{mmmmxzzw}
\mathrm{Tel}(\mathscr{R}(\tilde\cf))\subseteq\mathrm{Aut}(\mathscr{R}(\tilde\cf)),\qquad \mathrm{Tel}(\mathscr{C}(\cf))\subseteq\mathrm{Aut}(\mathscr{C}(\cf))\ee
 which are generated by (generalized) Thomas-Seidel
twists, telescopic functors\footnote{\ See section \ref{autttoequiv}.}, ``obvious'' auto-equivalences of $\mathscr{A}(\tilde\cf)$, and $T$, $[1]$. The continuous deformation subgroups
$\mathrm{Tel}(\mathscr{R}(\tilde\cf))^0$, $\mathrm{Tel}(\mathscr{C}(\cf))^0$ are also set in correspondence by $\rho$. It is reasonable to expect that ``generically''  the sub-groups \eqref{mmmmxzzw} coincide with the full auto-equivalence groups; if this is the case, the dashed arrow in \eqref{rrrqwwzaX} is an equivalence modulo commensurability and
\be\label{ooopqw}
\mathbb{S}(\cf)\thickapprox \mathrm{Tel}(\mathscr{R}(\tilde\cf))\big/\mathrm{Tel}(\mathscr{R}(\tilde\cf))^0.
\ee  
Pragmatically, we shall study the \textsc{rhs} of this equation rather than $\mathbb{S}(\cf)$ itself. In doing this, we may sometimes detect only a proper subgroup of the actual $S$-duality group, but a very interesting one.
The \textsc{rhs} of eqn.\eqref{ooopqw} is easily written as an explicit group of integral $r\times r$ matrices ($r=2g+f$, where $g$ is the dimension of the Coulomb branch and $f=\mathrm{rank}\,F$).

\begin{rem}
In the above argument we have supposed, for simplicity, that the 2d theory is UV superconformal ($\equiv$ the corresponding brane category is fractional Calabi-Yau). The conclusion may be extended to the asymptotically-free case, as the explicit description of the duality groups will show (see \S.\ref{assysyfree}). \end{rem}

\subsection{The $S$-duality group as a concrete matrix group}\label{uuuuqmmm82}

The (pragmatic) identification of the $S$-duality group (up to commensurability) with the group of auto-equivalences of
the root category, modulo the ones acting trivially on the Grothendieck group, eqn.
\eqref{ooopqw}, allows to identify it with a concrete group of integral matrices, as we now explain.
\medskip 

In the set-up and with the \textbf{Assumptions} of the previous subsections, we have
\be
K_0(\mathscr{B}(\tilde\cf))\equiv K_0(\mathscr{A}(\tilde f))\simeq \bigoplus_{i=1}^r \Z [T_i]\equiv \Gamma.
\ee
Since $[X[1]]=-[X]$, the map
\be
[-]\colon \mathscr{B}(\tilde\cf)\to \Gamma, \qquad X\mapsto [X],
\ee
factorizes through the orbit category
\be
\mathscr{B}(\tilde\cf)\to \mathscr{B}(\tilde\cf)\big/[2\Z]\to K_0(\mathscr{B}(\tilde\cf)\big/[2\Z])\equiv\Gamma.
\ee
On $K_0(\mathscr{B}(\tilde \cf))$ there is an (integral) Euler bilinear form
\be
\chi([X],[Y])=\sum_{k\in\Z}(-1)^k\,\dim\mathrm{Hom}(X,Y[k]),
\ee
which induces a bilinear form on the
root category $\mathscr{R}(\tilde\cf)$.
Given $X$, $Y\in \mathscr{R}(\tilde\cf)$, we write their Grothendieck classes as
\be[X]=\sum_i x_i[T_i],\qquad 
[Y]=\sum_i y_i[T_i],
\ee
and the Euler pairing in the matrix form
\be
\chi([X],[Y])=x_i \, \chi([T]_i,[T]_j)\,y_j \equiv x_i \,\boldsymbol{E}_{ij}\, y_j.
\ee
The Euler matrix $\boldsymbol{E}$ is unimodular. From eqn.\eqref{Seerre} we see that
\be
\chi(X,SY)=\chi(Y,X).
\ee

We introduce the 2d quantum monodromy matrix\footnote{\ Minus the 2d quantum monodromy is called the Coxeter matrix $\boldsymbol{\Phi}$; however it seems more natural to work with $\boldsymbol{H}$.} $\boldsymbol{H}$ \cite{Cecotti:1992rm} as the image of the functor $S$ into the Grothendieck group
\be
[ST_i]=[T_j]\boldsymbol{H}_{ji}
\ee 
Then
\be
\boldsymbol{E}_{ij}= \chi(T_i,T_j)= \chi(T_j, ST_i)=\chi(T_j,T_k)\boldsymbol{H}_{ki}=\boldsymbol{E}_{jk}\boldsymbol{H}_{ki}
\ee
and
\be\label{2dmono}
\boldsymbol{H}= \boldsymbol{E}^{-1}\boldsymbol{E}^t.
\ee
$\boldsymbol{H}$ is an isometry of the Euler form $\boldsymbol{E}$
\be
\boldsymbol{H}^t \boldsymbol{E}\boldsymbol{H}= \boldsymbol{E},
\ee
and hence leaves invariant its symmetric and anti-symmetric parts
\be\label{formsforms}
C=\boldsymbol{E}+\boldsymbol{E}^t,\qquad \omega=\boldsymbol{E}-\boldsymbol{E}^t.
\ee

The statement of the 4d/2d correspondence \cite{Cecotti:2010fi} implies  (whenever a ``good'' 2d correspondent exists) that the Dirac pairing  is given by
the skew-symmetric bilinear form
\be\label{heeua}
\boldsymbol{E}-\boldsymbol{E}^t\equiv \boldsymbol{E}(\boldsymbol{1}-\boldsymbol{H}).
\ee
The flavor lattice $\Gamma_\text{flavor}$ is the radical of the Dirac form $\boldsymbol{E}-\boldsymbol{E}^t$. From \eqref{heeua} $\Gamma_\text{flavor}$ is the invariant sublattice of $\boldsymbol{H}$
\be\label{whatflaH}
\Gamma_\text{flavor}\equiv\Big\{ \boldsymbol{x}\in\Gamma\;\Big|\; \boldsymbol{H}\boldsymbol{x}=
\boldsymbol{x}\Big\},\qquad f=\mathrm{rank}\,\Gamma_\text{flavor}.
\ee
One shows that the restriction of the symmetric integral form $C$
on the flavor lattice $\Gamma_\text{flavor}$ is positive-definite.

\subsubsection{The matrix form of the duality group}

We have a group homomorphism
\be\label{uuuumzf}
\mathbf{bf}\colon \mathrm{Aut}(\mathscr{R}(\tilde\cf))\to GL(r,\Z),\qquad \sigma\longmapsto \boldsymbol{\sigma},
\ee 
where $\boldsymbol{\sigma}$ is the integral $r\times r$ matrix
\be
[\sigma T_i]=[T_j]\boldsymbol{\sigma}_{ji}.
\ee
Since all auto-equivalences commute with $S$,
\be
\big[\boldsymbol{H},\boldsymbol{\sigma}\big]=0.
\ee
By the argument around eqn.\eqref{arrqaun}, the kernel of the map $\mathbf{bf}$ is precisely $\mathrm{Aut}(\mathscr{R}(\tilde\cf))^0$. Then its
 image $\mathbf{Aut}(\mathscr{R}(\tilde\cf))\equiv\mathbf{bf}(\mathrm{Aut}(\mathscr{R}(\tilde\cf))\subset GL(r,\Z)$ is related by the correspondence $\rho$ in \eqref{rrrqwwzaX} with the $S$-duality group $\mathbb{S}(\cf)$ of the corresponding 4d theory. We conclude that

\begin{fact}
We have the correspondence
\be
\xymatrix{\mathbb{S}(\cf)\ar@{<..>}[r] &\mathbf{Aut}(\mathscr{R}(\tilde\cf))\subset C(\boldsymbol{H})\subset GL(r,\Z),}
\ee
where $C(\boldsymbol{H})$ is the \emph{centralizer} of $\boldsymbol{H}$ in $GL(r,\Z)$.
Thus the $S$-duality is (roughly) the concrete matrix group $\mathbf{Aut}(\mathscr{R}(\tilde\cf))$.
\end{fact}

In \S.\ref{actgro}. we describe effective methods to  determine the group $\mathbf{Aut}(\mathscr{R}(\tilde\cf))$ very explicitly.

\begin{rem}\label{rathersublte} It is tempting to interpret the matrices $\boldsymbol{\sigma}$ as giving the action of $S$-duality on the charges $v\in\Gamma$. However this is not correct in general: the lattice $\Gamma$ is the Grothendieck group of $\mathscr{R}(\tilde\cf)$ not of $\mathscr{C}(\cf)$. Since $\Gamma\simeq K_0(\mathsf{mod}\,\mathsf{Jac}(Q_\ct,\cw_\ct))$ the identification of $\Gamma$ with the IR charge lattice depends on $\ct$, i.e.\! on the region in parameter space.
However it is still true that the group generated by the matrices $\boldsymbol{\sigma}$ is commensurable with a group of $S$-dualities.  
\end{rem}

\subsubsection{Review of the 2d quantum monodromy}
The matrix $\boldsymbol{H}$ representing 2d quantum monodromy on the Grothendieck group of the brane category $\mathscr{B}(\tilde\cf)$
of a ``good'' 2d theory $\tilde\cf$,
which has a 4d correspondent should satisfy a number of properties \cite{Cecotti:2010fi,Cecotti:2012va}.

As before, we assume that the 2d brane category $\mathscr{B}(\tilde\cf)$ has some fractional CY object, and write $\hat c$ for the maximum of their fractional 
CY dimensions (seen as elements of $\mathbb{Q}$). Physical regularity requires $0\leq \hat c<2$.

\begin{fact}[\!\!\cite{Cecotti:1992rm,Cecotti:2012va}] 
One has:
\begin{itemize}
\item $\boldsymbol{H}$ has spectral radius $1$ (hence all its eigenvalues are roots of unity). Then its characteristic polynomial factorizes into cyclotomic polynomials
\be\label{char-pol}
\det\!\big[\boldsymbol{H}-z\big]=(z-1)^f \prod_{d\in D}\Phi_d(z)^{w_d},
\ee
where $f$ is the rank of the flavor group $F$ of the 4d theory, $D$ is a finite subset of $\mathbb{N}_{\geq 2}$, and $\Phi_d(z)$ is the $d$-th cyclotomic polynomial. $w_d$ are \emph{positive} integers satisfying $\sum_{d\in D}\phi(d)\,w_d=2g$, $g$ being the dimension of the 4d Coulomb branch;
\item the Jordan blocks of $\boldsymbol{H}$ are at most of size 2, and there is no non-trivial block for the eigenvalue
$\lambda=1$. In other words, the minimal polynomial of $\boldsymbol{H}$ has the form
\be\label{min-pol}
\text{min-pol}\boldsymbol{H}(z)=(z-1)^a\prod_{d\in D}\Phi_d(z)^{m_d},\quad a=\begin{cases}0 & f=0\\
1& f\geq 1,\end{cases}\ \text{and }\begin{cases}m_d\in\{1,2\}\\
m_d\leq w_d;\end{cases}
\ee
\item if $\hat c<1$, $\boldsymbol{H}$ is semisimple. If $\hat c\geq1$, there cannot be Jordan blocks associated to eigenvalues of the form\footnote{\ This mysterious statement becomes obvious when stated in physical terms. The $h$'s in the range \eqref{veryrelevant} correspond to 2d \emph{relevant} chiral primary operators, which are irrelevant in the UV, and hence do not produce logarithmic violation of scaling ($\equiv$ non-trivial Jordan blocks).} 
\be\label{veryrelevant} 
\lambda=e^{2\pi i h}\qquad\text{with}\quad
0\leq h < 1-\hat c/2.\ee
Indeed,
the symmetric matrix $\boldsymbol{E}+\boldsymbol{E}^t$ is positive definite when restricted to the correspondent eigenspaces.
\end{itemize}
\end{fact}

\subsection{$\mathbb{S}(\cf)$ and the Siegel modular group $Sp(2g,\Z)$}\label{spsp2g}

On general physical grounds \cite{zumn,mybook}, the group of dualities of any 4d $\cn=2$ theory should act on the electro-magnetic charges by an arithmetic  subgroup $\cg_\Z\subset Sp(2g,\R)$ where $g$ is the complex dimension of the Coulomb branch. 
%
Our $S$-duality group is slightly more general than plain electro-magnetic duality, since it  acts non-trivially on the flavor charges, as it happens in $SU(2)$ $N_f=4$ SQCD, where the action on the flavor is through $Spin(8)$ triality
\cite{SW2}. Let us describe how the group $\cg_\Z$ arises in the homological approach. 

By construction, the Dirac form $\omega$ (cfr.\! eqn.\eqref{formsforms}) induces a non-degenerate (integral) symplectic form $\Omega$ on the rank $2g$ lattice of electro-magnetic charges
\be
\Gamma_\text{e.m.}\equiv \Gamma\big/\Gamma_\text{flavor}.
\ee
$\Gamma_\text{flavor}$ spans the $(+1)$-eigenspace of $\boldsymbol{H}$
(cfr.\! eqn.\eqref{whatflaH}).
Since $\boldsymbol{H}$ centralizes
$\mathbf{Aut}(\mathscr{R})$ in $GL(r,\Z)$, all auto-equivalences of $\mathscr{R}(\tilde\cf)$ preserve the flavor sublattice $\Gamma_\text{flavor}$ and the electromagnetic one $\Gamma_\text{e.m.}$, and we have an embedding
\be
\mathbf{Aut}(\mathscr{R})\subset GL(2g,\Z)\times GL(f,\Z).
\ee
Moreover, any element of $\mathbf{Aut}(\mathscr{R})$ leaves invariant both the positive symmetric form $C$ and the symplectic form $\Omega$ (cfr.\! eqn.\eqref{formsforms}).
Hence
\be\label{uuyyhX}
\mathbf{Aut}(\mathscr{R})\subset Sp(2g,\Z)_\Omega \times O(f,\Z)_C,
\ee
where the groups in the \textsc{rhs}
are the arithmetic groups preserving
the non-degenerate forms $\Omega$ and $C$. They are lattices in $Sp(2g,\R)$ and $O(f,\R)$, respectively.

Since $C$ is positive-definite, the real group $O(f,\R)_C$ is compact, and then its arithmetic sub-group $O(f,\Z)_\Omega$ is finite. Hence, modulo commensurability, we may forget the second factor in the \textsc{rhs} of \eqref{uuyyhX}. Up to $\Z$-equivalence $\Omega$ may always be written in normal form \cite{book2}
\be
\Omega=\left(\begin{array}{c|c}
0 & \phantom{-}D\phantom{-}\\ \hline
-D\phantom{-}& 0\end{array}\right),\qquad D=\mathrm{diag}(d_1,\dots,d_g)\in\mathbb{N}^g,\ \ \text{with }
d_i\mid d_{i+1}.
\ee
When all $d_i$'s are equal, $Sp(2g,\Z)_\Omega$ reduces to the Siegel modular group $Sp(2g,\Z)$.
The physical meaning of the invariants $d_i$ will be discussed elsewhere \cite{toappear}.

\begin{fact}[see e.g.\! \cite{arithm}] The arithmetic group $Sp(2g,\Z)_\Omega$ is commensurable to the Siegel modular group $Sp(2g,\Z)$.
\end{fact}

Thus, modulo finite groups, the $S$-duality group $\mathbb{S}(\cf)$ may be seen as a subgroup of the Siegel modular group $Sp(2g,\Z)$
(although in a rather subtle way, see \textbf{Remark \ref{rathersublte}}).

\subsection{The flavor Weyl group}\label{weylll}
$\Gamma_\text{flavor}$ is the lattice of flavor charges which remain unbroken for generic values of the masses. The masses just softly break the UV flavor group $F$, which is non-Abelian in general, to its maximal torus $U(1)^f$; 
$\Gamma_\text{flavor}$ is then the group of characters of the maximal torus of $F$. Suppose that (modulo isogeny)
\be
F=U(1)^h\times F_1\times F_2\times\cdots\times F_t,\qquad F_\ell\ \text{a simple Lie group}.
\ee
Then
\be
\Z^h \oplus \Gamma^\text{root}_1\oplus\cdots\oplus\Gamma^\text{root}_t\subseteq \Gamma_\text{flavor}\subseteq Z^h\oplus\Gamma^\text{weight}_1\oplus\cdots\oplus\Gamma^\text{weight}_t
\ee
where $\Gamma^\text{root}_\ell$ (resp.\! $\Gamma^\text{weight}_\ell$) is the root (resp.\! weight) lattice of the simple Lie group $F_\ell$. On physical grounds, the Weyl group
\be
\mathrm{Weyl}(F_1)\times 
\mathrm{Weyl}(F_2)\times\cdots\times
\mathrm{Weyl}(F_\ell) 
\ee
should be part of the $S$-duality group
$\mathbf{Aut}(\mathscr{R}(\tilde\cf))$. Since the action of this group on $\Gamma_\text{flavor}$ leaves the positive-definite quadratic form $C$ invariant, we get
\be
C\simeq  M_h\oplus a_1\,C_1\oplus a_2\,C_2\oplus\cdots \oplus a_t C_t,
\ee
where $C_\ell$ is the Cartan matrix of
$F_\ell$, $a_\ell$ is a positive integer, and $M_h$ a positive-definite, symmetric, integral $h\times h$ matrix. 
Given the category $\mathscr{R}(\tilde\cf)$,
we know $C$ and hence can infer the non-Abelian  enhancement of the flavor group in the UV, $U(1)^f\to F$.
The image of $\mathbf{Aut}(\mathscr{R}(\tilde\cf))$ in the finite group $O(f,\Z)_{C}$, eqn.\eqref{uuyyhX}, is then a subgroup of
\be\label{tttttmmmqq9}
O(h,\Z)_{M_h} \times \Big(\mathrm{Weyl}(F_1)\ltimes \mathrm{Aut}(F_1)\Big)\times\cdots \times \Big(\mathrm{Weyl}(F_\ell)\ltimes \mathrm{Aut}(F_\ell)\Big)\times \mathfrak{S},
\ee
where $\mathrm{Aut}(F_k)$ is the automorphism group of the $F_k$ Dynkin graphs, $\mathfrak{S}$ permutes isogenous factors in $F$.

\subsection{Tables of 4d/2d correspondences}

In this last part of the section, we list the 2d correspondent theories to the 4d $\cn=2$ models we focus on in this paper.  
Up to some exception (as $D^{(1,1)}_4(G)$), the 4d models $\cf$ listed in the \textbf{Introduction} are related to  
2d (2,2) Landau-Ginzburg (LG) models
with four chiral superfields and superpotential $W(x,y,u,v)$ as in the following table:

\be
\begin{tabular}{l@{\hspace{25pt}}l}\hline\hline
model & $W(x,y,u,v)$\\\hline
$(G,G^\prime)$ & $W_G(x,y)+W_{G^\prime}(u,v)$\\
$(G,\widehat{A}(p,q))$ & 
$W_{\widehat{A}(p,q)}(x,y)+W_G(u,v)$\\
$(G,\widehat{D}_r)$ & $W_{\widehat{D}_r}(x,y)+W_G(u,v)$\\
$(G, \widehat{E}_r)$ &
$W_{\widehat{E}_r}+W_G(u,v)$\\
$E^{(1,1)}_6(A_{n-1})$ & $x^3+y^3+u^3+v^n$\\
$E^{(1,1)}_7(G)$ & $x^4+y^4+W_G(u,v)$\\
$E^{(1,1)}_8(G)$ & $x^6+y^3+W_G(u,v)$\\\hline\hline
\end{tabular}
\ee
The Du Val polynomials $W_G(u,v)$ and the affine functions $W_{\widehat{H}}(x,y)$ are listed in table \ref{duval}.

  \begin{table}
  \begin{center}
  \begin{minipage}{190pt}
  \begin{center}
  \begin{tabular}{ccc}\hline\hline
  $G$ & $W_G(x,y)$ &$\hat c(G)$\\\hline
  $A_{n-1}$ & $x^n+y^2$ & $(n-2)/n$\\
$D_n$ & $x^{n-1}+xy^2$ & $(n-2)/(n-1)$\\
$E_6$ & $x^4+y^3$ & $5/6$\\
$E_7$ & $x^3+xy^3$ & $8/9$\\
$E_8$ & $x^5+y^3$ & $14/15$\\\hline\hline  
  \end{tabular}
  \caption{\label{duval} $ADE$ minimal singularities.}\end{center}
  \end{minipage}\qquad\quad
  \begin{minipage}{150pt}
  \begin{center}
  \begin{tabular}{cc}\hline\hline
  $\widehat{H}$ & $W_{\widehat{H}}(u,v)$\\\hline
  $\widehat{A}(p,q)$ & $e^{pu}+e^{-qu}+v^2$\\
$\widehat{D}_r$ & $u^{r-2}+y^2+x^2y^2$\\
$\widehat{E}_r$ & $u^{r-3}+y^3+x^2y^2$\\
\hline\hline  
  \end{tabular}
  \caption{\label{afff}Affine $W$'s.}\end{center}
  \end{minipage}
  \end{center}
  \end{table}

From the 2d perspective, the following  identifications of 4d $\cn=2$ SCFTs are obvious
\be\label{oooqe}
\begin{aligned}
E_6^{(1,1)}(A_1)&= (D_4,A_2),& E_6^{(1,1)}(A_2)&= (D_4,D_4),& 
E_6^{(1,1)}(A_3)&=(D_4, E_6),\\
E_6^{(1,1)}(A_4)&=(D_4, E_8)
&
E_7^{(1,1)}(A_1)&=(A_3,A_3),&
E_7^{(1,1)}(A_2)&=(A_3,E_6),\\
E_7^{(1,1)}(D_4)&=(E_6,E_6),
&
E_8^{(1,1)}(A_1)&=(A_5,A_2),&
E_8^{(1,1)}(A_2)&=(A_5, D_4),\\
E_8^{(1,1)}(A_3)&=(A_5,E_6),
&
E_8^{(1,1)}(A_4)&=(A_5,E_8).
\end{aligned}
\ee

In particular, for the models 
\be\label{ppaw8yX}
E^{(1,1)}_7(A_{n-1})\quad\text{and}\quad  E^{(1,1)}_8(A_{n-1})
\ee
 one is lead to the \emph{derived}  category of branes in  the Landau-Ginzburg models associated to the triangle singularities 
\be\label{wwwwz}
x^4+y^4+z^n,\quad\text{and respectively}\quad x^6+y^3+z^n.
\ee
Orlov \cite{orlov2} shows that these 
branes categories are equivalent to  the Frobenius stable categories
\begin{equation}
\underline{\mathsf{vect}}\,\mathbb{X}(p_1,p_2,p_3)
\end{equation} 
where $(p_1,p_2,p_3)$ are the exponents in the superpotential \eqref{wwwwz}, that is, $(4,4,n)$ and, respectively, $(6,3,n)$.
The category $\underline{\mathsf{vect}}\,\mathbb{X}(p_1,p_2,p_3)$
\cite{adechain}  is defined as the quotient of the category of bundles
(coherent sheaves of positive rank) over the weighted projective line $\mathbb{X}(p_1,p_2,p_3)$ \cite{wpl} by the ideal of line bundles.

The authors of ref.\!\cite{adechain} construct an explicit tilting object $T\in \underline{\mathsf{vect}}\,\mathbb{X}(p_1,p_2,p_3)$
such that
\be
\mathrm{End}(T)\simeq \C A_{p_1-1}\times\C A_{p_2-1}\times
\C A_{p_3-1}.
\ee
Then \cite{adechain}
\be
\underline{\mathsf{vect}}\,\mathbb{X}(p_1,p_2,p_3)\simeq D^b(\mathsf{mod}(\C A_{p_1-1}\times\C A_{p_2-1}\times
\C A_{p_3-1})). 
\ee
Note that the models $(A_{p_1-1},A_{p_2,-1})$ correspond to the case $\underline{\mathsf{vect}}\,\mathbb{X}(p_1,p_2,2)$.

This explicit realization of the 4d/2d correspondence for a subclass of the theories of interest allows to double check many results, by comparing the two sides of the correspondence. This is particularly helpful for the models \eqref{ppaw8yX} which do not belong to the well-understood cases $(G,G^\prime)$ and $(G,\widehat{H})$.  The models \eqref{ppaw8yX} (and their generalizations) will be studied in section \ref{moregeneral}.

\section{Auto-equivalences of triangle categories}\label{autttoequiv}

We review some general aspects of auto-equivalences in triangle categories.

\subsection{Thomas-Seidel twists in a triangle category}
\label{TSS1}

A particular class of auto-equivalences of a triangulated category $\mathscr{C}$ is given by the Thomas-Seidel twists associated to spherical objects \cite{seidel2001braid,twi}; they have the advantage of having a very explicit
form. In many cases the spherical twists generate the full group of auto-equivalences\footnote{\ In facts, all autoequivalences of a triangulated category are given by spherical twists provided one generalizes from twists around spherical \emph{objects} to twists around spherical \emph{functors} \cite{segal}. } $\mathrm{Aut}(\mathscr{C})$, and hence we may expect them to
produce a `substantial' part of the $S$-duality group.

We recall some basic definition.

\begin{defn} An object $A$ of a triangulated category $\mathscr{T}$
with Serre functor $S$ is called
\emph{spherical} iff there exists a $n\in\mathbb{N}$ such that:
\begin{align}
\label{cyau1}
SA&=A[n] &&\text{($A$ is CY of dimension $n$)}\\
\mathrm{Hom}(A,A[k])&=\begin{cases}
\C &\text{for }k=0,n\\
0 & \text{otherwise.}
\end{cases}&&\text{($A$ has the homology of a $n$-sphere)}
\label{qqquasimm}
\end{align}

As in \cite{twi},
we write $\mathrm{Hom}^\bullet(A,X)$ for the complex $\bigoplus_k \mathrm{Hom}(A[k],X)[k]$ of $\C$-spaces with the zero differential, and say that our triangulated category $\mathscr{T}$ is \textit{$\mathrm{Hom}^\bullet$-finite} iff $\dim \mathrm{Hom}^\bullet(A,B)<\infty$ for all objects $A, B$.
\end{defn}


\begin{defn} Let $A$ be a spherical object in the triangulated category $\mathscr{T}$ with Serre functor $S$ and $\mathrm{Hom}^\bullet$-finite. The \textit{Thomas-Seidel twist $T_A\colon\mathscr{T}\to\mathscr{T}$} \cite{seidel2001braid}, is the auto-equivalence $X\mapsto T_A(X)$ defined by the exact triangle
\be\label{wwtrian}
\mathrm{Hom}^\bullet(A,X) \otimes A \to X \to T_A(X)\to,
\ee
where the first arrow is the canonical evaluation.
\end{defn}

Let $\sigma\in\mathrm{Aut}(\mathscr{C})$ be an auto-equivalence; we have the adjoint action
\be
\sigma T_A\sigma^{-1}= T_{\sigma(A)}.
\ee
In particular,
\be\label{lasw2}
T_{A[1]}=T_A,\qquad\quad ST_AS^{-1}=T_{SA}=T_{A[n]}=T_A,
\ee
that is, $T_A$ depends only on the $\mathfrak{A}$-orbit of $A$.

\subsubsection{Braid group actions}

The set of all Thomas-Seidel twists
of all spherical objects generate a group of auto-equivalence which often is a braid group.

\begin{defn}
An $(A_m)$-configuration, $m \geq 1$, in the triangulated category $\mathscr{C}$ is a collection of $m$ spherical
objects $A_1,... , A_m$ such that
\be
\dim \mathrm{Hom}^\bullet_{\mathscr{T}}(A_i, A_j ) = 
\left\{\begin{array}{lr}
1,  \ \ \ |i - j| = 1,\\
0, \ \ \ |i - j| \geq 2.
\end{array}
\right.
\ee\end{defn}

\begin{fact}[\!\!\cite{seidel2001braid}]
\label{ffffccq} 
The twist $T_A$ along any spherical object $A$ is an exact autoequivalence
of $\mathscr{C}$. Moreover, if $A_1,... , A_m$ is an $(A_m)$-configuration, the
twists $T_{A_i}$
satisfy the $\mathcal{B}_{m+1}$ braid relations up to graded natural isomorphism:
\begin{equation}\label{braid}\begin{aligned}T_{A_i}T_{A_{i+1}}T_{A_i}&\cong  T_{A_{i+1}} T_{A_i}T_{A_{i+1}} \ \  for  \ \ i = 1,... ,m - 1,\\
T_{A_i}T_{A_j}&\cong T_{A_j}T_{A_i} \ \ for \ \  |i - j| \geq 2.\end{aligned}
\end{equation}\end{fact}

\subsubsection{$p$-periodic categories}

If the category $\mathscr{T}$ is $p$-periodic, $[p]\simeq \mathrm{Id}$, it cannot be
$\mathrm{Hom}^\bullet$-finite, and the triangle \eqref{wwtrian} makes no sense. It is natural to restrict the direct sum to \emph{one} orbit, i.e.\! to replace \eqref{wwtrian} with
\be\label{wwtrianII}
\bigoplus_{k=0}^{p-1}\mathrm{Hom}(A[k],X) \otimes A[k] \to X \to T_A(X)\to.
\ee
We then define ``spherical'' an object $A\in\mathscr{T}$ iff
it is $n$--CY and
\be\label{mmbcd}
\dim\mathrm{Hom}(A,A[k])=\delta_{k,0}^{(p)}+\delta_{k,n}^{(p)},
\ee
where $\delta_{i,j}^{(p)}$ is the mod $p$ Kronecker delta (see eqn.\eqref{modelta}). This definition is a bit tricky when $n\mid p$, and for the moment we exclude this special case.
Likewise, we say that
a sequence of $m\geq 1$ spherical objects in the present sense form an $(A_m)$--configuration iff
\be
\sum_{k=0}^{p-1}\dim \mathrm{Hom}(A_i, A_j[k] ) = 
\left\{\begin{array}{lr}
1,  \ \ \ |i - j| = 1,\\
0, \ \ \ |i - j| \geq 2.
\end{array}
\right.
\ee 
With these modifications on the definitions (and some mild assumptions) the functors $T_{A_i}$ yield auto-equivalences which satisfy the braid relations \eqref{braid}.

\subsection{Telescopic functors}\label{s:ddddd4561}
We need a generalization of the twist inspired by the telescopic functors \cite{lenzingmeltzer,meltzer} which played the central role in the discussion of $S$-duality for the tubular SCFTs \cite{shepard}. 

Let us consider the first condition for a spherical objects, the CY one \eqref{cyau1}. All auto-equivalence commute with $S$ (since $S$ is unique), and \eqref{cyau1} enforces that property, cfr.\! eqn.\eqref{lasw2}.\footnote{\ The second condition \eqref{qqquasimm} guarantees that $T_A$ induces an isometry of the Euler form.} Write $S=T[n]$. If $A$ is not $n$-CY, the orbit
$\{T^k A\}_{k\in\Z}$ does not reduce to the single object $A$. However, if the orbit is periodic, $T^{k+p}A\simeq T^kA$,
we can still enforce commutativity with $T$ (hence with $S$) by ``averaging'' over the $T$-orbit. In order to achieve periodicity we are free to twist $T$ by any number of shifts. Let us make the ``averaging'' procedure more precise.
\medskip

$\mathscr{R}$ a triangular category with
Serre duality functor $S=T[n]$.
$\mathscr{R}$ is allowed to be $\ell$-periodic ($\ell=\infty$ means non periodic).

\begin{defn}
$A\in\mathscr{R}$ has a \textit{spherical $T$-orbit of period $p\in\mathbb{N}$} iff $p$ is the smallest positive integer such that
\begin{align}\label{four4}
T^p A&\simeq A\\
\dim\mathrm{Hom}(A,T^kA[j])&=\delta^{(p)}_{k,0}\,\delta^{(\ell)}_{j,0}+\delta^{(p)}_{k,1}\,\delta^{(\ell)}_{j,n},
\label{four2}
\end{align}
where $\delta^{(p)}_{i,j}$ is the mod $p$ Kronecker delta \eqref{modelta}.
\end{defn}

A spherical object is just
a spherical  $T$-orbit of period $p=1$.
\medskip

To a spherical $T$-orbit of period $p$, $\{T^k A\}_{k=0}^{p-1}$
we associate a pair of quasi-inverse auto-equivalences, $L_{A}$ and $R_{A}$, called the \textit{telescopic functors.}
If $\ell=\infty$ they are 
defined by the following cones \cite{lenzingmeltzer,meltzer}
\begin{gather}\label{telescaa}
\bigoplus_{i=0}^{p-1} \mathrm{Hom}^\bullet(T^i A,X) \otimes T^i A \to X \to L_{A}(X) \to\\
R_{A}(X) \to X \to \bigoplus_{i=0}^{p-1} D\mathrm{Hom}^\bullet(X,T^i A) \otimes T^i {A} \to\\
L_AR_A\simeq R_AL_A\simeq \mathrm{Id}.\end{gather}
For $p=1$, $L_A$ coincides with the Thomas-Seidel twist $T_A$. 

More generally, for $\ell<\infty$, we replace
\be
\mathrm{Hom}^\bullet(T^iA,X)\otimes T^i A\equiv \bigoplus_{k\in\Z}\mathrm{Hom}^\bullet(T^iA[k],X)\otimes T^i A[k]
\ee
with
\be 
\bigoplus_{k=0}^{\ell-1}\mathrm{Hom}^\bullet(T^iA[k],X)\otimes T^i A[k].
\ee

\paragraph{Uniform notation.} 
It is convenient to rewrite all functors introduced above in a more transparent notation. Let $\mathfrak{F}\simeq \Z^2$ be the free Abelian group generated by $T\equiv S[-n]$ and $[1]$
and $\mathfrak{F}_A\subset \mathfrak{F}$ the isotropy subgroup of the object $A$. Then we may 
write the defining triangle in an uniform way
\be\label{unifffr}
\bigoplus_{\rho\in \mathfrak{F}/\mathfrak{F}_A}\mathrm{Hom}(\rho A,X)\otimes \rho A\to X\to L_A(X)\to,
\ee
while for an object $A$ with a spherical orbit
\be
\sum_{\rho\in\mathfrak{F}/\mathfrak{F}_A}\dim\mathrm{Hom}(\rho A, A)=2.
\ee

The proof that $L_A$ is an auto-equivalence for general $p$ is a generalization of the one for $p=1$.
Note that $L_A$ depends only on the $\mathfrak{F}$-orbit of the object $A$, not on the object itself.

Clearly $L_A$ sends spherical orbits into spherical orbits. Then the telescopic functors associated to two spherical $T$-orbits, $\{T^kA\}$, $\{T^kB\}$, satisfy the adjoint action formula
\be\label{uuuqeer}
L_A L_B R_A\simeq  L_{L_A(B)}.
\ee

\begin{defn}\label{ttttealbbr} Let $\{T^kA_i\}_{k=0}^{p_i-1}$ be  a collection of spherical $T$-orbits of period $p_i$ ($i=1,\dots,m$). We say that they form a $(A_m)$-configuration iff
\be
\sum_{\rho\in\mathfrak{F}/\mathfrak{F}_{A_i}} \dim \mathrm{Hom}(\rho  A_i,A_j)
= 
\left\{\begin{array}{lr}
1,  \ \ \ |i - j| = 1,\\
0, \ \ \ |i - j| \geq 2.
\end{array}
\right.
\ee 
Note that the definition implies that all orbits in the $(A_m)$-configuration have the same period, $p_i=p$.
\textbf{Fact 2} generalizes to $p>1$. \end{defn}
 
\begin{fact}\label{braidtele}
The telescopic functors $L_{A_i}$ associated to a $(A_m)$-configuration of spherical orbits satisfy the $\cb_{m+1}$ relations
 \begin{equation}\label{braid2}
 \begin{aligned}L_{A_i}L_{A_{i+1}}L_{A_i}&\cong  L_{A_{i+1}} L_{A_i}L_{A_{i+1}} \ \  for  \ \ i = 1,... ,m - 1,\\
L_{A_i}L_{A_j}&\cong L_{A_j}L_{A_i} \ \ for \ \  |i - j| \geq 2.\end{aligned}
\end{equation}
\end{fact}

Given a triangle category as above, $\mathscr{R}$, we write $\mathrm{Tel}(\mathscr{R})\subset\mathrm{Aut}(\mathscr{R})$ for the subgroup generated by the telescopic functors (including twists) and the ``obvious'' auto-equivalences ($S$, shifts, etc.). 
$\mathrm{Tel}(\mathscr{R})$ is expected to be an interesting part of the auto-equivalence group and it is often the full $\mathrm{Aut}(\mathscr{R})$.

\subsection{Half-hypermultiplets and spherical half-orbits}\label{TSS3}

\subsubsection{Physical motivation}\label{phphmot}
As we shall illustrate in section 5, in a Lagrangian $\cn=2$ QFT a quark which is a \emph{full}\footnote{\ A quark is a \emph{full} hypermultiplet iff its fermionic states form a quaternionic representation of the gauge group of the form $W\oplus W^\vee$.} hypermultiplet corresponds to an object with a spherical orbit and hence to a telescopic auto-equivalence, that is, a duality of the Lagrangian QFT.

In a QFT with a weakly coupled Lagrangian description,
the duality associated to a full quark may be understood semiclassically. 
The simplest instance is $SU(2)$ coupled to $N_f$ full quarks
studied by Seiberg and Witten \cite{SW2}.
We recall their argument: the spectrum contains dyons which, as the coupling goes to zero, $g\to0$, become classical solitons with masses $O(1/g^2)$. The hypermultiplet fermions, $\psi_\pm^a$, $a=1,\dots,N_f$, have zero modes in the soliton background, whose quantization yields the Clifford algebra in dimension $2N_f$. Since the $\psi_\pm^a$ carry electric and flavor charge, it follows that the dyons of even (odd) electric charge are in the chiral (anti-chiral) representation of the flavor group $\mathrm{Spin}(2N_f)$. The telescopic duality associated with the $a$-th quark is simply the action on the dyons of the zero mode operator $\psi^a_+$ (or $\psi^a_-$ together with a shift by 2 of the electric charge); in the language of coherent sheaves over weighted projective lines, this auto-equivalence for $SU(2)$ with $N_f$ quarks is called the ``one-point shift'' (cfr.\! \S.10.3 of \cite{lenzinghandbook}). 

This physical understanding of simple telescopic auto-equivalences leaves the feeling that something is missing:
what about half-hypermultiplets?\footnote{\ A quark supermultiplet in a Lagrangian $\cn=2$ QFT is a \emph{half-hypermultiplet} if it states form a quaternionic representation of the gauge group which cannot be written in the form $W\oplus W^\vee$.}
They also should play a role in the duality group. Clearly, the above semiclassical argument may be applied also in this case (provided the theory is free from $\Z/2\Z$ anomalies \cite{Witten:1982fp}). 

In terms of \textsc{susy} representation theory, the basic difference between full- and half-hypermultiplets is that in the second case the PCT-conjugate states belong to the same irreducible representation, while in the first case we have the direct sum of a representation and its PCT-conjugate one. From the point of view of the triangle categories which describe the BPS sector, PCT conjugation is essentially the shift $X\to X[1]$.
Then the ``abstract'' categorical version of the distinction between full- and half-hypers is the following: in the first case
$X$ and $X[1]$ belong to \emph{distinct} $T$-orbits, while in the second one to the \emph{same} $T$-orbit which may be decomposed into two half-orbits. Physical intuition suggests that there must be a generalization of the telescopic auto-equivalences associated to spherical
\emph{half}-orbits.

Now we give the proper definitions.

\subsubsection{Auto-equivalences from spherical half-orbits}
\begin{defn}
An object $A$ in a $2m$-periodic triangle category 
$\mathscr{R}$, with Serre functor $S=T[n]$, has \textit{a periodic $T$-orbit of half-period $q$,}
if $q$ is the smallest positive integer such that
\be\label{hhaw}
T^qA=A[m],
\ee
so that in the notation introduced at the end of \S.\ref{s:ddddd4561}
\be
\mathfrak{F}/\mathfrak{F}_A=\big\{T^k[j],\ k=0,1,\dots, q-1, \ 
j=0,1,2\dots, 2m-1 \big\}.
\ee 
The object $A$ has a \textit{spherical half-orbit} iff, in addition to \eqref{hhaw},\footnote{\ Again, we omit some special cases as $q=1$ or $m\mid n$.}
\be
\dim\mathrm{Hom}(A,T^kA[j])=
\delta^{(2q)}_{k,0}\delta^{(2m)}_{j,0}+\delta^{(2q)}_{k,q}\delta^{(2m)}_{j,m}+
\delta^{(2q)}_{k,1}\delta^{(2m)}_{j,n}+\delta^{(2q)}_{k,q+1}\delta^{(2m)}_{j,m+n}.
\ee
\end{defn}

Then
\be
\sum_{\rho\in \mathfrak{F}/\mathfrak{F}_A} \dim \mathrm{Hom}(\rho A,A)=2,
\ee
so the functor $X\to L_A(X)$ defined in eqn.\eqref{unifffr} still makes sense and is an auto-equivalence.

\subsection{Twists in $\mathscr{C}(\cf)$ vs.\! telescopic functors in
$\mathscr{R}(\tilde\cf)$}\label{ujjujuja}

We return to the set-up of section \ref{4d2d3ct}, in particular to diagram \eqref{rrrqwwzaX}.

\subsubsection{First case: $\tilde\cf$ is a SCFT}\label{fffistcaase}

The category $\mathscr{B}(\tilde\cf)=D^b\mathscr{A}(\tilde\cf)$ has Serre functor $S\equiv T[2]$, tilting object $\ct=\oplus_i T_i$, and fractional Calabi-Yau dimension $\tfrac{a}{b}<2$. Thus $T^b=[a-2b]\equiv [-m]$. In particular, the 4d cluster category
\be
\mathscr{C}(\cf)=\mathsf{Hu}_\triangle\!\big(\mathscr{B}(\tilde\cf)/ T^\Z\big)
\ee
is $m$-periodic.
$\mathscr{C}(\cf)$ is 2-Calabi-Yau, so the condition
\eqref{cyau1}
 holds for all objects with $n=2$; in other words
all objects belong to a $T$-orbit with $p=1$. 
The condition of $A\in\mathscr{C}(\cf)$ being spherical, eqn.\eqref{mmbcd},
becomes
\be\label{mmbcdII}
\dim\mathrm{Hom}_{\mathscr{C}}(A,A[j])=\delta_{j,0}^{(m)}+\delta_{j,2}^{(m)},
\ee
the special case being $m|2$;
since $m=\mathfrak{m}\cdot\gcd(a,2)$, this is equivalent to the condition $\mathbb{M}\simeq\mathrm{Id}$, as in all Lagrangian SCFTs. Note that $m>1$ (since $m=1$ is incompatible with the existence of a cluster-tilting object).
Then
from eqn.\eqref{qqquasimm}
we see that all spherical objects are automatically 
 rigid. Hence, by \textbf{Fact \ref{norigg}}, all spherical objects of $\mathscr{C}(\cf)$ actually lay in the orbit category \eqref{ororbt}. Given a spherical object
$A$  in the orbit category, we consider one of its lifts in the derived category
$\mathscr{B}(\tilde\cf)$ (still written $A$). Then
\be
\begin{split}
\mathrm{Hom}_\mathscr{C}\big(A,A[j])&= \bigoplus_{k\in\Z} \mathrm{Hom}_\mathscr{B}(A,T^k A[j]\big)=\\
&=\bigoplus_{k\in\Z} 
D\,\mathrm{Hom}_\mathscr{B}\big(A,
T^{1-k} A[2-j]\big),
\end{split}
\ee 
so if $A$ is spherical in $\mathscr{C}(\cf)$ its lift $A\in\mathscr{B}(\tilde\cf)$ satisfies
\be\label{yyyqaz9X}
\dim\mathrm{Hom}_{\mathscr{B}}\big(A,T^k  A[j]\big)=\sum_{s\in\Z}\Big(\delta_{k,bs}\,\delta_{j,ms}+\delta_{k,bs+1}\,\delta_{j,ms+2}\Big),
\ee
or, equivalently, its image in the $2\mathfrak{m}$-periodic category $\mathscr{CR}(\tilde\cf)$
\be\label{yyyqaz10}
\dim\mathrm{Hom}\big(A,T^k A[j]\big)=\begin{cases}\delta^{(b)}_{k,0}\,\delta_{j,0}^{(2\mathfrak{m})}+\delta^{(b)}_{k,1}\,\delta^{(2\mathfrak{m})}_{j,2} &a\ \text{even}\\
\delta^{(2b)}_{k,0}\,\delta_{j,0}^{(2\mathfrak{m})}+\delta^{(2b)}_{k,1}\,\delta^{(2\mathfrak{m})}_{j,2}+
\delta^{(2b)}_{k,b}\,\delta_{j,\mathfrak{m}}^{(2\mathfrak{m})}+\delta^{(2b)}_{k,b+1}\,\delta^{(2\mathfrak{m})}_{j,\mathfrak{m}+2} &a\ \text{odd}.
\end{cases}
\ee
Conversely, if $A\in\mathscr{B}(\tilde\cf)$ satisfies \eqref{yyyqaz9X}, its image $A\in\mathscr{C}(\cf)$ is spherical 
unless $\mathfrak{m}=1$.
However, we shall see that the $\mathfrak{m}=1$ behaves very much as the general case, and we don't need to consider it an exception. 
In other terms, if $A\in\mathscr{C}(\cf)$ is spherical iff 
 $A\in\mathscr{CR}(\tilde\cf)$
 has a spherical (half-) $T$-orbit.

We denote the image of $A$ in $\mathscr{R}(\tilde\cf)$ by the same symbol. In $\mathscr{R}(\tilde\cf)$ eqns.\eqref{yyyqaz9X},\eqref{yyyqaz10} become 
\be
\dim \mathrm{Hom}_{\mathscr{R}}(A, T^k A[j])= 
\begin{cases}(\delta^{(b)}_{k,0}+\delta^{(b)}_{k,1})\delta^{(2)}_{j,0} &a\ \text{even}\\
(\delta^{(2b)}_{k,0}+\delta^{(2b)}_{k,1})\delta_{j,0}^{(2)}+(\delta^{(2b)}_{k,b}+\delta^{(2b)}_{k,b+1})\delta^{(2)}_{j,1} &a\ \text{odd}.
\end{cases}
\ee
where $p$ is the period of the object $A$ in $\mathscr{R}(\tilde\cf)$.
The last equation is identical to the characterization of spherical (half-) $T$-orbits in $\mathscr{R}(\cf)$.
Thus there is a one--to--one correspondence of twist auto-equivalences in $\mathscr{C}(\cf)$ and of telescopic auto-equivalences in
$\mathscr{R}(\tilde\cf)$.
\medskip

To be more explicit, it is convenient to define $L_A$ directly at the level of the 
derived category $\mathscr{B}(\tilde\cf)$. Let $A\in\mathscr{B}(\tilde\cf)$ satisfy
eqn.\eqref{yyyqaz9X}; then
(cfr.\! \eqref{qqqq1z})
\be
\mathfrak{F}/\mathfrak{F}_A=\mathfrak{A}
\ee
and, for all $X\in\mathscr{B}(\tilde\cf)$, we define $L_A(X)\in\mathscr{B}(\tilde\cf)$ by the cone
\be
\to \bigoplus_{\rho\in\mathfrak{A}}\mathrm{Hom}_\mathscr{B}(\rho A,X)\otimes\rho A\xrightarrow{\;\text{can}\;}
X\to L_A(X)\to
\ee
which makes sense by eqn.\eqref{lljazfrr}. $X\mapsto L_A(X)$ is an autoequivalence of the derived category.
With reference to the diagram \eqref{rrrqwwzaX}, let us apply the functor $\pi_\star$ ($\star=\emptyset, r,c,cr$) to this triangle. We get
\be\label{aaaq12}
\bigoplus_{\rho\in \mathfrak{A}/\mathfrak{A}_\star} \left(\bigoplus_{\sigma\in\mathfrak{A}_\star}\mathrm{Hom}_\mathscr{B}(\sigma \rho A, X)\right)\otimes \rho \pi_\star A\to \pi_\star X \to \pi_\star L_A(X)\to
\ee
Now the sum in the large parenthesis is just
\be
\mathrm{Hom}_{\mathscr{B}/\mathfrak{A}_\star}(\rho A, X)\equiv \mathrm{Hom}_{\mathsf{Hu}_\triangle\!(\mathscr{B}/\mathfrak{A}_\star)}(\rho A, X),
\ee
while for the image of $A$ in $\mathsf{Hu}_\triangle(\mathscr{B}/\mathfrak{A}_\star)$
\be
\mathfrak{F}/\mathfrak{F}_A=\mathfrak{A}/\mathfrak{A}_\star\qquad \star=\emptyset,c,r,cr.
\ee
Hence the triangle \eqref{aaaq12} corresponds to the triangle which defines the 
telescopic functor in the 
triangular category $\mathsf{Hu}_\triangle\!(\mathscr{B}/\mathfrak{A}_\star)$, see eqn.\eqref{unifffr}. We see that 
for $A, X\in \mathscr{B}$
with $A$ satisfying \eqref{yyyqaz9X} we have
\be
L_{\varpi_\star(A)}\big(\varpi_\star(X)\big)= \varpi_\star\!\big(L_A(X)\big).
\ee
In particular, this equation holds for $X=T_i$ (the indecomposable summands of the tilting object $\ct$).
Since an auto-equivalence is uniquely determined by its action on the full subcategory $\mathsf{add}\,\ct$, we see that 
the subgroups\,\footnote{\ \label{footTel}Properly speaking, by the groups $\mathrm{Tel}(\mathsf{Hu}_\triangle(\mathscr{B}(\tilde\cf)/\mathfrak{A}_\star))$ we mean the subgroups of the auto-equivalence groups which are \textit{generated by telescopic functors of the form $L_A$ with $A\in
\mathscr{B}(\tilde\cf)/\mathfrak{A}_\star$.} For $\star=\emptyset,c$ these are the full set of telescopic functors, since all objects with spherical orbits belong to the orbit subcategory. For $\star=r,cr$ we cannot exclude that there are other spherical orbits in the triangular hull (but we believe that there are none). } 
\be
\mathrm{Tel}(\mathsf{Hu}_\triangle(\mathscr{B}(\tilde\cf)/\mathfrak{A}_\star))\subset
\mathrm{Aut}(\mathsf{Hu}_\triangle(\mathscr{B}(\tilde\cf)/\mathfrak{A}_\star))
\ee
are related by the maps $\varpi_\star$ and are ``essentially'' equal. 
We stress that the case $\mathfrak{m}=1$ is no exception: although the object $\varpi_c(A)$ is not spherical in the standard sense, the autoequivalence $L_{\varpi_c(A)}$ still makes sense being induced from $L_A$.  
In conclusion,
\begin{fact}
$A\in \mathscr{C}(\cf)$ is spherical if and only if the canonical image in $\mathscr{R}(\tilde\cf)$ of its lift $\tilde A\in\mathscr{B}(\tilde\cf)$
belongs to a spherical $T$-orbit in $\mathscr{R}(\tilde\cf)$ (and hence also in
$\mathscr{CR}(\tilde\cf)$).
In the diagram
\be\label{ddddiagrmqi}
\begin{gathered}
\xymatrix{&& \mathrm{Tel}(\mathscr{CR}(\tilde\cf))\ar[dll]_{\mathrm{Tel}(\varpi_1)}\ar[drr]^{\mathrm{Tel}(\varpi_2)}\\
\mathrm{Tel}(\mathscr{C}(\cf))\ar@{<..>}[rrrr]^\thickapprox_\rho&&&&\mathrm{Tel}(\mathscr{R}(\tilde\cf))}\end{gathered}
\ee
the group homomorphisms $\mathrm{Tel}(\varpi_a)$ ($a=1,2$) are  surjective with finite kernels.
Then the dashed arrow is a commensurability relation (isomorphism up to finite groups). 
\end{fact}

\begin{proof}
$\varpi_1$, $\varpi_2$ are one-to-one on the orbits of the spherical objects,
and defines a one-to-one correspondence $\rho$ between the orbits in $\mathscr{R}(\tilde\cf)$ and $\mathscr{C}(\cf)$.
Since the telescopic functors depend only on the orbit, this sets a well-defined correspondence between the telescopic functors of the three categories. The solid morphisms in \eqref{ddddiagrmqi} act as\be
L_{A_1}L_{A_2}\cdots L_{A_w}\longmapsto L_{\varpi_a(A_1)}
L_{\varpi_a(A_2)}\cdots L_{\varpi_a(A_w)}.
\ee
In diagram \eqref{ddddiagrmqi} all relation between the generators $L_{A_i}$ which holds in the upstairs group remains valid in the downstairs ones. Therefore the two morphisms
$\mathrm{Aut}(\varpi_a)$ are epi. We already know that their kernels are finite.\end{proof}

\subsubsection{Second case: $\tilde\cf$ is asymptotically-free}\label{assysyfree}

In this case the 2d quantum monodromy $\boldsymbol{H}$ is not semi-simple but rather unipotent, that is, satisfies the minimal equation \eqref{min-pol} with some $m_d$ equal 2.
\medskip

Instead of working out the general case, we study a special class of brane categories $\mathscr{B}=D^b\mathscr{A}$ in strictly asymptotically-free (2,2) models which will suffice for the applications we have in mind. 
The general case is analogous.

\paragraph{Spherical orbits in cluster categories of tame weighted projective lines.}
We consider the categories of branes of $\sigma$-models with target space
a weighted projective line (in the sense of \cite{wpl,lenzinghandbook,ictplect})\footnote{\ We follow the notations and conventions of \cite{shepard}.} $\mathbb{X}(\boldsymbol{p})$,
of weights $\boldsymbol{p}\equiv (p_1,p_2,p_3)$)
with strictly positive Euler characteristic
 \be
 \chi(\mathbb{X}(\boldsymbol{p}))=2-\sum_{i=1}^3\left(1-\frac{1}{p_i}\right)>0.
 \ee
This is the condition for (strict) asymptotic-freedom (AF)  \cite{Cecotti:2011rv,Cecotti:2012va}.
 We have $\mathscr{B}(\boldsymbol{p})=D^b\mathsf{coh}\,\mathbb{X}(\boldsymbol{p})$, the derived category of coherent sheaves on $\mathbb{X}(\boldsymbol{p})$. The Serre duality functor is
 \be
 S\colon X\mapsto X\otimes \omega [1],
 \ee
where $\omega$ is the dualizing sheaf whose degree is, by definition, $-\chi(\mathbb{X}(\boldsymbol{p}))$.
It is negative for a AF $\sigma$-model, and indeed, the degree of $\omega$
is the coefficient of the $\beta$-function of the 2d QFT.

The corresponding cluster category
\be
\mathscr{C}(\boldsymbol{p})\equiv \mathscr{B}(\boldsymbol{p})/\langle S^{-1}[2]\rangle^\Z=D^b\mathsf{coh}\,\mathbb{X}(\boldsymbol{p})/\langle \otimes\, \omega^{-1}[1]\rangle^\Z,
\ee
has been studied in \cite{BKL}. It is the category with the same objects as $\mathsf{coh}\,\mathbb{X}(\boldsymbol{p})$ and Hom-spaces
\be
\mathrm{Hom}_{\mathscr{C}(\boldsymbol{p})}(X,Y)=
\mathrm{Hom}_{\mathsf{coh}}(X,Y)\oplus \mathrm{Ext}^1_{\mathsf{coh}}(X,Y\otimes\omega^{-1}).
\ee 
The corresponding root category
\be
\mathscr{R}(\boldsymbol{p})=D^b\mathsf{coh}\,\mathbb{X}(\boldsymbol{p})\big/[2\Z],
\ee
may be seen as the category $\mathsf{coh}\,\mathbb{X}(\boldsymbol{p})\bigvee \mathsf{coh}\,\mathbb{X}(\boldsymbol{p})[1]$ with morphism spaces ($X,Y\in\mathsf{coh}\,\mathbb{X}(\boldsymbol{p})$)
\be
\begin{gathered}
\mathrm{Hom}_\mathscr{R}(X,Y)=
\mathrm{Hom}_\mathscr{R}(X[1],Y[1])=
\mathrm{Hom}(X,Y),\\
\mathrm{Hom}_\mathscr{R}(X,Y[1])=
\mathrm{Hom}_\mathscr{R}(X[1],Y)=
\mathrm{Ext}^1(X,Y).
\end{gathered}
\ee

\begin{fact} Assume  $\chi(\mathbb{X}(\boldsymbol{p}))>0$.
Let $X\in \mathscr{C}(\boldsymbol{p})$ be spherical \emph{(respectively, let $Y\in\mathscr{R}(\boldsymbol{p})$ belong to a spherical orbit).} Then, as a coherent sheaf, $X$ \emph{(resp.\! $Y$)} is fractional Calabi-Yau with $\hat c=1$, i.e.\!
there is an integer $p$ such that
$S^pX=X[p]$ in $D^b\mathsf{coh}\,\mathbb{X}(\boldsymbol{p})$. The CY objects in $D^b\mathsf{coh}\,\mathbb{X}(\boldsymbol{p})$ are the sheaves of zero-rank; they form an Abelian category, namely a $\mathbb{P}^1$-family of stable tubes all of which but (at most) three are homogeneous. The exceptional tubes have periods $\{p_1,p_2,p_3\}$.
The simples in the tubes are the only objects in spherical orbits.
\end{fact}

\begin{proof} Let $X\in\mathsf{coh}\,\mathbb{X}(\boldsymbol{p})$ seen as an element of the cluster category $\mathscr{C}(\boldsymbol{p})$. One has
\be
\begin{split}
\dim \mathrm{Hom}_{\mathscr{C}(\boldsymbol{p})}(X,X[m])&=\dim \mathrm{Hom}_{\mathscr{C}(\boldsymbol{p})}(X,X\otimes \omega^m)=\\
&=
\dim \mathrm{Hom}_\mathsf{coh}(X,X\otimes\omega^m)+
\dim \mathrm{Hom}_\mathsf{coh}(X,X\otimes\omega^{2-m}).
\end{split}
\ee
If $X$ has positive rank $\rho>0$, i.e.\! it is a bundle, the \textsc{rhs} goes like
\be
\rho\, |m|+\mathrm{const}\quad \text{for }|m|\gg 1,
\ee 
and cannot be spherical. $Y\in\mathscr{R}(\boldsymbol{p})$ is fractional CY iff it is fractional CY in 
$\mathscr{B}(\boldsymbol{p})$.
The rest of the \textbf{Fact} are standard facts about the weighted projective lines, see e.g.\! \cite{wpl,lenzinghandbook,ictplect}.\end{proof}

The simples in the tubes indeed belong to spherical orbits. The simples $\cs_{a,i}$ ($a\in \Z/p_i\Z$) in the $i$-th exceptional stable tube $C_{p_i}$ of period $p_i$ have fractional CY dimension $\frac{p_i}{p_i}$
\be
S^{p_i}\cs_{a,i}=\cs_{a,i}[p_i],
\ee
and setting $S=\tau [1]$ one has 
\be\label{hhazx5}
\tau \cs_{a,i}=\cs_{a+1,i}
\ee
 and
\be
\begin{split}
\dim \mathrm{Hom}(\cs_{a,i},\tau^k \cs_{b,j}[m])&=
\dim \mathrm{Hom}(\cs_{a,i},\cs_{b+k,j}[m])=\\
&=\delta_{ij}\Big(\delta^{(p_i)}_{a,b+k}\,\delta_{m,0}+\delta^{(p_i)}_{a,b+k-1}\,\delta_{m,1}\Big).
\end{split}
\ee
The telescopic functors $L_{\cs_{a,i}}$
depend only on the tube $i$; in $\mathscr{R}(\boldsymbol{p})$ this follows from the fact that the $\cs_{a,i}$ with the same $i$ belong to the same $\tau$-orbit \eqref{hhazx5}. In $\mathscr{C}(\boldsymbol{p})$ one has
\be
L_{\cs_{a+1,i}}=L_{\tau\cs_{a,i}}=L_{\cs_{a,i}[1]}=L_{\cs_{a,i}}.
\ee
The $L_i\equiv L_{\cs_{a,i}}$
 commute between themselves and act on the bundles as shift in the gradings
\be\label{kkaq20v}
X\longmapsto X\otimes \co(\vec x_i)\equiv X(\vec x_i)\quad \text{for }\mathrm{rank}\,X>0
\ee
see \textbf{Theorem 10.8} of \cite{lenzinghandbook}.
Eqn.\eqref{kkaq20v} holds in 
all four categories $\mathsf{coh}\,\mathbb{X}(\boldsymbol{p})$,
$\mathscr{B}(\boldsymbol{p})$, 
$\mathscr{C}(\boldsymbol{p})$,
and $\mathscr{R}(\boldsymbol{p})$.

This shows (for this class of asymptotically-free examples) the isomorphism
\be
\mathrm{Tel}\,\mathscr{C}(\cf)\simeq \mathrm{Tel}\,\mathscr{R}(\tilde\cf).
\ee
In addition, the auto-equivalence group contains the shifts $[k]$, and the permutations of exceptional tubes of same period $p_i$ (see \cite{shepard} for their physical interpretation).

The general asymptotically-free case is expected to be similar. In particular, the analysis may be generalized to the  auto-equivalence groups
generated by objects with spherical half-orbits.

\begin{rem}
The mirror symmetric (2,2) models to the above $\sigma$-models are the LG theories with superpotentials in table
\ref{afff} whose periods may be read in table \ref{mmmatt}.
\end{rem}

\begin{rem}
In the 4d perspective, each exceptional tube of period $p_i>1$ corresponds to a $D_p$ Argyres-Douglas superconformal matter system coupled to $SU(2)$ SYM, as in table \ref{mmmatt}. The $i$-th constituent system, taken in isolation, is described by its own cluster category, which is a Dynkin cluster category of type  $D_{p_i}$. The 4d quantum monodromy of the sub-constituent, $S$, has period $(h(D_{p_i})+2)/2=p_i$ and, as always, has the interpretation of a $U(1)_R$ rotation by $2\pi$ \cite{Cecotti:2010fi} (of the sub-constituent only). On the cluster category of the fully coupled theory $\mathscr{C}(\boldsymbol{p})$ this symmetry operation becomes the telescopic functor $L_i$. Via the chiral anomaly and the Witten effect \cite{effwitten}, it implies a shift in the electric charge ($\equiv$ degree in the math language) of the dyons ($\equiv$ bundles for mathematicians), see eqn.\eqref{kkaq20v}.   
\end{rem}

\begin{rem} The story changes dramatically when $\chi(\mathbb{X}(\boldsymbol{p}))=0$, i.e.\! when the (2,2) $\sigma$-model is conformal. In this case we have other spherical orbits, and we get an auto-equivalence group containing $SL(2,\Z)$ as aspected on general physical grounds. See \cite{shepard} for details.
\end{rem}

\subsection{Explicit matrix realization of the duality group}\label{actgro}

We return to the general case and the explicit $r\times r$ matrix realization of the $S$-duality group discussed in \S.\,\ref{uuuuqmmm82} via its action on the Grothendieck group of $\mathscr{R}(\tilde\cf)$.

We write down the explicit matrix $\boldsymbol{L}_A\equiv\mathbf{bf}(L_A)$ which yields the action of the telescopic auto-equivalence $L_A$ on the Groethendieck group $\Gamma\equiv K_0(\mathscr{R}(\tilde\cf))$
\be
[L_AT_i]=[T_j](\boldsymbol{L}_A)_{ji}.
\ee
From the definition, for an orbit of period $p$
\begin{equation}
[R_AX]=[X]-\sum_{k=1}^p \chi(X,T^kA)\,[T^kA],\qquad
[L_AX]=[X]-\sum_{k=1}^p \chi(T^kA,X)\,[T^kA],\\
\end{equation}
and\footnote{ We write $\boldsymbol{a}=(a_1,\dots,a_r)$,
$\boldsymbol{x}=(x_1,\dots,x_r)$.} 
\be
\begin{cases}[A]=[T_i]a_i\\
[X]=[T_i]x_i\end{cases}\quad \Rightarrow\quad \begin{cases}[T^kA]=(-1)^{nk}\,[T_j](\boldsymbol{H}^k\boldsymbol{a})_{j}\\
\chi(T^kA,X)=(-1)^{nk}\,\boldsymbol{a}^t (\boldsymbol{H}^t)^k\boldsymbol{E}\boldsymbol{x}.\end{cases}
\ee
Then
\begin{align}
(\boldsymbol{L}_A)_{ij}=\delta_{ij}-\sum_{k=1}^p\big(\boldsymbol{H}^k\boldsymbol{a}\big)_i\,\big(\boldsymbol{a}^t\boldsymbol{E}\boldsymbol{H}^{-k}\big)_j\\
(\boldsymbol{R}_A)_{ij}=\delta_{ij}-\sum_{k=1}^p\big(\boldsymbol{H}^k\boldsymbol{a}\big)_i\,\big(\boldsymbol{a}^t\boldsymbol{E}\boldsymbol{H}^{1-k}\big)_j
\end{align}
From these expression it is obvious that $\boldsymbol{L}_A$ and $\boldsymbol{R}_A$ commute with $\boldsymbol{H}$; using the fact that $A$ belongs to a spherical orbit
\be\label{jjjaq1c}
\boldsymbol{a}^t\boldsymbol{E}\boldsymbol{H}^k\boldsymbol{a}=\delta^{(p)}_{k,0}+\delta^{(p)}_{k,1}
\ee
and $\boldsymbol{H}=\boldsymbol{E}^{-1}\boldsymbol{E}^t$, it is easy to check that
\be
\boldsymbol{L}_A\boldsymbol{R}_A=\boldsymbol{1},\qquad\Rightarrow\qquad
\boldsymbol{L}_A,\ \boldsymbol{R}_A\in GL(r,\Z).
\ee
We shall write $\mathbf{Tel}(\mathscr{R}(\tilde\cf))\subset GL(r,\Z)$ for the concrete group of matrices representing the group $\mathrm{Tel}(\mathscr{R}(\tilde\cf))$ on the charge lattice $\Gamma$.

\subsection{General properties of the matrix $\boldsymbol{L}_A$}
$\boldsymbol{L}_A$, $\boldsymbol{R}_A$ satisfy the equations
\begin{align}\label{tttttqzbv}
\big(\boldsymbol{L}_A-\boldsymbol{1}\big)\big(\boldsymbol{L}_A+\boldsymbol{H}^{-1}\big)&=0\\
\big(\boldsymbol{R}_A-\boldsymbol{1}\big)\big(\boldsymbol{R}_A+\boldsymbol{H}\big)&=0
\end{align}
which show that $\boldsymbol{L}_A$, $\boldsymbol{R}_A$ preserve $\boldsymbol{E}$. Note that the equation satisfied by $\boldsymbol{L}_A$ is independent of the particular spherical orbit $A$.

\emph{A priori,} the matrix $\boldsymbol{L}_A$ may have a 
non-trivial Jordan blocks in two cases:
\begin{itemize}
\item[A)]
 in correspondence to a non-trivial Jordan block of $\boldsymbol{H}$;
 \item[B)] 
associated with the eigenvalue $1$ in the eigenspace $-1$ of $\boldsymbol{H}$. 
\end{itemize}
Below we shall see that A) cannot happen (unless the $\boldsymbol{H}$ block is associated to the eigenvalue $-1$). We recall that the size of the Jordan blocks is at most 2.

Then the only non-trivial Jordan blocks of $\boldsymbol{L}_A$ appear in the $(-1)$-eigenspace of the 2d monodromy. Let $\Gamma_{-1}$
be the lattice
\be
\Gamma_{-1}\equiv \Big\{\boldsymbol{x}\in \Gamma\;\Big|\; \boldsymbol{H}\boldsymbol{x}=-\boldsymbol{x}\Big\}
\ee
which is equipped with a non-degenerate skew-symmetric integral form induced from $\Omega$. 
The operators $\boldsymbol{L}_A\big|_{\Gamma_{-1}}$ are unipotent of index 2
\be\label{oooaxzwq}
\left(\boldsymbol{L}_A\big|_{\Gamma_{-1}}-1\right)^2=0.
\ee

\subsubsection{$\mathrm{Tel}(\mathscr{R}(\tilde\cf))$ acts on $\Gamma_\text{flavor}$ as $\mathrm{Weyl}(L)$} \label{wweeylassdwa}

We already know that the action of $\mathrm{Tel}(\mathscr{R})$ on the flavor lattice factors through a finite group.
Here we show that it is a crystallographic  reflection group; by the Coxeter classification \cite{finiterefl} this means that the image
\be
\mathrm{Tel}(\mathscr{R}(\tilde\cf)))\to \mathrm{Aut}(\Gamma_\text{flavor})\simeq GL(f,\Z)
\ee
is
the Weyl group of a finite-dimensional semi-simple Lie algebra $L$.

\begin{fact}\label{ffffmx} For all spherical orbits $A$,
the restriction of $\boldsymbol{L}_A$ to the flavor lattice $\Gamma_\text{flavor}$, $\boldsymbol{L}_A\big|$, is an \emph{involution} $(\boldsymbol{L}_A\big|)^2=\boldsymbol{1}$, and in facts $\boldsymbol{L}_A\big|$ is a \emph{reflection} $\sigma_A$, meaning that, in addition, the matrix $(\boldsymbol{1}-\boldsymbol{L}_A)\big|$ has precisely rank 1. In particular, the restriction to $\Gamma_\text{flavor}$ of the telescopic functors of
an $(A_m)$-configuration of spherical orbits yields an action of $\mathrm{Weyl}(A_m)\equiv\mathfrak{S}_{m+1}$.
\end{fact}

\begin{proof}
$\boldsymbol{x}\in \Gamma_\text{flavor}$ iff and only if $\boldsymbol{H}\boldsymbol{x}=
\boldsymbol{x}$. Then, for $\boldsymbol{x}\in \Gamma_\text{flavor}$
\be
0=\big(\boldsymbol{L_A}-\boldsymbol{1}\big)\big(\boldsymbol{L}_A+\boldsymbol{H}^{-1}\big)\boldsymbol{x}\equiv \big(\boldsymbol{L}_A^2-\boldsymbol{1}\big)\boldsymbol{x},
\ee
so $\boldsymbol{L}_A\big|$ is an involution. On the other hand,
\be
\big(\boldsymbol{1}-\boldsymbol{L}_A\big)\Big|= \boldsymbol{v}\otimes \boldsymbol{v}^t \boldsymbol{E}\Big|,\qquad\text{where}\quad \boldsymbol{v}=\frac{1}{p}\sum_{k=1}^p\boldsymbol{H}^k\boldsymbol{a}.
\ee
$(\boldsymbol{1}-\boldsymbol{L}_A)\big|=0$ iff $\boldsymbol{v}=0$, which is impossible since  eqn.\eqref{jjjaq1c} yields
$\boldsymbol{a}^t\boldsymbol{E}\boldsymbol{v}=2/p$.\end{proof}

\begin{rem}
This results agrees with the physical picture in \S.\,\ref{weylll}, see eqn.\eqref{tttttmmmqq9}. The fact that the Weyl group of the Lie group $L$ is a factor of the image of $\mathrm{Tel}(\mathscr{R})$ does not imply that $L$ is a factor of the flavor group. Besides the ambiguity arising from the Abelian sector, encoded in the group $O(h,\Z)_C$, we have to keep into account the automorphisms of the Dynkin graphs. For instance, in the case of
$SU(2)$ with $N_f=4$ the $S$-duality acts on the $SO(8)$ flavor charges as the group \cite{shepard}
\be
\mathrm{Aut}(SO(8))\ltimes \mathrm{Weyl}(SO(8))\simeq \mathrm{Weyl}(F_4).
\ee
In the same fashion, for $n\neq4$
\be
\mathrm{Aut}(SO(2n))\ltimes \mathrm{Weyl}(SO(2n))\simeq \mathrm{Weyl}(SO(2n+1)).
\ee
\end{rem}

\subsubsection{Complex reflection groups}

Exploiting the special form \eqref{min-pol} of the minimal polynomial for $\boldsymbol{H}$, modulo commensurability we can write
\be\label{xxx1834}
\Gamma\thickapprox \Gamma_\text{flavor}\oplus\left(\bigoplus_{d\in D} \Gamma_d\right)
\ee
where
\be
\Gamma_d\equiv \Big\{ \boldsymbol{x}\in\Gamma\;\Big|\; \Phi_d(\boldsymbol{H})^{m_d}\,\boldsymbol{x}=0\Big\}
\ee
The decomposition \eqref{xxx1834} is 
\emph{orthogonal} for the Euler form $\boldsymbol{E}$, that is,
if $\boldsymbol{x}_d\in \Gamma_d$
and $\boldsymbol{y}_{d^\prime}\in \Gamma_{d^\prime}$
\be
\boldsymbol{x}^t_d\,\boldsymbol{E}\,\boldsymbol{y}_{d^\prime}=0\quad\text{unless }d=d^\prime.
\ee

All auto-equivalence preserves the sub-lattices in the \textsc{rhs} of\eqref{xxx1834} individually. We have already discussed the action of the telescopic functors on the first summand $\Gamma_\text{flavor}$,
and also on $\Gamma_2$ (see argument around eqn.\eqref{oooaxzwq}). Now we focus on one particular $\Gamma_d$ with $d\geq 3$. 
We write $\boldsymbol{H}_d$ (resp.\! $\boldsymbol{R}_{A,d}$) for the integral matrix obtained by restricting $\boldsymbol{H}$ (resp.\! $\boldsymbol{R}_A$) to $\Gamma_d$. 
We write 
$$\mathbf{Tel}_d\subset GL(\mathrm{rank}\,\Gamma_d,\Z)$$
for the matrix group generated by the $\boldsymbol{R}_{A,d}$ of all objects $A$ with spherical orbits.

Let $\mathbb{Q}[\zeta]$ be the cyclotomic field of a primitive $2d$--th root of unity $\zeta$.
$\boldsymbol{H}_d$ (resp.\! $\boldsymbol{R}_{A,d}$) may be set in Jordan canonical form over $\mathbb{Q}[\zeta]$ with eigenvalues
of the form $\zeta^{2\ell}$ (resp.\!
$1$ and $-\zeta^{2\ell}$) with\footnote{ $(\Z/d\Z)^\times$ is the group of unities in the ring $\Z/d\Z$. By definition $|(\Z/d\Z)^\times|=\phi(d)$.} $\ell\in(\Z/d\Z)^\times$, all having the same multiplicity and Jordan structure.
We consider the $\mathbb{Q}[\zeta]$-space
$V_d= \Gamma_d\otimes \mathbb{Q}[\zeta]$. One has
\be\label{uuuqqan}
V_d=\bigoplus_{\ell\in(\Z/d\Z)^\times} W_{d,\ell},\qquad\text{where}\quad W_{d,\ell}=\Big\{v\in V_d\;\Big|\; (\boldsymbol{H}-\zeta^{2\ell})^2v=0\Big\}.
\ee
Let 
$Z$ be the operator which on $W_{d,\ell}$ acts as $\zeta^{-\ell}$.

\begin{lem}
The form (linear on the first argument, anti-linear in the second one)
\be\label{ttttghscd}
H(v,w)= \chi(Z v,w^*), \qquad v,w\in V_d,
\ee
is Hermitian and $\boldsymbol{H}_d$ invariant. It decomposes over $\mathbb{Q}[\zeta]$ into a direct sum of Hermitian forms $H_{d,\ell}$ on each space $W_{d,\ell}$. The dimension of the radical of $H_{d,\ell}$ and its signature are independent of $\ell\in(\Z/d\Z)^\times$.

\end{lem}

The decomposition \eqref{uuuqqan} yields an embedding
\be\label{uuuuqmmsd}
\mathbf{Tel}_d\to \prod_{\ell\in(\Z/d\Z)^\times} U\big(H_{d,\ell}, \mathbb{Q}[\zeta]\big)
\ee
where $U\big(H_\ell, \mathbb{Q}[\zeta]\big)$ is the group of ``unitary'' matrices preserving the Hermitian form $H_{d,\ell}$. The images in the several factors in the \textsc{rhs} of\eqref{uuuuqmmsd} are all conjugate\footnote{\ See \textbf{VI.(1.9)} of ref.\!\cite{frolich}.} under $\mathrm{Gal}(\mathbb{Q}[\zeta]/\mathbb{Q})\simeq (\Z/2d\Z)^\times$; it is enough to consider just one image, say the one with $\ell=1$, $\mathbf{Tel}_{d,1}$.
\medskip

\paragraph{Finiteness conditions.} If $H_{d,1}$ is definite, $U(H_{d,1},\C)$ is compact, and then $\mathbf{Tel}_{d,1}$ is a \emph{finite}  group. This is guaranteed to happen in two cases:
\begin{itemize}
\item if $w_d=1$ (cfr.\! eqn.\eqref{char-pol});
\item for all $d$ such that (cfr.\! eqn.\eqref{veryrelevant})
\be\label{veryrelevant2}
\frac{1}{d} < 1-\frac{\hat c}{2}.
\ee
\end{itemize}

For all $\mathbf{Tel}_{d,\ell}$, finite or infinite, we have:

\begin{fact} Let $d\geq3$.
The restriction $\boldsymbol{R}_{A,d,\ell}$ of $\boldsymbol{R}_{A}$ in the sub-space $W_{d,\ell}$  
acts as a complex reflection, that is,
is semi-simple with eigenvalues
\be\label{hhhasmo}
-\zeta^{2\ell},\ 1,\ 1,\ 1,\ \cdots,\ 1.
\ee
Thus, the image in the $\ell$--th factor in eqn.\eqref{uuuuqmmsd}, $\mathbf{Tel}_{d,\ell}$, is a \emph{complex reflection group.}
\end{fact}

\noindent\textsc{Proof.}
One has
\be
\big(\boldsymbol{1}-\boldsymbol{L}_A\big)\Big|_{W_{d,\ell}}= \boldsymbol{v}_{d,\ell}\otimes \boldsymbol{v}_{d,\ell}^t \,\boldsymbol{E},\qquad\text{where}\quad \boldsymbol{v}_{d,\ell}=\frac{1}{d}\sum_{k=0}^{d-1} \zeta^{-2k\ell}\boldsymbol{H}^k\boldsymbol{a},
\ee
and $v_{d,\ell}\neq0$ since by \eqref{jjjaq1c}
\be
d\,\boldsymbol{a}^t \,\boldsymbol{E}\,\boldsymbol{v}_{d,\ell}=1+\zeta^{-2\ell}\neq0
\ee
since $\zeta^{-2}=e^{-2\pi i/d}$, $d\neq 1,2$, and $(\ell,d)=1$.\hfill $\square$
\medskip

Thus the matrices $\boldsymbol{R}_{A,d,\ell}$ are complex reflections\footnote{\ A matrix is a \emph{complex reflection} iff all its eigenvalues \emph{but one} are equal to 1.}; a group generated by complex reflection is called a \textit{a complex reflection group.} We conclude that each $\mathbf{Tel}_{d}$ is a complex reflection group (finite or infinite). The reflection group
$\mathbf{Tel}_d$ is crystallographic with respect to the 
cyclotomic integers $\co_{\zeta^2}\subset\mathbb{Q}[\zeta^2]$; indeed, 
\be
\Gamma_d\simeq \co_{\zeta^2}^{\mathrm{rank}\,\Gamma_d/\phi(d)}\qquad \text{as $\Z$-modules}.
\ee

\paragraph{Finite reflection groups.} The \emph{finite} complex reflection groups have been fully classified by Shephard-Todd \cite{shep-tood}. Then, whenever $H_d$ is definite, we reduce the problem of determining $\mathbf{Tel}_d$ to a comparison with a known list (this is the strategy used in \cite{shepard}).

The complete list of finite reflection groups is given by two classes: \textit{i)} an infinite sequence $G(d,e,r)$, depending on three integers $d$, $e$, $r$, which may be seen as a cyclotomic generalization of the symmetric groups $\mathfrak{S}_r=G(1,1,r)$; \textit{ii)} 34 exceptional groups denoted $G_4, G_5,\cdots, G_{37}$ whose matrices have sizes $\leq 8$.

Shephard-Todd groups are most conveniently presented as quotients of generalized braid groups, i.e.\!
they are generated by a set of complex reflections $t_A$ which satisfy two kinds of relations \cite{stbraid,broue}:
\begin{itemize}
\item[A)] braid relations of the Coxeter form
\be
\overbrace{t_At_Bt_At_B\cdots}^{s\ \text{factors}}= 
\overbrace{t_Bt_At_Bt_A\cdots}^{s\ \text{factors}};
\ee
\item[B)]
order relations for the generators $t_A$  of the form 
\be
t_A^{d_A}=1.
\ee
\end{itemize}
 The braid presentation is especially suited for our purposes:
we identify the $t_A$ with the matrices $\boldsymbol{L}_A$ restricted to a subspace $W_{d,\ell}$ with definite Hermitian form \eqref{ttttghscd}.
The $t_A$ inherit from the telescopic functors their braiding relations\eqref{braid2}, while from eqn.\eqref{hhhasmo} we get the order relations
\be
(t_A)^{2 d/\gcd(2,d)}=1.
\ee
Note that the exponent is independent of $A$. Inverting the logic, if the braid and order relations satisfied by the group $\mathbf{Tel}_{d,\ell}$ appear in the Shephard--Todd list, we conclude that  $\mathbf{Tel}_{d,\ell}$ is finite.
The relations are best written in terms of a graph \cite{stbraid,broue} which generalizes the usual Dynkin graphs of the real reflection groups.

\subsection{Spherical (half) orbits in
$\mathscr{R}(G)$}\label{jggss90}

For later use we need to classify the spherical (half) orbits in the root category of a Dynkin graph $G$ of $ADE$ type.
In this sub-section we consider \textit{$1$-spherical $\tau$-orbits,} i.e.\! in the general framework of
\S\S.\ref{TSS1}-\ref{TSS3}
we take $n=1$ and $T=S[-n]\equiv\tau$ (the AR translation \cite{assem}).
We stress that no \emph{new}  auto-equivalence of $\mathscr{R}(G)$ may arise from such orbits. Indeed,

\begin{fact}\label{kkk123}
Let $A\in\mathscr{R}(G)$ be an object with a spherical (half-) $\tau$-orbit.
Then
\be\label{Dorbitscirc-1}
L_A\simeq \begin{cases}\tau^{-1} & G\neq D_{2n}\\
\varepsilon_A\tau^{-1} & G=D_{2n},
\end{cases}
\ee
where $\varepsilon_A$ is a non-trivial involution, $\varepsilon_A^2=\mathrm{Id}$. 
\end{fact}

\begin{proof} In the Dynkin case\footnote{\ As 2d (2,2) theories these are the minimal $\cn=2$ SCFTs.}  the 2d monodromy is minus the Coxeter of $G$, $\boldsymbol{H} =-\boldsymbol{\Phi}$. With the exception of the eigenvalue $-1$ for $D_{2n}$,
all eigenspaces of $\boldsymbol{\Phi}$ have dimension $1$.  For $G\neq D_{2n}$, it follows from eqn.\eqref{hhhasmo} that  $\boldsymbol{L}_A=-\boldsymbol{H}^{-1}=\boldsymbol{\Phi}^{-1}$. Then $\tau L_A$ fixes all Grothendieck classes and is equivalent to the identity. For
$G=D_{2n}$, $\boldsymbol{\Phi}\boldsymbol{L}_A$ has all eigenvalues $+1$ but one which is $-1$.
Then
$(\tau L_A)^2$ fixes all classes so that $\tau L_A\not\simeq \mathrm{Id}$ and
$(\tau L_A)^2\simeq \mathrm{Id}$. \end{proof}

Our interest for the spherical (half)orbits in $\mathscr{R}(G)$ is purely technical:
in future sections we shall ``twist'' together spherical orbits of  several Dynkin graphs to get non-trivial new auto-equivalences. The reader may prefer to skip this subsection.

\begin{figure}
\begin{align*}
&\vec A_{n-1}\colon
&&\xymatrix{1\ar[r] & 2\ar[r] &3\ar[r] &4\ar[r]&\cdots\ar[r]& n-2\ar[r]& n-1}\\
&\vec D_n\colon&&\begin{gathered}
\xymatrix{&&&&&s
\\
n-2&n-3\ar[l]&\cdots\ar[l] &3\ar[l] &2\ar[l] &1\ar[l]\ar[u]\ar[r]&c}
\end{gathered}\\
&\vec E_n\colon&&\begin{gathered}
\xymatrix{&&n-1\ar[d]&&&
\\
1\ar[r]&2\ar[r]&n&n-2\ar[l]&\cdots\ar[l] &4\ar[l]&3\ar[l]}
\end{gathered}
\end{align*}
\caption{\label{refDyqi}Reference orientations of Dynkin quivers. The node numbers are chosen so that for all arrows $\psi$ one has $t(\psi)>s(\psi)$.}
\end{figure}

\medskip

For concreteness we fix a reference orientation of the Dynkin quivers.
$\vec G$ will always mean a Dynkin quiver of type $G$ with the reference orientation. The reference quivers (with numbered nodes) are shown in figure \ref{refDyqi}. We write $S_i$ for the simple module of $\C\vec G$ with support at the $i$--th node, $P_i$ for its indecomposable projective cover \cite{assem}, and use the standard tilting object
\be
\ct=\bigoplus_{i\in \vec G} P_i = \C\vec G_{\C\vec G}.
\ee 
The Euler matrix in the $[P_i]$ basis reads
\be
\boldsymbol{E}_{ij}=\dim \mathrm{End}(P_i,P_j)=\dim (P_j)_i.
\ee

The indecomposable objects of 
\be\label{pppbbn4x}
\mathscr{R}(G)\simeq \mathsf{mod}\,\C \vec G\vee (\mathsf{mod}\,\C \vec G)[1]
\ee
 are in one-to-one correspondence with the roots $\alpha\in \Delta(G)$ of the Lie algebra $G$ and will be labelled by the roots $\{X_\alpha\}_{\alpha\in\Delta(G)}$. Positive roots correspond to modules of $\C \vec G$, while $X_{-\alpha}=X_\alpha[1]$. 
$S_i\equiv X_{\alpha_i}$, where $\alpha_i$ is $i$--th simple root.

\subsubsection{Explicit form of the root category $\mathscr{R}(G)$}

It is convenient to give an explicit realization of the root category of the Dynkin graph $G$ in terms of its Auslander-Reiten (AR) quiver \cite{assem}; it is just a $h(G)$-periodic version of the AR quiver for the derived category, constructed by Happel \cite{happel}. We write $G$ for the \emph{opposite} quiver of $\vec G$ (i.e.\! the quiver with all arrows inverted). 

Te AR quiver of $\mathscr{R}(G)$ is given by
\be\label{iiiiiiqwaz}
\Z G/\tau^{h(G)},
\ee
that is, the quiver whose nodes $v$ are pairs $(k,i)$
with $k\in\Z/h(G)\Z$ and $i\in G$.
An arrow $\psi\colon i\to j$ in $G$ yields $2h(G)$ arrows in $\Z G/\tau^{h(G)}$
\be
\begin{aligned}
(k,\psi)\colon& (k,i)\to (k,j)\\
\sigma(k, \psi)\colon& (k-1,j)\to (k,i)\qquad k\in\Z\big/h(G)\Z.
\end{aligned}
\ee 
The operation $\sigma$ is extended to all arrows by the rule that
\be
\sigma^2(k,\psi)=(k-1,\psi).
\ee
$\tau$ acts as $(k,i)\to (k-1,i)$ and $(k,\psi)\to (k-1,\psi)$ on all nodes and arrows. Note that for an arrow $\psi\colon u\to v\equiv (k,i)$
the path  $\psi\sigma(\psi)$ has source in the node $\tau v=(k-1,i)$.
The \emph{mesh} at $v$ is
\be
r_v=\sum_{\psi\colon t(\psi)=v} \psi\sigma(\psi).
\ee
Each node $v$ in the AR quiver \eqref{iiiiiiqwaz} represents an  isoclass of indecomposable objects of $\mathscr{R}(G)$, the projective modules $P_i$ being associated to the nodes $(1,i)$. The morphism space  $\mathrm{Hom}(v,u)$ is the vector space over the paths connecting $v$ and $u$ in the AR quiver, modulo the ideal generated by all meshes $r_v$. 

\begin{exe}
The AR quiver for $\mathscr{R}(A_3)$ is shown in figure \ref{jjja194}, and (half) the AR quiver of $\mathscr{R}(D_4)$ in figure \ref{jjja196}.
\end{exe}

\begin{figure}
\begin{tiny}
$$\xymatrix{&&111\ar[dr] && -001\ar[dr]
&& -010\ar[dr]^\alpha  &&-100\ar[dr] && 111\\
&011\ar[ur]\ar[dr]&&110\ar[ur]\ar[dr]_\psi&&-011\ar[ur]^{\sigma(\alpha)}\ar[dr]_{\sigma(\beta)}&&-110\ar[ur]\ar[dr]&&011\ar[ur]&\\
001\ar[ur]&&010\ar[ur]_{\sigma(\psi)}&&100\ar[ur]&&-111\ar[ur]_\beta&&001\ar[ur]&&
}
$$
\end{tiny}
\caption{\label{jjja194}The AR quiver of $\mathscr{R}(A_3)$. Objects are labelled by their Grothendieck class in the root lattice of $G$. The objects on the right are cyclically identified with the corresponding ones on the left.
$\tau$ acts by horizontal translation to the left. 
Examples of mesh relations are
$\psi\sigma(\psi)=0$ and $\alpha\sigma(\alpha)+\beta\sigma(\beta)=0$. The $\tau$ orbit of $P_2=011$ is actually twice a half-orbit.}
\end{figure}
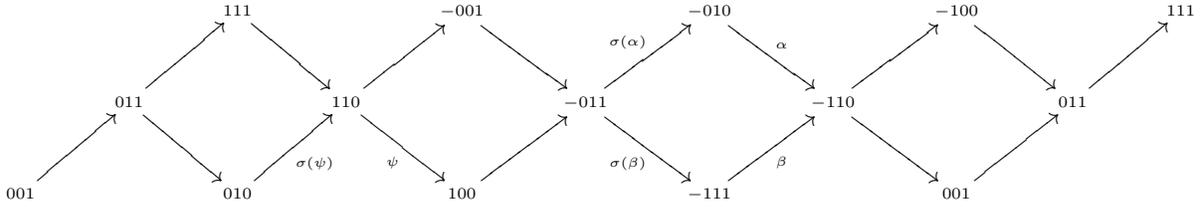

\subsubsection{Derived Picard groups
($S$-duality for Argyres-Douglas)}\label{mmmnnbaq}

 We start by reviewing the derived Picard groups of the Dynkin algebras \cite{picard}, i.e.\! the groups $\mathrm{Aut}(D^b\mathsf{mod}\,\C\vec G)$.
They are generated by the AR translation $\tau\equiv S[-1]$, the shift $[1]$, and the automorphisms of the quiver $\mathrm{aut}(\vec G)$. The relations between these generators are listed in ref.\cite{picard}.
After reducing to the root category
\be
\mathscr{R}(G)\equiv D^b\mathsf{mod}\,\C\vec G\big/[2\Z],
\ee
they take the form in table \ref{derivepicX}. In the physical terms, this table shows the $S$-duality group of the Argyres-Douglas models of type $G$.

We note a special case. For $G=A_{2n+1}$   the $\tau$-orbit of $P_{n+1}$  (nodes numbered as in figure \ref{refDyqi}) is twice an half-orbit, see figure \ref{jjja194}.

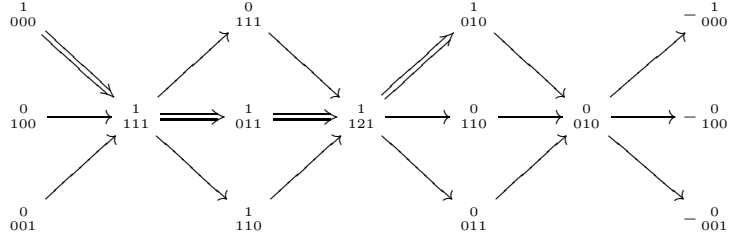
\begin{figure}
\begin{tiny}
$$\xymatrix{{\phantom{1}1\phantom{1}\atop000}\ar@{=>}[dr]&&{\phantom{1}0\phantom{1}\atop111}\ar[dr] && {\phantom{1}1\phantom{1}\atop010}\ar[dr]
&& -{\phantom{1}1\phantom{1}\atop000}\\
{\phantom{1}0\phantom{1}\atop100}\ar[r]&{\phantom{1}1\phantom{1}\atop111}\ar[ur]\ar[dr]\ar@{=>}[r]&{\phantom{1}1\phantom{1}\atop011}\ar@{=>}[r]&{\phantom{1}1\phantom{1}\atop121}\ar@{=>}[ur]\ar[dr]\ar[r]&{\phantom{1}0\phantom{1}\atop110}\ar[r]&{\phantom{1}0\phantom{1}\atop010}\ar[ur]\ar[dr]\ar[r]&-{\phantom{1}0\phantom{1}\atop100}\\
{\phantom{1}0\phantom{1}\atop001}\ar[ur]&&{\phantom{1}1\phantom{1}\atop110}\ar[ur]&&{\phantom{1}0\phantom{1}\atop011}\ar[ur]&&-{\phantom{1}0\phantom{1}\atop001}
}
$$
\end{tiny}
\caption{\label{jjja196}The first half of the AR quiver of $\mathscr{R}(D_4)$; the second half repeats up to a shift (i.e.\! an overall minus, in terms of Grothendieck classes). The double arrows represent a non-trivial path
from $P_s$ to $\tau^{-2}P_s\equiv \tau^{1-h(D_4)/2}P_s$ as described in the text.}
\end{figure}

\subsubsection{Spherical orbits and half-orbits in $\mathscr{R}(G)$} 
We are interested in periodic orbits
and half-orbits which are spherical.
Comparing with the table, we look for full orbits for: \textit{i)} $\vec A_{N-1}$, \textit{ii)} $\vec D_{2n+1}$ and objects such that
$\theta A\not\simeq A$, and \textit{iii)} $\vec E_6$ for $\theta A\not\simeq A$. In all other cases we consider half-orbits.
Two (half)orbits will be called equivalent if they are exchanged by the shift $[1]$.

\medskip

We recall that
an object $A\in\mathscr{R}(G)$ belongs to a spherical $\tau$-orbit of period (resp.\! half-period) $p$ iff $\tau^p A\simeq A$ (resp.\! $\tau^p A\simeq A[1]$) and
\begin{equation}
\sum_{k=0}^{p-1}\Big(\dim \mathrm{Hom}(A,\tau^k A)+
\dim \mathrm{Hom}(A,\tau^k A[1])\Big)= 2.
\end{equation}

In the case of a full orbit\footnote{\ A full orbit is necessarily mapped into a distinct orbit by $[1]$, whereas a half-orbit in closed under shifts.} this condition may be stated in terms of the AR quiver as the requirement that there is no non-zero path connecting the node of $A$ with another node on the same $\tau$-orbit, i.e.\! in figure \ref{jjja194} no non-trivial path connecting two nodes at the same horizontal level. It is clear from the figure that this happens precisely for the upper and lower levels since the mesh relations set all paths to zero in this case; for intermediate levels this cannot happen. The two spherical orbits in the figure are interchanged by $[1]$. 

In the case of a \emph{half}-orbit of 
half-period $p$, $\tau^{-p} A=A[1]$ and so
\be\label{yyyyyqaz}
\begin{split}
\dim\, &\mathrm{Hom}(A,\tau^{1-p}A)=
\dim\mathrm{Hom}(A,\tau A[1])=\\
&=\dim \mathrm{Hom}(\tau A,\tau A)=
\dim\mathrm{Hom}(A,A)=1.
\end{split}
\ee 
Then the \emph{half}-orbit is spherical iff the only non-trivial paths in the AR quiver which connect two distinct point in the orbit are one path from each 
$A$ to its translate $\tau^{1-p}A$ (with $p=h(G)/2$). An example of such path for $G=\vec D_4$ is shown in figure \ref{jjja196}.

\begin{table}
\begin{center}
\begin{tabular}{ccc@{\hspace{25pt}}lc}\hline\hline
$G$ & $\mathrm{aut}(\vec G)$ & $\mathrm{Aut}(\mathscr{R}(G))$ & relations\\\hline
$A_{N-1}$ & $1$ & $\Z_N\times \Z_2$ & $\tau^N=\mathrm{Id}$ & $[2]=\mathrm{Id}$\\
$D_4$ & $\mathfrak{S}_3$ & $\Z_6$
& $\tau^3=[-1]$ & $[2]=\mathrm{Id}$\\
$D_{2n}$ $n>2$ & $\mathfrak{S}_2$ & $\mathfrak{S}_2\times \Z_{2(2n-1)}$
& $\tau^{2n-1}=[-1]$ & $[2]=\mathrm{Id}$\\
$D_{2n+1}$ $n\geq2$ & $\mathfrak{S}_2$ & $\mathfrak{S}_2\times \Z_{4n}$
& $\tau^{2n}=\theta[-1]$ &$[2]=\mathrm{Id}$\\
$E_6$ & $\mathfrak{S}_2$ & $\mathfrak{S}_2\times \Z_{12}$
& $\tau^{6}=\theta[-1]$ &$[2]=\mathrm{Id}$\\
$E_7$ & $1$ &  $\Z_{18}$
& $\tau^{9}=[-1]$& $[2]=\mathrm{Id}$\\
$E_8$ & $1$ &  $\Z_{30}$
& $\tau^{15}=[-1]$& $[2]=\mathrm{Id}$\\\hline\hline
\end{tabular}
\caption{\label{derivepicX}Auto-equivalences of  $\mathscr{R}(G)$ \cite{picard}. The auto-equivalences of $D^b(\mathsf{mod}\,\C G)$ are obtained by omitting the relation $[2]=\mathrm{Id}$ in the last column. $\theta$ is the element of order 2 in $\mathrm{aut}(\vec G)$. }
\end{center}
\end{table}

Since the (half)period $p>2$ for all $G$, 
this argument shows the 

\begin{lem}
A necessary condition for an object $A\in\mathscr{R}(G)$ to have a spherical (half) orbit is that its AR triangle has an indecomposable middle term, i.e.\!
\be
\to \tau A\to M\to A\to \qquad \text{with }M\ \text{indecomposable.}
\ee
\end{lem}

This happens only for nodes in $\Z G/\tau^{h(G)}$ of the form $(\ast,i_1)$ with $i_1$ a node of valency 1 in $G$. If there are no nodes of valency $3$ in the graph, the orbits in $\Z G/\tau^{h(G)}$ generated by valency 1 vertices of $Q$ are spherical: indeed, no two nodes in the same valency-1 orbit may be connected by a non-trivial path, since meshes of nodes of valency 2 just enforce commutativity (i.e.\! path independence of the morphism, which is then easily seen to be zero). However, if $G$ contains
a full subquiver  of the form
(numbers inside squares denote the valency of the node in the total quiver $G$)
\be
\text{\fbox{1}}\longrightarrow\overbrace{\text{\fbox{2}}\longrightarrow\text{\fbox{2}}\longrightarrow\cdots
 \longrightarrow\text{\fbox{2}}}^{m\ \text{\fbox{2}\,'s}}
 \longrightarrow \text{\fbox{3}}
 \longrightarrow \text{\fbox{?}}
\ee 
we may construct a non-zero path
between $A\equiv (\ast,\text{\fbox{1}})$ and $\tau^{-(m+2)}A\equiv (\ast+m+2,\text{\fbox{1}})$ which factors through to a node in the \fbox{?}-orbit.
In the case of complete orbits,
$A$ has a spherical orbit only if
\be\label{tttmsx64}
m+2\geq h(G),
\ee
 In presence of a 3-valent vertex, the only complete orbits are those of $P_s$, $P_c$ for $D_{2n+1}$ and $P_1$, $P_5$ for $E_6$, and none of them satisfies the inequality \eqref{tttmsx64}. 
In view of \eqref{yyyyyqaz}, for \emph{half}-orbits the inequality \eqref{tttmsx64} gets replaced by the equality
\be
m+2=\frac{h(G)}{2}-1.
\ee

In conclusion,

\begin{fact}[The $G=A_{N-1}$ case]\label{aNspherorb}
Write $\theta$ for the \emph{highest root} of $A_{N-1}$.
In $\mathscr{R}(\vec A_{N-1})$ there are two (equivalent)  spherical orbits interchanged by the shift $[1]$; their period is $N$. Writing the objects in the first spherical $\tau$-orbit as
$A_i\equiv \tau^{i}A_0$, with $i\in \Z/N\Z$, we have  
\begin{align}
A_i&= \begin{cases}
X_{-\theta} &\text{for } i=0\\
X_{\alpha_i}
 &\text{for } 1\leq i\leq N-1.\end{cases}\label{eeeqqz}
\end{align}
\end{fact}

\begin{fact}[The $G=D_r$ case]\label{Dorbitscirc}
Up to shift by $[1]$:
\begin{itemize}
\item[1)] for $r>4$ there is a unique spherical half-orbit generated by the simple $S_v$ with support at the fundamental representation node ($S_v\equiv S_{n-2}$ in  figure \ref{refDyqi}). To describe the spherical half-orbit we define the roots
\be
\beta=\text{\begin{tiny}$\left[{\phantom{11\cdots1}0\phantom{1}\atop11\cdots 110}\right]$\end{tiny}},\qquad \gamma=\text{\begin{tiny}$\left[{\phantom{11\cdots1}1\phantom{1}\atop00\cdots 011}\right]$\end{tiny}},\ee
and let $\alpha_i$ be the simple roots numbered as in figure \ref{refDyqi}).
Then the spherical half-orbit is $A_i=\tau^i A_0$ with
$\tau^{n-1}A_i\simeq A_i[1]$ and
\be
A_i=\begin{cases} X_{\alpha_{n-2}} & i=0\\
X_{-\beta} & i=1\\
X_{-\gamma}& i=2\\
X_{-\alpha_{i-1}} & 3\leq i\leq n-1.
\end{cases}
\ee
\item[2)] for $\vec D_4$ there are three spherical half-orbits generated, respectively, by the three peripheral simples $S_v$, $S_s$, and $S_c$ permuted by the automorphism $\mathfrak{S}_3$ of the quiver. 
One has ($\alpha,\beta=v,s,c$)
\be\label{uuutvcm}
\sum_{k=1}^3\Big(\dim \mathrm{Hom}(S_\alpha, \tau^k S_\beta)+
\dim \mathrm{Hom}(S_\alpha, \tau^k S_\beta[1])\Big)=\begin{cases}2 &\text{if }\alpha=\beta\\
1 & \text{otherwise.}
\end{cases}
\ee
Thus we have a $\widehat{A}_2$-configuration of spherical orbits.
All orbits are obtained form the $S_v$ one by acting with the $\mathfrak{S}_3$ automorphisms.
The three associated telescopic functors $L_a$ satisfy the $C\cb_3$ braiding relations, i.e.
\be
L_aL_bL_a=L_bL_aL_b\qquad a,b\in\{1,2,3\}.
\ee
\end{itemize}
\end{fact}

\begin{fact}[$E_r$ case]
\label{Dorbitscirc2}
$\mathscr{R}(E_r)$ does not contain spherical (half)-orbits.
\end{fact}

\begin{proof} We present a second proof of this \textbf{Fact}, introducing an argument we shall use often.
For simplicity, we consider $G=E_7$ or $E_8$ and write $\tilde h=h(G)/2$ which is an \emph{odd} integer in both cases. Assume (absurd) that the object $A$ has a spherical half-orbit. Then
\be
(-1)^k\, \chi(A,\tau^k A)=\delta^{(\tilde h)}_{k,0}+\delta^{(\tilde h)}_{k,1}
\ee
so that
\be
\sum_{k=1}^{\tilde h}\big(-e^{-2\pi m/\tilde h}\big)^k\chi(A,\tau^k A)=1+e^{-2\pi i m/\tilde h}
= \tilde h\,\boldsymbol{a}^t \,\boldsymbol{E}\!\left(\frac{1}{\tilde h}\sum_{k=1}^{\tilde h} e^{-2\pi i mk/\tilde h}\,\boldsymbol{H}^k \right)\!\boldsymbol{a},
\ee
where $\boldsymbol{a}$ is the vector which represents $[A]$ in the Grothendieck group. The expression in the large parenthesis is just the projector on the $e^{2\pi i m/\tilde h}$--eigenspace of $\boldsymbol{H}$. If for some $m\in \Z/\tilde h\Z$ this eigenspace is zero, we get a contradiction and conclude that no object $A$ can have a spherical half-orbit. This is the case, since the number of distinct eigenvalues of $\boldsymbol{H}=-\boldsymbol{\Phi}$ is $r(G)<\tilde h\equiv h(G)/2$ for $G=E_7,E_8$.
\end{proof}

\begin{rem}\label{Dorbitscirc+1}Let us return to the involution $\varepsilon_A$ in eqn.\eqref{Dorbitscirc-1} for $G=D_{2n}$. If $n>2$ we have a unique spherical orbit (\textbf{Fact \ref{Dorbitscirc}}), so a unique telescopic functor, and a unique involution $\theta$ in the derived Picard group. Hence $\varepsilon_A=\theta$. For $n=2$ we have 3 independent telescopic functors, hence 3 distinct $\varepsilon_A$, and 3 involutions in the subgroup $\mathfrak{S}_3$ of the derived Picard group, so that $\{\varepsilon_A\}=\{\text{involutions in }\mathfrak{S}_3\}$.
\end{rem}   

\section{$S$-duality in $(G,G^\prime)$ and $(G,\widehat{H})$ models}\label{produalities}

The method used in section \ref{moregeneral} to compute explicitly the $S$-duality groups of a large class of 4d $\cn=2$ is rather abstract. In order to make the logic clear, we first consider a simpler set of models, where things are much easier to visualize. The more general approach of section \ref{moregeneral} is modeled on the present section, although at a higher level of abstraction. 
\medskip

As introductory examples we consider  
 the first two families of $\cn=2$ models,
$(G,G^\prime)$ and $(\widehat{H}, G)$. Their BPS spectra
were already studied in refs.\!\cite{Cecotti:2010fi,Cecotti:2013lda} from the point of view of the Representation Theory of quivers with superpotential \cite{Alim:2011kw}. Indeed, these models have a very convenient quiver with superpotential given by the triangle tensor product in the sense of Keller \cite{kellerper} of the 
two acyclic quivers $G$, $G^\prime$ (resp.\! $G$, $\widehat{H}$). These models are also well understood in terms of the 4d/2d correspondence \cite{Cecotti:2010fi}: their 2d counterparts are the (2,2) LG models with superpotentials
\be
W(x,y,u,v)=W_G(x,y)+W_{G^\prime}(u,v)
\ee
and, respectively,
\be
W(x,y,u,v)=W_G(x,y)+W_{\widehat{H}}(u,v)
\ee
where $W_G(x,y)$ are the $ADE$ minimal singularities (cfr.\! table \ref{duval}) and $W_{\widehat{H}}(u,v)$ the affine superpotentials (cfr.\! table \ref{afff}). In this section we study the $S$-duality groups of the $(G,L)$ models ($L=G^\prime$ or $\widehat{H}$) as a warm-up for the more complicated QFTs of section \ref{moregeneral}, and as an illustration of the ideas and techniques of homological $S$-duality.

\subsection{Review of the $(G,L)$ QFTs in the categorical language}\label{yyyqwa}

In this subsection $L$ stands for an acyclic quiver which is either of Dynkin type $G^\prime$ or of affine type $\widehat{H}$.
 The main difference between the two cases is that 
the Coxeter element of a Dynkin quiver is semi-simple, while for an (acyclic) affine quiver it is never semi-simple: this just reflects the fact that the $\cn=2$ theories $(G,G^\prime)$ are superconformal while the $(G,\widehat{H})$ ones are asymptotically-free. 
In this subsection we follow
\cite{kellerper}; we refer to that paper for more precise statements and further details.

\subsubsection{The cluster category $\mathscr{C}(G,L)$}

The triangle tensor product of $G$ and $L$ yields the quiver with superpotential
$G\boxtimes L$ \cite{kellerper}; its Jacobian algebra
\be
\mathsf{Jac}(G,L)= \C(G\boxtimes L)/\partial \cw
\ee
is a certain ``completion'' of the product path algebra $\C G\times \C L$ with extra ``diagonal'' arrows and corresponding relations \cite{kellerper}. 
The BPS states of the $(G,L)$ QFT (in a physical regime covered by the triangle product form of the quiver) are then given by the modules
of $\mathsf{Jac}(G,L)$ which are stable with respect to the $\cn=2$
central charge \cite{Alim:2011kw}
\be
Z\colon K_0(\mathsf{mod}\,\mathsf{Jac}(G,L))\to \C.\ee

The category of prime interest for us is not the module category $\mathsf{mod}\,\mathsf{Jac}(G,L)$ but rather the associated \emph{cluster category} $\mathscr{C}(G,L)$. Consider the derived category 
\be
\mathscr{D}(G,L)=D^b\big(\mathsf{mod}\,\C G\times 
\C L\big).
\ee
The AR translations in the module categories of the two factor algebras induce auto-equivalences of the above derived category which, following \cite{kellerper}, we write as $\tau\otimes 1$ and $1\otimes \tau$, respectively. Their composition is 
\be\label{serregg}
\tau\otimes\tau\equiv S[-2],
\ee
where $S$ is the Serre functor
of the triangle category $\mathscr{D}(G,L)$.

As described in \S.\ref{lllksdv}, the cluster category $\mathscr{C}(G,L)$ is the triangular hull of the orbit category of the derived category of $\mathscr{D}(G,L)$ with respect to $S[-2]\equiv \tau\otimes\tau$, that is,
\be\label{kkka45n}
\mathscr{C}(G,L)= \mathsf{Hu}_\triangle\!\Big(\mathscr{D}(G,L)\Big/\langle \tau\otimes\tau\rangle^\Z\Big).
\ee

\paragraph{Fractional Calabi-Yau objects in the derived category.}
The 2d correspondent of a 4d $(G,G^\prime)$ model is a $(2,2)$ SCFT.
Let $X\in D^b(\mathsf{mod}\,\C G\times \C G^\prime)$, and let $h(G)$, $h(G^\prime)$ be the Coxeter numbers of the Dynkin graphs $G$ and $G^\prime$. From eqn.\eqref{serregg}
\be
\begin{split}
S^{h(G)h(G^\prime)}X&= \big(\tau^{h(G)h(G^\prime)}\otimes
\tau^{h(G)h(G^\prime)}\big)X[2h(G)h(G^\prime)]=\\
&=X[2h(G)h(G^\prime)-2h(G^\prime)-2h(G)]
\end{split}
\ee
where we used that, in $D^b(\mathsf{mod}\, \C G)$ with $G$ of Dynkin type, one has $\tau^{h(G)}=[-2]$\footnote{\ For a more precise statement, see table \ref{derivepicX}, omitting the relation $[2]=\mathrm{Id}$ which does not hold in the derived category.}.
Then all objects in the derived category $D^b(\mathsf{mod}\,\C G\times \C G^\prime)$ have Calabi-Yau fractional dimension
\be\label{whatc}
\hat c(G,G^\prime)=\frac{2\big(h(G)h(G^\prime)-h(G)-h(G^\prime)\big)}{h(G)h(G^\prime)}= \hat c(G)+\hat c(G^\prime),\qquad \big(\text{in }\mathbb{Q}\big)
\ee
where $\hat c(G)$ is the Virasoro central charge of the 2d minimal $\cn=2$ SCFT of type $G$. Of course, \eqref{whatc} is precisely the physical definition of $\hat c$. 
When $\hat c(G)+\hat c(G^\prime)<1$ the 4d theory is an Argyres-Douglas model.  $\hat c(G)+\hat c(G^\prime)$ is equal $1$ only for the three pairs $(D_4,A_2)$, $(A_3,A_3)$ and $(A_5,A_2)$, which correspond to the three elliptic complete SCFTs $E^{(1,1)}_6$, $E^{(1,1)}_7$, and $E^{(1,1)}_8$, respectively. For $\hat c(G)+\hat c(G^\prime)>1$, the cluster category $\mathscr{C}(G,G^\prime)$ is strictly larger than the orbit category (\!\!\cite{imaorbit} \textbf{Theorem 1.4}). 

The 2d model corresponding to the category $D^b(\mathsf{mod}\,\C G\times \C \widehat{H})$ is not UV conformal, hence not all objects are fractional Calabi-Yau. However objects with zero $G$ magnetic charge\footnote{\ \label{ringde}An object $X\in D^b(\mathsf{mod}\, \C G\times \C \widehat{H})$ has \textit{zero $G$ magnetic charge} iff $\boldsymbol{m}([X])=0$ where $$\boldsymbol{m}\equiv \mathrm{Id}\otimes \partial_R\colon K_0(D^b(\mathsf{mod}\, \C G\times \C \widehat{H}))\to \Z^{r(G)}$$ and $\partial_R$ is the Ringel defect in $\mathsf{mod}\,\C \widehat{H}$.} are fractional CY with dimension
$1+\hat c(G)$ (in $\mathbb{Q}$).

\subsection{The root category $\mathscr{R}(G,G^\prime)$ and its
auto-equivalences}

As in \S.\ref{rrrorott} the root category $\mathscr{R}(G,G^\prime)$ is defined as
\be
\mathscr{R}(G,G^\prime)=\mathsf{Hu}_\triangle\Big(\mathscr{D}(G,G^\prime)/[2\Z]\Big).
\ee
We are interested in the group of its auto-equivalences, or, more precisely, in the group 
\be
\mathbf{Aut}(\mathscr{R}(G,G^\prime))\subset GL\big(r(G)r(G^\prime),\Z\big)
\ee which represent their action on its Grothendieck group $\sim\Z^{r(G)\,r(G^\prime)}$.

\subsubsection{Auto-equivalences inherited from $\mathscr{R}(G)$ and $\mathscr{R}(G^\prime)$}

The auto-equivalence $T\equiv \tau\otimes\tau= S[-2]$ of $\mathscr{D}(G,G^\prime)$ induces an auto-equivalence of the root category which we still write $T$. More generally, all auto-equivalences $\sigma\in \mathrm{Aut}(\mathscr{R}(G))$ (resp.\! $\sigma^\prime\in \mathrm{Aut}(\mathscr{R}(G^\prime))$), induce auto-equivalences of $\mathscr{R}(G, G^\prime)$ of the form
\be
\sigma\otimes \mathrm{Id},\qquad \mathrm{Id}\otimes \sigma^\prime.
\ee

The auto-equivalences of $\mathscr{R}(L)$, with $L=G$ or $\widehat{H}$, may be read from 
the derived Picard group of the corresponding hereditary algebra, see   
\cite{picard}. For $G$ Dynkin,
the auto-equivalence group is described 
in table \ref{derivepicX}.
For $\widehat{H}$ affine, up to physically irrelevant motions of $\mathbb{P}^1$ (if we have less than three exceptional tubes), the derived Picard group is generated by the automorphisms of the quiver, $\tau$, and $[1]$. 

In general $\mathscr{R}(G,G^\prime)$ has additional auto-equivalences of a more subtle kind which are generated by telescopic functors. These are the  more interesting ones for our physical applications. We start by studying $T$-orbits in the root category.

\subsubsection{$T$-orbits  in $\mathscr{R}(G,G^\prime)$}\label{telinR}

Given an object $A\in \mathscr{D}(G,G^\prime)$, we use the same symbol $A$ to denote its canonical
image in the root category. 
For all object $A$
\be
T^{\mathrm{\,lcm}(h(G),h(G^\prime))}A=A\!\left[-2\frac{h(G)+h(G^\prime)}{\gcd(h(G),h(G^\prime))}\right]
\simeq A \ \text{in }\mathscr{R}(G,G^\prime),
\ee
so that all objects in the orbit category 
 belong to $T$-orbits of period $p$ dividing $\mathrm{lcm}(h(G),h(G^\prime))$. To be more precise,
 let us introduce the reduced 
 Coxeter number $\tilde h(G)$
 and the automorphism $\theta(G)$
 of the Dynkin graph $G$,
 \begin{align}\label{yyymml}
 &\tilde h(G) =\begin{cases}
 \tfrac{1}{2}\,h(G) &\text{for }G=A_1,D_r,E_r\\
 \phantom{\tfrac{1}{2}}\,h(G)
 &\text{otherwise,}
 \end{cases}
 &&\theta(G)=
 \begin{cases}
 \theta & \text{for }G=D_{2n+1},\ E_6\\
 \mathrm{Id} & \text{otherwise,}\end{cases}\\
 &s(G)=2\,\tilde h(G)/h(G), && \theta(G)^2=\mathrm{Id}.
 \end{align}
The minimal relation in
$D^b\,\mathsf{mod}\,\C G$ (cfr.\!\! table \ref{derivepicX}), may then be written in an unified way as
 \be
 \tau^{\tilde h(G)}\simeq \theta(G)\,[-s(G)].
 \ee
 
The minimal relation between 
 $T\equiv\tau\otimes\tau$, the shifts $[k]$, and graph automorphisms in $\mathscr{D}(G,G^\prime)$ is then
 \be\label{miniide}
T^{\,m(G,G^\prime)} 
 = \Theta(G,G^\prime)\!\big[-\Sigma(G,G^\prime)\big],
 \ee
 where\footnote{\ For brevity, sometimes we write $\gcd(a,b)$ simply as $(a,b)$.}
\begin{align}
\Sigma(G,G^\prime)&\equiv\frac{s(G)\tilde h(G^\prime)}{(\tilde h(G),\tilde h(G^\prime))}+\frac{s(G^\prime)\tilde h(G)}{(\tilde h(G),\tilde h(G^\prime))},\\
m(G,G^\prime)&\equiv\mathrm{lcm}(\tilde h(G),\tilde h(G^\prime)),\\
\Theta(G,G^\prime)&\equiv \theta(G)^{\tilde h(G^\prime)/(\tilde h(G),\tilde h(G^\prime))}\otimes \theta(G^\prime)^{\tilde h(G)/(\tilde h(G),\tilde h(G^\prime))}, \label{bigtheta}
 \end{align}
from which we may read the minimal period $p$ for each object in
$\mathscr{D}(G,G^\prime)/[2\Z]$. 
There are two possibilities:
\begin{itemize}
\item[A)] the integer $\Sigma(G,G^\prime)$ is \emph{odd:} 
the objects fixed by the automorphism 
$\Theta(G,G^\prime)$
have \textit{half-orbits}
of half-period $q=m(G,G^\prime)$, while all other objects have periodic orbits of period $p=2 m(G,G^\prime)$;
\item[B)] the integer $\Sigma(G,G^\prime)$ is \emph{even:}
 the objects fixed by the automorphism 
$\Theta(G,G^\prime)$
have periodic orbits of period $p=m(G,G^\prime)$ while all other objects have double period $2p$.
\end{itemize}
In particular, for the models $(A_{r-1},A_{r^\prime-1})$ ($r,r^\prime\geq 3$) all objects have period $p=\mathrm{lcm}(r,r^\prime)$. 
In case B), $\Theta(G,G^\prime)=\mathrm{Id}$, except when one (or both) Dynkin graphs $G$, $G^\prime$ are of type $D_{2\ell+1}$ or $E_6$
and moreover
\be
s(G) \tilde h(G^\prime)=s(G^\prime) \tilde h(G)=\big(\tilde h(G),\tilde h(G^\prime)\big)\mod2 \big(\tilde h(G),\tilde h(G^\prime)\big).
\ee 
Except in this special case, all objects in $\mathscr{R}(G,G^\prime)$ have the same period $p$. 
With the exception of the free hypermultiplet,
\be
\mathscr{R}(A_1,A_1)\simeq \mathscr{C}(A_1,A_1)\simeq \mathsf{vect}\oplus \mathsf{vect}[1],
\ee
we always have $p>1$.
From now on we assume
$G,G^\prime\neq A_1$, unless otherwise stated.\footnote{\ If $G$ or $G^\prime$ is equal $A_1$ the model is equivalent to an Argyres-Douglas theory whose duality group has already been discussed in the previous section.}
\medskip

The first condition in \textbf{Definition 4}, eqn.\eqref{four4}, is then satisfied by all objects in $\mathscr{R}(G,G^\prime)$.
Let us consider the second one, eqn.\eqref{four2}. In the present case $n=2$, and the condition becomes  
\be\label{yyyyxxxzaq}
\dim\mathrm{Hom}(A,T^kA[m])=\delta^{(p)}_{k,0}\,\delta^{(2)}_{m,0}+\delta^{(p)}_{k,1}\,\delta^{(2)}_{m,0}.
\ee
In particular, an object belonging to a spherical orbit is a \emph{rigid brick.}
By footnote \ref{footTel},
we limit ourselves to $A$'s in the orbit category $\mathscr{D}(G,G^\prime)/[2\Z]$.

\subsubsection{Necessary conditions for the existence of spherical (half-)orbits}

\paragraph{Spherical full orbits.}
If the model $(G,L)$ has no flavor charge\footnote{\ Compare with the physical discussion in \S.\,\ref{phphmot}.}, there is no \emph{full} 
spherical orbit in $\mathscr{R}(G,L)$.
In facts, a stronger statement holds:

\begin{fact} \label{needdddfla}
A necessary condition for the existence of a spherical full-orbit of period $p$ is
\be\label{uuuuurggs}
\left\{0,\; \frac{1}{p},\; \frac{2}{p},\;\cdots,\; \widehat{\,\frac{1}{2}\,},\;\cdots,\; \frac{p-1}{p}\mod 1\right\}\subseteq
\mathrm{Spectrum}\!\left(\frac{1}{2\pi i}\log\boldsymbol{H}\right),
\ee
(the notation $\widehat{\tfrac{1}{2}}$ means that one-half should be omitted from the  list).
In particular, $1$ should belong to the spectrum of $\boldsymbol{H}$, that is,
$\Gamma_\text{flavor}\neq0$.
\end{fact}

\begin{proof} A necessary condition for the full $T$-orbit of an object $A$ to be spherical (of period $p$) is that for all $s\in\Z/p\Z$
\be\label{jjjaq1}
\frac{1}{p}\big(1+e^{2\pi i s/p}\big)=\frac{1}{p}\sum_{k=1}^p e^{2\pi i k s/p}\,\chi(A,T^k A)= \boldsymbol{a}^t\boldsymbol{E}\!\left(\frac{1}{p}\sum_{k=1}^p e^{2\pi i s/p}\boldsymbol{H}^k\right)\!\boldsymbol{a}
\ee
where $\boldsymbol{a}$ is the integral vector representing $[A]\in K_0(\mathscr{R}(G,G^\prime))\simeq \Z^{r(G)\,r(G^\prime)}$. 
The sum in the large parenthesis  is the projector on the $e^{-2\pi i s/p}$-eigenspace of $\boldsymbol{H}$. If for some $s\neq p/2$ this eigenspace vanishes, we get a contradiction and conclude that no object may have a spherical orbit.\end{proof}

 For $(G,G^\prime)$ we have
\be \mathrm{Spectrum}\!\left(\frac{1}{2\pi i}\log\boldsymbol{H}\right)
=
\left\{ \frac{\ell}{h(G)}+\frac{\ell^\prime}{h(G^\prime)}\mod1\;\bigg|\; \ell\in e(G),\ \ell^\prime\in e(G^\prime)\right\}
\ee
where $e(G)$ is the set of exponents of the Lie algebra $G$. 

\begin{exe}
For $(A_{N-1},A_{N^\prime-1})$ we have no spherical orbit whenever $\gcd(N,N^\prime)=1$.
\end{exe}

\paragraph{Spherical half-orbits.} In agreement with the discussion in \S.\,\ref{phphmot}, spherical \emph{half}-orbit may exist even if no flavor charge is present. Indeed, for half-orbits of half-period $q$ eqn.\eqref{jjjaq1} gets replaced by\footnote{\ Recall that, if $\boldsymbol{a}$ is the dimension vector of an object belonging to a periodic \emph{half}-orbit, $\boldsymbol{H}^q\boldsymbol{a}=-\boldsymbol{a}$.}
\be
\frac{1}{q}\big(1+e^{\pi i (2s-1)/q}\big)=\frac{1}{q}\sum_{k=0}^{q-1} e^{\pi i k (2s-1)/q}\,\chi(A,T^k A)= \boldsymbol{a}^t\boldsymbol{E}\!\left(\frac{1}{2q}\sum_{k=0}^{2q-1} e^{2\pi i k(2s-1)/(2q)}\boldsymbol{H}^k\right)\!\boldsymbol{a}
\ee
for all $s\in\Z/q\Z$. Then in $\mathscr{R}(G,G^\prime)$
\begin{fact} 
A necessary condition for the existence  of a spherical \emph{half-}orbit of half-period $q$ is
\be\label{uuuuurggs}
\left\{\frac{1}{2q},\; \frac{3}{2q},\;\frac{5}{2q}\;\cdots,\; \widehat{\,\frac{1}{2}\,},\;\cdots,\; \frac{2q-1}{2q}\mod 1\right\}\subseteq
\mathrm{Spectrum}\!\left(\frac{1}{2\pi i}\log\boldsymbol{H}\right).
\ee
\end{fact}

\subsubsection{Action of $\tau\otimes 1$ and simple generators}\label{qqqqadx1}

If $A\in\mathscr{R}(G,L)$ belongs to a spherical (half-)orbit, also
 $(\tau^k\otimes 1)A$  belongs to a spherical (half-)orbit for all $k$.
Let $d$ be the smaller positive integer such that $(\tau^d\otimes 1)A=T^nA[j]$ for some $n,j\in\Z$. The objects $\{(\tau^k\otimes 1)A\}_{k=1}^d$ generate $d$ inequivalent spherical (half-)orbits. 

For $L=G^\prime$, 
and 
\be\label{lllllnnnxw3}
\big(\theta(G)\otimes 1\big)A=
\big(1\otimes \theta(G)\big)A=A
\ee
the condition $(\tau^d\otimes 1)A=T^n A[j]$ reduces to
\be
n=d\mod \tilde h(G),\qquad n=0\mod \tilde h(G^\prime).
\ee
By Chinese remainder, the smallest positive solution is 
$d=\gcd(\tilde h(G),\tilde h(G^\prime))$, and so,
acting with $\tau\otimes 1$ on a spherical (half-)orbit satisfying \eqref{lllllnnnxw3}
we
produce $d=\gcd(\tilde h(G),\tilde h(G^\prime))$ inequivalent spherical (half-)orbits. The multiplicity of each eigenvalue $e^{2\pi is/p}\neq-1$ in the spectrum of $\boldsymbol{H}$ is less or equal $d$. The restriction of the group $\mathbf{Tel}(G,G^\prime)$ to the $\boldsymbol{H}$-eigenspace
 \be
 V_s=\Big\{v\in \Gamma\otimes\C\;\Big|\; \boldsymbol{H}v=e^{2\pi i s/p}v\Big\},\qquad\dim V_s\leq d,
 \ee
is then an unitary reflection group acting on a vector space of dimension $\leq d$; its is generated by at most $d$ simple reflections
  \be
  \sigma_{s,i}=1-v_{s,i}\otimes w_{s,i}^t\qquad (i=1,\dots,\dim V_s\leq d)
  \ee
 whose vectors $v_i$ (resp.\! co-vectors $w_i^t$) span $V_s$ (resp.\! $V_s^\vee$). 
The spanning condition is satisfied by the vectors and co-vectors of 
the restrictions to the eigenspaces $V_s$ of  the $d$ telescopic matrices 
\be
\boldsymbol{L}_{(\tau^k\otimes 1)A}\Big|_{V_s}\equiv (\boldsymbol{\Phi}^k\otimes \boldsymbol{1})\boldsymbol{L}_A(\boldsymbol{\Phi}^{-k}\otimes \boldsymbol{1})\Big|_{V_s}=1-v_{s,k}\otimes w^t_{s,k}
\ee
 generated by the action of $\tau\otimes 1$. 
 
It is natural to expect that ``generically'' the telescopic functors generated by the action of $\tau\otimes 1$ on a single spherical $T$-orbit 
 form a set of simple generators of the group $\mathbf{Tel}(G,G^\prime)$
of the $(G,G^\prime)$ SCFT. This expectation is confirmed in the special cases $(D_4,A_2)$, $(A_3,A_3)$ and $(A_5,A_2)$ which correspond to the three elliptic complete models $E^{(1,1)}_r$ \cite{shepard}.

In the rest of the section, we shall explore the group $\mathbf{Tel}(G,G^\prime)$ generated by a set of ``simple reflections'' of the above form. We do not rule out the possibility of further enhancements of the duality group in  special models.

\subsubsection{Diadic objects}  
In the spirit of \S.\ref{qqqqadx1}, we need to find one spherical (half-)orbit of $\mathscr{R}(G,G^\prime)$, to which we apply $\tau^k\otimes 1$ to produce a generating set of $\mathbf{Tel}(G,G^\prime)$. To find some spherical (half-)orbits,  we focus on a simple class of objects of the root category. 

\begin{defn}
We say that an object
\be
A\in\mathsf{mod}\,\C G\times \C G^\prime
\ee
is \emph{diadic} if it  has the form
\be
A=X\otimes X^\prime,
\ee
with $X\in \mathsf{mod}\,\C G$, $X^\prime\in\mathsf{mod}\,\C G^\prime$  \emph{indecomposable.} We extend this
definition to $X, X^\prime$ objects of the corresponding derived categories, which are the repetitive categories of $\mathsf{mod}\,\C G$ and
$\mathsf{mod}\,\C G^\prime$, respectively.
The canonical images of $X\otimes X^\prime$ in $\mathscr{C}(G,G^\prime)$
and $\mathscr{R}(G,G^\prime)$ will  be also called \emph{diadic} and denoted by the same symbol. \end{defn}

We want to understand when a diadic object $X\otimes X^\prime$ belongs to a spherical orbit (resp.\! spherical half-orbit). A first necessary condition is that $X$ and $X^\prime$ are rigid bricks which both belong to a full-orbit or both belong to a half-orbit (resp.\! one to a full- and one to a half-orbit).
Let $q$ be the (half)period of the rigid brick $X$. We write
\be\label{iiicfg}
\dim \mathrm{Hom}(X, \tau^k X)+\dim \mathrm{Hom}(X, \tau^k X[1]) - \delta^{(q)}_{k,0}-\delta^{(q)}_{k,1}=\sum_{h\in H} a_h\,\delta^{(q)}_{k,h}
\ee
where $H\subset \Z/q\Z\setminus\{0,1\}$ is the subset of $k$'s such that the \textsc{lhs} is non-zero. The $a_h$'s are positive integers which satisfy the Serre symmetry $a_h=a_{1-h}$. $X$ belongs to  a spherical (half)orbit if and only if $H=\emptyset$.

\begin{fact}[criterion for a spheric diadic (half)orbit] Let $q$, $q^\prime$ be the (half)periods of the rigid bricks $X$, $X^\prime$. $X\otimes X^\prime$ belongs to a spherical (half)orbit iff $\gcd(q,q^\prime)>1$ and
the following three sets are all empty
\be
\begin{aligned}
K&=\big\{h\in H\;\big|\; h=0\ \text{or }1\mod \gcd(q,q^\prime)\big\},\\ 
K^\prime&=\big\{h\in H^\prime\;\big|\; h=0\ \text{or }1\mod \gcd(q,q^\prime)\big\},\\
J&=\big\{(h,h^\prime)\in H\times H^\prime\;\big|\;h-h^\prime=0\mod\gcd(q,q^\prime)\big\}.
\end{aligned}
\ee
\end{fact}

\begin{proof} Assume $\gcd(q,q^\prime)>1$. Then 
$X\otimes X^\prime$ belong to a spherical orbit iff\begin{equation}
\begin{split}
2&=\sum_{k=1}^{\mathrm{lcm}(q,q^\prime)}\sum_{j=0}^1\dim\mathrm{Hom}(X\otimes X^\prime, T^k(X\otimes X^\prime)[j])\equiv\\
&\equiv \sum_{k=1}^{\mathrm{lcm}(q,q^\prime)} \left(\sum_{j=0}^1\dim\mathrm{Hom}(X,\tau^k X[j])\right)\cdot\left(\sum_{j=0}^1\dim\mathrm{Hom}(X^\prime,\tau^k X^\prime[j])\right)\equiv\\
&\equiv2+\sum_{h\in K}a_h+\sum_{h^\prime\in K^\prime}a^\prime_{h^\prime}+\sum_{(h,h^\prime)\in J}a_h\,a^\prime_{h^\prime}.
\end{split}\qedhere
\end{equation}
\end{proof}

\paragraph{Tensor product of spherical (half)orbits.}
The cheapest way to satisfy the above criterion is to take $H=H^\prime=\emptyset$, that is, $X$ and $X^\prime$ which belong to spherical (half)orbits.
From \S.\ref{jggss90}, we see that in this case $A\equiv X\otimes X^\prime$ satisfies eqn.\eqref{lllllnnnxw3}.
If $\tau^k X$ and $\tau^{k^\prime}X$ are two spherical (half)orbits of (half)periods $q$, $q^\prime$, we have
$qq^\prime$ diadic objects $\tau^k X\otimes \tau^{k^\prime}X^\prime$ which form spheric (half)orbits of length $\mathrm{lcm}(q,q^\prime)$. Therefore
we have $\gcd(q,q^\prime)$ disjoint spherical (half)orbits
\be
\big\{T^k(\tau^a X\otimes X^\prime)\big\} \qquad a\in \Z\big/\!\gcd(q,q^\prime)\Z,
\ee
for each pair of spherical (half)orbits in $\mathscr{R}(G)$, $\mathscr{R}(G^\prime)$. Then a pair of spherical (half)orbits $\tau^k X$, $\tau^{k^\prime}X^\prime$ produces $\gcd(q,q^\prime)$ telescopic functors
\be\label{uuurq13}
\cl_a=L_{\tau^a X \otimes X^\prime} \qquad a\in \Z\big/\!\gcd(q,q^\prime)\Z.
\ee

The spherical orbits in $\mathscr{R}(A_{N-1})$ were described in \textbf{Fact \ref{aNspherorb}}. With the notation used there, for $\mathscr{R}(A_{N-1},A_{N^\prime-1})$ we
have the distinct telescopic functors
\be
\cl_a=L_{A_a\otimes A_1^\prime}\qquad a\in  \Z\big/\!\gcd(N,N^\prime)\Z.
\ee
However, in addition to the tensor-product ones, in $\mathscr{R}(A_{N-1},A_{N^\prime-1})$ there are also diadic spherical orbits which are \emph{not} tensor products of spherical (half)orbits, as well as spherical orbits which are \emph{non-diadic,} as we now discuss.

\subsubsection{Diadic spherical orbits in $\mathscr{R}(A_{N-1},A_{N^\prime-1})$}
The path algebra $\C \vec A_{N-1}$ of the linear $A_{N-1}$ quiver
is a \emph{uniserial} algebra,\footnote{\ \label{uuunisss} For uniserial (Nakayama) algebras see Chapter V of \cite{assem}.} and its indecomposable modules are uniquely identified by their top and length. We write $M_{i,\ell}$ for the indecomposable module of length $\ell$ and top $S_i$.
The $\tau$-orbit of the module
 $M_{i,\ell}$ in $\mathscr{R}(\vec A_{N-1})\equiv D^b\mathsf{mod}\,\C\vec A_{N-1}/[2\Z]$  has the form
 \be
\tau^{i-1}M_{1,\ell}=\begin{cases} M_{i,\ell} & \text{for }1\leq i+\ell\leq N\\
M_{i+\ell-N,N-\ell}[-1] &\text{for }N-\ell< i\leq N.
\end{cases}
 \ee 
Since the $\tau$-orbits of modules of length
$\ell$ and $N-\ell$ are interchanged by the shift $[1]$, with no loss
 we may restrict to orbits of objects with $\ell\leq N/2$. A simple computation yields:
 
 \begin{lem} For $1\leq \ell\leq N/2$, one has in
 $\mathscr{R}(\vec A_{N-1})$ 
 \be
 \sum_{j=0}^1\dim \mathrm{Hom}(M_{i,\ell},\tau^k M_{i,\ell}[j])=\sum_{h=1-\ell}^{\ell}\delta^{(N)}_{k,h}.
 \ee
 \end{lem} 
 Hence in $\mathscr{R}(A_{N-1},A_{N^\prime-1})$ for $\ell\leq N/2$ and 
 $\ell^\prime\leq N^\prime/2$
 \be
 \sum_{k=1}^{\mathrm{lcm}(N,N^\prime)}\sum_{j=0}^1 \dim\mathrm{Hom}\big(M_{i,\ell}\otimes M_{j,\ell^\prime},T^k(M_{i,\ell}\otimes M_{j,\ell^\prime}[j])\big)=\sum_{h=1-\ell}^\ell\;\sum_{h^\prime=1-\ell^\prime}^{\ell^\prime} \delta^{(\gcd(N,N^\prime))}_{h,h^\prime}.  
 \ee
The diadic orbit $\{T^k(M_{i,\ell}\otimes M_{j,\ell^\prime})\}$ is spherical iff the \textsc{rhs} is 2. This cannot happen if
 $\ell$ and $\ell^\prime$ are both $\geq 2$; at least one of the two factors $M_{i,\ell}$, $M_{j,\ell^\prime}$  should be simple. 
Modulo interchanging the two factors, we may assume $M_{j,\ell^\prime}=S_1$ with no loss. Then the condition
reduces to
\be
\sum_{h=1-\ell}^\ell \Big(\delta^{(\gcd(N,N^\prime))}_{h,0}+\delta^{(\gcd(N,N^\prime))}_{h,1}\Big)=2,
\ee
that is,
\be
2\ell\leq \gcd(N,N^\prime).
\ee
Therefore,

\begin{fact}\label{fffassw} The distinct diadic spherical $T$-orbits in $\mathscr{R}(A_{N-1},A_{N^\prime-1})$ are generated by
\be
\tau^a S_1\otimes M_{1,\ell^\prime},\quad 
M_{1,\ell}\otimes \tau^a S_1,\quad
\tau^a S_1\otimes M_{1,\ell^\prime}[-1],\quad 
M_{1,\ell}\otimes \tau^a S_1[-1],  
\ee
where $2\leq 2\ell\leq \gcd(N,N^\prime)$ and $1\leq a\leq \gcd(N,N^\prime)$.
Up to shift, there are
\be
\gcd(N,N^\prime)\left[\frac{\gcd(N,N^\prime)}{2}\right]
\ee
diadic spherical orbits.
\end{fact}

\paragraph{Simple telescopic functors.} The telescopic functor associated to orbits which are tensor products of spherical orbits, eqn.\eqref{uuurq13}, suffice to produce all the telescopic functor associated to the spherical diadic orbits in \textbf{Fact \ref{fffassw}}.  For instance, the triangle which defines $L_{S_i\otimes S_j}(\tau M_{i,\ell}\otimes S_j)$ is
\be
\to \tau M_{i,\ell}\otimes S_j\to L_{S_i\otimes S_j}(\tau M_{i,\ell}\otimes S_j)\to S_i\otimes S_j\to,
\ee 
so that $L_{S_i\otimes S_j}(M_{i+1,\ell}\otimes S_j)= M_{i,\ell+1}\otimes S_j$
and
\be\label{xmjjg}
L_{S_i\otimes S_j}L_{M_{i+1,\ell}\otimes S_j}R_{S_i\otimes S_j}=L_{M_{i,\ell+1}\otimes S_j}.
\ee
All telescopic functors associated to diadic spherical orbits may then be written  as words in the telescopic functors $\cl_a$ associated to tensor products of spheric orbits, 
eqn.\eqref{uuurq13}. More general words in the $\cl_a$ produce telescopic functor associated to spherical orbits which are not diadic.

This result is consistent with the heuristic discussion at the end of \S.\ref{qqqqadx1}. We take the $\cl_a$ in eqn.\eqref{uuurq13}
as the ``simple'' telescopic functors.  

\subsubsection{Braiding relations in $\mathrm{Aut}(A_{N-1},A_{N^\prime-1})$}\label{999uuty87}

In view of \textbf{Fact \ref{braidtele}}, to determine the quadratic and cubic braid relations between the $\cl_a$ we need to compute the sum
\be
\begin{split}
\sum_{k=1}^{\mathrm{lcm}(N,N^\prime)}\sum_{j=0}^1
& \dim\mathrm{Hom}(A_i\otimes A^\prime_a,T^k(A_j\otimes A_b)[j])=\\
&=
2\,\delta^{(\mathrm{gcd}(N,N^\prime))}_{i-j,a-b}+\delta^{(\mathrm{gcd}(N,N^\prime))}_{i-j,a-b+1}+
\delta^{(\mathrm{gcd}(N,N^\prime))}_{i-j+1,a-b}.
\end{split}
\ee

When $\mathrm{gcd}(N,N^\prime)>2$, at most one of the three terms in the \textsc{rhs} can be non-zero for given $i,a,j,b$. If the first term is non zero, $A_i\otimes A^\prime_a$ and
$A_j\otimes A^\prime_b$ belong to the same spherical $T$-orbit and hence
$L_{A_i\otimes A^\prime_a}\equiv
L_{A_j\otimes A^\prime_b}$.
If the non-zero term is the second or third, $A_i\otimes A^\prime_a$ and
$A_j\otimes A^\prime_b$ form
a $(A_2)$-configuration of spherical orbits in the sense of \textbf{Definition \ref{ttttealbbr}}, and hence the corresponding pair of telescopic functors satisfy the $\cb_3$ relation.
Finally, if all three terms vanishes, the telescopic functors $L_{A_i\otimes A^\prime_a}$ and $L_{A_j\otimes A^\prime_b}$ commute.
We conclude

\begin{fact}\label{tttqas2}
Let 
$\mathrm{gcd}(N,N^\prime)>2$. 
Then the $\gcd(N,N^\prime)$ simple telescopic functors $\cl_a$ satisfy the relations of the \emph{affine} Artin braid group 
$G_{\tilde A_{\gcd(N,N^\prime)-1}}\subset \cb_{\mathrm{gcd}(N,N^\prime)+1}$, that is, 
\begin{align}\label{rel1}
\cl_a\cl_{a+1}\cl_a&=\cl_{a+1}\cl_a\cl_{a+1}, &&a\in\Z/\gcd(N,N^\prime)\Z\\
\cl_a\cl_b&=\cl_b\cl_a,\quad\text{for }|a-b|\geq 2.\label{rel2}
\end{align}
\end{fact}

For a review of Artin's groups see
\cite{review}. The center of the affine Artin braid group is trivial,
$Z(G_{\tilde A_n})=1$ \cite{circcirc}.

\begin{rem} In the special case $\gcd(N,N^\prime)=2$ one expects that 
the two independent simple telescopic functors $\cl_1$ and $\cl_2$ satisfy ``model-dependent'' higher braid relations of the form
\be
\overbrace{\cl_1\cl_2\cl_1\cl_2\cdots}^{s>3\ \text{factors}}=
\overbrace{\cl_2\cl_1\cl_2\cl_1\cdots}^{s>3\ \text{factors}}.
\ee
The precise braiding relation for a given pair $(N,N^\prime)$ may be determined from the explicit matrices $\boldsymbol{\cl}_a$ which yield the action of the functors $\cl_a$ on the lattice $K_0(\mathscr{R}(A_{N-1},A_{N^\prime-1})$, cfr.\! \S.\ref{actgro}.
We have checked the pairs $(N,N^\prime)=(4,6)$, $(4,10)$, $(6,10)$, $(6,14)$ and $(10,14)$; from these examples we infer the following rule:\medskip

\textit{Let $\cl_1$, $\cl_2$ be the two distinct ``simple'' telescopic endo-functors of $\mathscr{R}(A_{N-1},A_{N^\prime-1})$ with $\gcd(N,N^\prime)=2$. They satisfy the following relations
\begin{gather}
\cl_1^{\mathrm{lcm}(N,N^\prime)}=\cl_2^{\mathrm{lcm}(N,N^\prime)}=\mathrm{Id}\\
\big(\cl_1\cl_2\big)^{\mathrm{lcm}(N,N^\prime)/2}=
\big(\cl_2\cl_1\big)^{\mathrm{lcm}(N,N^\prime)/2}=\mathrm{Id}.\label{kkkdge}
\end{gather} } 
\end{rem}

\begin{exe} Suppose $N\mid N^\prime$. Then for each root $\alpha \in \Delta(A_{N-1})$ we have a distinct telescopic endo-functor\footnote{\ We label the objects of the $\mathscr{R}(G)$ with the roots of the Lie algebra $G$ using the correspondence $$\text{(root of $G$) $\leftrightarrow$ (indecomposable object of $\mathscr{R}(G)$)}$$ discussed after eqn.\eqref{pppbbn4x}.} in $\mathscr{R}(A_{N-1},A_{N^\prime-1})$
\be
\cl_\alpha\equiv L_{X_\alpha\otimes A_1},
\ee
and the adjoint action of the telescopic functors on themselves is given by
the Weyl group of $SU(N)$
\be
\cl_\alpha \cl_\beta\cl_\alpha^{-1}=\cl_{w_\alpha(\beta)},
\ee
so that they generate a braid group
$\cb_N$ in one-less generator. Indeed, from \eqref{xmjjg}
\be
\cl_N=\cl_1\cl_2\cdots \cl_{N-2}\,\cl_{N-1}\cl_{N-2}^{-1}\cdots\cl_2^{-1}\cl_1^{-1}.
\ee

\end{exe}

\paragraph{The group of auto-equivalences of $\mathscr{R}(A_{r-1},A_{r^\prime-1})$.}
Besides the telescopic functors in \textbf{Fact \ref{tttqas2}}, $\mathscr{R}(A_{r-1},A_{r^\prime-1})$ as the auto-equivalences generated by the shift $[1]$, the translation $T=\tau\otimes\tau$, and the functor 
\be\label{zamozamo}
\begin{split}
\mathsf{Z}&=\tau^{a N/\mathrm{gcd}(N,N^\prime)}\otimes 
\tau^{bN^\prime/\mathrm{gcd}(N,N^\prime)}\\
&\text{where }a,b\in \Z\
\text{such that }\frac{a\,N}{\mathrm{gcd}(N,N^\prime)}-
\frac{b\,N^\prime}{\mathrm{g.c.d.}(N,N^\prime)}=1.\end{split}
\ee
They satisfy the relations
\be
T^{\,\mathrm{lcm}(N,N^\prime)}=\mathrm{Id},\qquad \mathsf{Z}^{\,\mathrm{gcd}(N,N^\prime)}=\mathrm{Id},\qquad [1]^2=\mathrm{Id}.
\ee
Moreover, $T$ and $[1]$ commute with all auto-equivalences, and hence with the $\cl_a$.
On the other hand,
\be
\mathsf{Z}\, \cl_a\,\mathsf{Z}^{-1}=\cl_{a+1}.\label{r3}
\ee
The group generated by
the $\cl_a$ and $\mathsf{Z}$, subjected to the relations \eqref{rel1}\eqref{rel2}\eqref{r3}
is known as the \emph{circular braid group} (or annular braid group)
$C\cb_{\mathrm{gcd}(N,N^\prime)}$
\cite{xxx} which is  isomorphic to the Artin braid group of finite type associated with the Dynkin graph $B_{\mathrm{gcd}(N,N^\prime)}$ \cite{xxx}. Its center is the infinite cyclic group generated by $\mathsf{Z}$ \cite{xxx}.

The group which acts effectively
on the root (cluster) category is the quotient of $C\cb_{\mathrm{gcd}(N,N^\prime)}$ by
some normal subgroup $\mathfrak{N}$ (the isotropy group).
Thus,

\begin{fact}\label{ttttqqqz}
Let $\mathrm{gcd}(N,N^\prime)>2$.
The group of auto-equivalences of the root category $\mathscr{R}(A_{N-1},A_{N^\prime-1})$  consists at least in a group of the form
\be\label{uuuq12}
\Z\big/\mathrm{lcm}(N,N^\prime)\Z\times \Z\big/2\Z \times C\cb_{\mathrm{gcd}(N,N^\prime)}\big/\mathfrak{N},
\ee 
where the first factor is generated by
$T$ and the second one by $[1]$.\end{fact}

The auto-equivalences of the cluster category $\mathscr{C}(A_{N-1},A_{N^\prime-1})$ then follows as in \S.\ref{fffistcaase}.

\paragraph{Special models and consistency checks.} We have the equivalences
\be\label{eqeeertk}
\mathscr{C}(D_4)\simeq \mathscr{C}(A_2,A_2),\qquad 
\mathscr{C}(E_6)\simeq \mathscr{C}(A_3,A_2),\qquad
\mathscr{C}(E_8)\simeq \mathscr{C}(A_4,A_2)
\ee
and hence consistency requires a correspondence between the spherical (half)orbits we found in \textbf{Facts \ref{Dorbitscirc},\;\ref{Dorbitscirc2}} for $\mathscr{R}(D_4)$, $\mathscr{R}(E_6)$, and $\mathscr{R}(E_8)$ and the ones found in the present subsection for (respectively) $\mathscr{R}(A_2,A_2)$, $\mathscr{R}(A_3,A_2)$, and $\mathscr{R}(A_4,A_2)$. 

This is trivially true for the last two pairs in eqn.\eqref{eqeeertk} since on both sides of the equivalence we found no spherical object/orbit. In $\mathscr{R}(A_2,A_2)$
we found $\gcd(h(A_2),h(A_2))=3$ spherical orbits, in perfect agreement with
the 3 spherical orbits of  \textbf{Facts \ref{Dorbitscirc}}.
Hence
\be
L_{S_\alpha}\longleftrightarrow L_{\tau^a S_1\otimes S_1},\qquad \alpha=v,s,c,\ \ a\in\Z/3\Z,
\ee
which generate a braid group $C\cb_3$.

\subsubsection{Telescopic functors in $\mathscr{R}(D_{n+1},D_{n^\prime+1})$}

Using eqn.\eqref{uuurq13} we construct the ``simple'' telescopic functors $\cl_a$ associated to tensor products of the spherical half-orbits described in \textbf{Fact \ref{Dorbitscirc}}. In the case $n,n^\prime>3$ we have $\gcd(n,n^\prime)$
independent such functors
\be
\cl_a=L_{\tau^a A_0\otimes A_0^\prime},\qquad a\in \Z\big/\!\gcd(n,n^\prime)\Z.
\ee
To get the braid relations between the $\cl_a$'s we compute
\be
\sum_{k=1}^{\mathrm{lmc}(n,n^\prime)}\sum_{j=0}^1\dim\mathrm{Hom}(A_i\otimes A^\prime_0,T^k(A_j\otimes A_0)[j])=
2\,\delta_{ij}^{(\gcd(n,n^\prime))}+\delta_{i,j+1}^{(\gcd(n,n^\prime))}+\delta_{i,j-1}^{(\gcd(n,n^\prime))}.
\ee
\textbf{Fact  \ref{ttttqqqz}}
applies to the SCFT $(D_{n+1},D_{n^\prime+1})$ ($n,n^\prime>3$) with the replacement 
\be
(N,N^\prime)\to (n,n^\prime).
\ee
For $n=3m>3$, $n^\prime=3$ we get 9 simple telescopic functors (notation as in eqn.\eqref{uuutvcm})
\be
\cl_{a,\alpha}=L_{\tau^a A_0\otimes S_\alpha},\qquad a\in\Z/3\Z,\ \alpha=v,s,c.
\ee
and, for $n=n^\prime=3$, 27 simple telescopic functors
\be
\cl_{a,\alpha,\dot\alpha}= L_{\tau^a S_\alpha\otimes S_{\dot\alpha}},\qquad a\in\Z/3\Z,\ \alpha, \dot\alpha=v,s,c.
\ee 

\paragraph{Braid relations.}
For $n=3m>3$, $n^\prime=3$ we have
\be
\sum_{k=1}^{3m}\sum_{j=0}^1\dim\mathrm{Hom}\big(\tau^a A_0\otimes S_\alpha,T^k(\tau^b\otimes S_\beta)[j]\big)=
2\delta_{\alpha,\beta}\,\delta^{(3)}_{a,b}+\delta_{a,b+1}^{(3)}+\delta_{a,b-1}^{(3)},
\ee
and the braid relations are
\be
\cl_{a+1,\alpha}\cl_{a,\beta}\cl_{a+1,\alpha}=\cl_{a,\beta}\cl_{a+1,\alpha}\cl_{a,\beta},\qquad \cl_{a,\alpha}\cl_{a,\beta}=\cl_{a,\beta}\cl_{a,\alpha}.
\ee
For $n=n^\prime=3$ we have
\be
\begin{split}
\sum_{k=1}^3\sum_{j=0}^1 \dim\mathrm{Hom}\big(\tau^a S_\alpha\otimes& S_{\dot\alpha}, T^k(\tau^b S_\beta\otimes S_{\dot\beta})\big)=\delta^{(3)}_{a,b}\big(1+3 \delta_{\alpha,\beta}\delta_{\dot\alpha,\dot\beta}-\delta_{\alpha,\beta}-\delta_{\dot\alpha,\dot\beta}\big)+\\
&+\big(\delta^{(3)}_{a,b+1}+\delta^{(3)}_{a,b-1}\big)\big(\delta_{\alpha,\beta}+\delta_{\dot\alpha,\dot\beta}-\delta_{\alpha,\beta}\delta_{\dot\alpha,\dot\beta}\big)
\end{split}
\ee
from which we read the braid relations.

\subsubsection{$L_A$ functors in $\mathscr{R}(A_{N-1},D_{n+1})$} 

If $\gcd(N,n)>1$ we have the diadic spherical half-orbits generated by
the tensor products
\be
A_i\otimes A^\prime_a\quad\text{or } A_i\otimes S_\alpha\ \text{for }n=3.
\ee
where $A_i$ ($A^\prime_0$ and $S_\alpha$) are described in \textbf{Fact \ref{aNspherorb}}
(resp.\! \textbf{Fact \ref{Dorbitscirc}}).
For $n\neq3$ we have $\gcd(N,n)$ simple auto-equivalences
\be
\cl_a=L_{\tau^a A_1\otimes A^\prime_0}\qquad a\in\Z/\!\gcd(N,n)\Z.
\ee
If $\gcd(N,n)>2$, these auto-equivalence satisfy the braid relations \eqref{rel1},\eqref{rel2},\eqref{r3} and \textbf{Fact  \ref{ttttqqqz}}
applies to the SCFT $(A_{N-1},D_{n+1})$ with the replacement $N^\prime\to n$.

\begin{rem}\label{tttuuuuura}
As in \S.\ref{999uuty87}, a consistency check is in order. The model
$(A_2,D_4)$ is equivalent to 
$E_6^{(1,1)}$. For $(A_2,D_4)$ the group of auto-equivalences we have found is (up to some finite group) a realization of $C\cb_3$ (generated by three telescopic functors)
while for $E_6^{(1,1)}$ (again modulo a finite group) we have a realization of
$\cb_3$ (generated by two telescopic functors). This is consistent since
$C\cb_3$ is a subgroup of $\cb_3$ of finite index.
\end{rem}

\subsection{$S$-duality in $(G,\widehat{H})$ models} 
The categorical description of the asymptotically-free models $(G,\widehat{H})$ is obtained from the  $(G,G^\prime)$ one by replacing the finite-type Dynkin graph $G^\prime$ with the affine one $\widehat{H}$. This has a crucial consequence: the Coxeter element $\boldsymbol{\Phi}$ of an affine Lie algebra is never semi-simple,
and hence $T=\tau\otimes\tau$ has not finite order in the root category
$\mathscr{R}(G,\widehat{H})$.
Consequently, not all its objects belong to periodic $T$-orbits, but only the proper subclass of its CY objects.\footnote{\ The class of CY-objects in $\mathscr{D}(G,\widehat{H})$ coincide with the class of objects having zero $G$ magnetic charge (cfr.\! footnote \ref{ringde}). Objects are CY iff they are mutually local with respect to the $W$ bosons. }
Since telescopic functors are constructed out of spherical \emph{finite} orbits, the categories
$\mathscr{R}(G,\widehat{H})$ are  ``poorer'' of spherical orbits than the $\mathscr{R}(G,G^\prime)$ ones, and hence have ``smaller'' $S$-duality groups.

This can be understood on physical grounds.\footnote{\ The following physical discussion holds in the derived Jacobian module category of the triangle form of the quiver, in a weakly-coupled physical regime such that the stable light modules are in fact modules of $\C G\times \C\widehat{H}$ (i.e.\! modules with vanishing extra ``diagonal'' arrows).} As explained in ref.\!\!\cite{Cecotti:2012va}, the non-trivial Jordan blocks of $\boldsymbol{H}$  measure the non-zero $\beta$-functions of the gauge couplings, that is, the deviation from conformal invariance. In a $\cn=2$ theory a scale anomaly implies a $U(1)_R$ anomaly, so that a
$U(1)_R$ rotation by an angle $\phi$ implies a shift in the Yang-Mills vacuum angle
\be
\theta\to\theta+\alpha\,\phi,
\ee
where $\alpha$ is the $\beta$-function coefficient. The physical interpretation\footnote{\ $T$ acts on the Grothendieck class of the root category as $\boldsymbol{H}$, the 2d monodromy matrix. Thus, by definition, in 2d $T$ is a chiral rotation by $2\pi$. From appendix \ref{4d2dchiral}, we see that a $U(1)_R$ rotation by an angle $2\pi$ in 2d corresponds to a $U(1)_R$ rotation in 4d by an angle $\phi=2\pi(1-\hat c/2)$. In $\mathbb{Q}$ we have $\hat c(G,\widehat{H})=2(1-1/h(G))$ (cfr. table \ref{duval}) and $\phi=2\pi/h(G)$. } of the functor $T$ is a $U(1)_R$ rotation by $\phi=2\pi/h(G)$. We know from the Witten effect \cite{effwitten} that a $U(1)_R$ rotation by $2\pi$, generated by
$T^{h(G)}$, acts on the electric and magnetic charges as
\be
\begin{pmatrix}e\\ m\end{pmatrix}\to \begin{pmatrix}e+\alpha\, m\\ m\end{pmatrix}\equiv \begin{pmatrix}1 & \alpha\\
0 & 1\end{pmatrix} \begin{pmatrix}e\\ m\end{pmatrix}\label{kkkk1c}
\ee
where $\alpha$ measures the
chiral anomaly, but also the failure of $\boldsymbol{H}$ to be semi-simple\footnote{\ The $2\times 2$ matrix in eqn.\eqref{kkkk1c} is the restriction of $\boldsymbol{H}^{h(G)}$ to the subgroup of the Grothendieck group generated by a dual pair of electric/magnetic charges.}. If $\alpha$ is not zero, there is a preferred $S$-duality frame, defined by the condition that the magnetic charge is invariant under $2\pi$ chiral rotations. Then
only transformations which preserve this magnetic charge can be dualities, and the $S$-duality group is restricted to a parabolic subgroup of the superconformal one.

\subsubsection{The regular subcategory and spherical orbits}

As mentioned in the introduction, the models $(G,\widehat{H})$ have the physical interpretation of SYM with gauge group $G$ coupled to some \emph{superconformal} matter system (which may contain its own SYM subsectors). There is a well-defined triangle sub-category of the derived category $\mathscr{D}(G,\widehat{H})$
which corresponds to the matter sector (or, more generally, to the states of zero magnetic charge).
We shall discuss such sub-constituent categories in the next section. Here we limit to consider simple objects in $\mathscr{D}(G,\widehat{H})$ which do belong to the matter sub-category. Since the matter is conformal, objects in the sub-category are fractional CY, and are mapped into periodic objects of $\mathscr{R}(G,\widehat{H})$.

We consider the subcategory $\mathcal{R}\subset\mathsf{mod}\,\C \widehat{H}$ of regular modules \cite{bowey}. It is a $\mathbb{P}^1$ family of stable tubes, almost all homogeneous, except (at most) three which have periods 
 $\{p_1,p_2,p_3\}$ (listed for the various $\widehat{H}$ in the left part of table \ref{mmmatt}). In the $i$--th tube there are $p_i$ regular simples $R_{a,i}$ ($a\in\Z/p_i\Z$) on which the AR translation $\tau$ acts periodically with period $p_i$ \cite{bowey}
 \begin{gather}
 R_{a+1,i}=\tau R_{a,i},\qquad \tau^{p_i} R_{a,i}=R_{a,i}\\
\label{uuuycm}
\dim \mathrm{Hom}^\bullet(R_{a,i},\tau^k R_{a,j})=\delta_{ij}\big(\delta^{(p_i)}_{k,0}+\delta^{(p_i)}_{k,1}\big).
 \end{gather}
 
\paragraph{Spherical orbits in $\mathscr{R}(A_{N-1},\widehat{H})$.} The diadic objects of the form $X\otimes R_{a,i}\in \mathscr{R}(A_{N-1},\widehat{H})$ are periodic of period
$\mathrm{lcm}(N,p_i)$. The same argument as in \textbf{Fact \ref{needdddfla}} shows that for $\gcd(N,p_i)=1$ there are no spherical orbits. Using eqn.\eqref{iiicfg}, for $\mathrm{lcm}(N,p_i)>1$, we get
\be
\sum_{k=1}^{\mathrm{lcm}(N,p_i)}\sum_{j=0}^1\dim\mathrm{Hom}(X\otimes R_{a,i},T^k(X\otimes R_{a,i})[j])=2+\sum_{h\in H} a_h\big(
\delta^{(\gcd(N,p_i))}_{h,0}+
\delta^{(\gcd(N,p_i))}_{h,1}\big)
\ee
The condition is as in \textbf{Fact \ref{fffassw}}: all objects $X$ of length $2\ell\leq\gcd(N,p_i)$ produce spherical orbits. The simple telescopic functors are of the form
\be\label{dusu1}
\cl_{a,i}=L_{A_0\otimes \tau^a R_{1,i}}\qquad a\in \Z/\!\gcd(N,p_i).
\ee
The independent telescopic functor in the $i$-th tube, $\cl_{i,a}$ are cyclically permuted by the auto-equivalence $\mathsf{Z}_i$ defined as in eqn.\eqref{zamozamo} (with $N^\prime\to p_i$)
\be\label{qwza1}
\mathsf{Z}_i\,\cl_{a,i}\,\mathsf{Z}_i^{-1}=\cl_{a+1,i}, \qquad
\mathsf{Z}_i^{\gcd(N,p_i)}=1.
\ee
while telescopic functors associated with distinct tubes commute.
If the condition
 $\mathrm{g.c.d.}(N,p_i)>2$ holds,  the $\cl_{a,i}$'s generate
the affine Artin braid group $G_{\tilde A_{\gcd(N,p_i)-1}}$
\be\label{qwza}
\cl_{a,i}\,\cl_{a+1,i}\,\cl_{a,i}=\cl_{a+1,i}\,\cl_{a,i}\,\cl_{a+1,i},\qquad \cl_{a,i}\cl_{b,i}=\cl_{b,i}\cl_{a,i}\ \text{for }|a-b|>1,
\ee
 while the functors $\{\mathsf{Z}_i,\ \cl_{a,i},\;|\; a\in\Z/\gcd(N,p_i)\Z)$) generate the cyclic (annular) braid group $C\cb_{\gcd(N,p_i)}$.
In particular, if only one exceptional tube is present\footnote{\ If several exceptional tubes are present the $S$-duality group contains, in addition to the $S$-dualities produced by telescopic functors, $T$ and $1\otimes\tau$, the permutation of the tubes of the same period $p_i$.} and
$\gcd(N,p)>2$, the $S$-duality group has again the form
\be\label{dusu2}
\mathbf{Tel}(A_{N-1},\widehat{A}(p,1))=\Big(\Z/\mathrm{lcm}(N,p)\Z\times C\cb_{\gcd(N,p)}\Big)/\mathfrak{N},\qquad \gcd(N,p)>2,
\ee
where $\Z/\mathrm{lcm}(N,p)\Z$ stands for the cyclic group generated by $T$ and where $\mathfrak{N}$ is a normal subgroup. 

\subsubsection{$S$-duality in $SU$ linear quiver gauge theories. I}\label{aaa15i} 

We illustrate the above conclusion in an important class of examples: the model $(A_{N-1},\widehat{A}(p,1))$ where $p\mid N$ and $p>1$
\cite{Cecotti:2013lda}.
This model is actually the linear quiver theory in figure \ref{linequiv}
where $m\equiv N/p$. It has flavor symmetry $U(1)^{p-1}$ and
all its gauge couplings are exactly marginal, except the one in the first node, which is just asymptotically-free.
Since the model is Lagrangian, $\mathfrak{m}=1$. The $S$-duality group of the full theory may act on the electric/magnetic charges of the first node only through the parabolic group \eqref{kkkk1c}. We expect a more interesting $S$-duality action on the electric/magnetic charges of the other nodes which have vanishing $\beta$-function: a finite-index subgroup $\Gamma_k\subset SL(2,\Z)$ should rotate the electric/magnetic charges of the $k$--th group leaving all the others invariant. 

\begin{figure}
$$
\xymatrix{
*++[o][F-]{\ \phantom{\Bigg|a}\text{\begin{small}$pm$
\end{small}}\phantom{a\Bigg|}\ }
\ar@{-}[r]&
*++[o][F-]{\phantom{\Bigg|}\text{\begin{small}$(p-1)m$
\end{small}}\phantom{\Bigg|}\!\!}
\ar@{-}[r]&
*++[o][F-]{\phantom{\Bigg|}\text{\begin{small}$(p-2)m$
\end{small}}\phantom{\Bigg|}\!\!}
\ar@{-}[r]&\cdots\ar@{-}[r]&
*++[o][F-]{\ \phantom{\Bigg|a}\text{\begin{small}$2m$
\end{small}}\phantom{a\Bigg|}\ }
\ar@{-}[r]&*++[o][F-]{\ \phantom{\Bigg|a}\text{\begin{small}$m$
\end{small}}\phantom{a\Bigg|}\ }
}
$$
\caption{\label{linequiv} The $(A_{pm-1},\widehat{A}(p,1))$ theory as a linear quiver gauge theory. Nodes represent $SU((p-k)m)$ gauge sectors coupled through bifundamental hypermultiplets represented by the edges. The $\beta$--functions of all gauge coupling vanish, except for the leftmost node which has a negative $\beta$-function with coefficient $(p+1)m$.}
\end{figure}
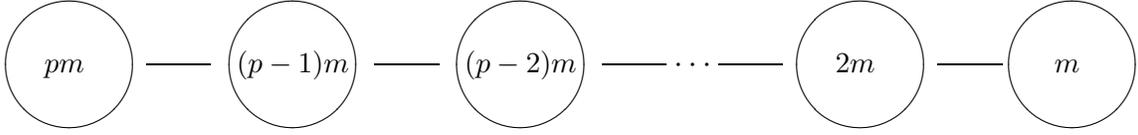

From this physical picture, we get a prediction for the $S$-duality group of the model $(A_{pm-1},\widehat{A}(p,1))$: it should contain at least the
``obvious'' sub-group
\be\label{iiiqazwe}
\Gamma_1\times \Gamma_2\times \cdots\Gamma_{p-1}\times \text{(parabolic)}.
\ee

We shall return to the relation of the $S$-duality group $\Gamma_k$ of each SYM sub-sector with the duality group of the fully interacting theory
in \S.\,\ref{againandangain}. Here we discuss the dualities of the fully interacting theory without reference to the properties of its constituents.
\medskip

The braid presentation of the $S$-duality group $\mathbf{Tel}(A_{pm-1},\widehat{A}(p,1))$ is obtained from
eqns.\eqref{dusu1}--\eqref{dusu2},
by setting $\mathrm{lcm}(N,p)=N$ and $\gcd(N,p)=p$. Note that for $p=1$ (pure SYM) there is no telescopic functor (since there is no flavor charge).
For $p>1$ we have (at least) the $p$ simple telescopic functors $\cl_a$ ($a\in \Z/p\Z$). 

If $p>2$
these functors satisfy the relations
\begin{align}\label{uuuju0}
\cl_{a}\,\cl_{a+1}\,\cl_{a}&=\cl_{a+1}\,\cl_{a}\,\cl_{a+1},
&\cl_{a}\cl_{b}&=\cl_{b}\cl_{a}\ \  \text{for }|a-b|>1,\\
\label{uuuju}
\mathsf{Z}\,\cl_{a}\,\mathsf{Z}^{-1}&=\cl_{a+1}, 
&\mathsf{Z}^p&=1,
\end{align}
that is, they yield a realization of the annular braid group $C\cb_p$.

Then, the $S$-duality group of the linear quiver in figure \ref{linequiv} is given by the product of the cyclic group generated by the quantum monodromy $\mathbb{M}$ by a group of the form
\be
C\cb_p/\mathfrak{N},
\ee
where $\mathfrak{N}$ is a model-dependent normal subgroup which can be determined by writing down the explicit matrices $\boldsymbol{L}_{\tau^a A_0\times R_1}$ specific for each model.
We shall present some concrete examples in the next section.

\paragraph{The special case $p=2$.}
Whenever $p>2$, 
the $S$-duality group $\mathbf{Tel}(A_{pm-1},\widehat{A}(p,1))$
has a presentation in terms of the generators 
$\mathsf{Z},\; \cl_{a}$ subjected to two sets of relations: the \emph{universal} ones 
\eqref{uuuju0},\eqref{uuuju} and the 
\emph{model-dependent} ones given by $\mathfrak{N}$.
When  $p=2$,
we get two simple telescopic functors $\cl_{1}$, $\cl_{2}$ and the universal relation reduce to \eqref{uuuju}.
Explicit computations suggests that the following higher braiding relation holds in all $p=2$ linear quivers starting from $SU(N)$ ($N=2m$)
\be
\overbrace{\cl_1\cl_2\cl_1\cl_2\cdots}^{N\ \text{factors}}=
\overbrace{\cl_2\cl_1\cl_1\cl_1\cdots}^{N\ \text{factors}}
\ee
in analogy with the rule \eqref{kkkdge}.

\subsubsection{$S$-duality in $SO/USp$ linear quiver theories}

The models of the form $(D_{mp+1},\widehat{A}(p,1))$ represents a linear quiver theory with alternating $SO$ and $U\!Sp$
gauge groups \cite{Cecotti:2013lda} as in figure \ref{usposqui}. Edges now represent bi-fundamental half-hypermultiplet which do not carry any flavor charge.
$p=1$ yields pure $SO(2m+2)$ SYM;
to get an interesting $S$-duality group we assume $p>1$.

\begin{figure}
$$
\xymatrix{
*++[o][F=]{\phantom{\Bigg|\!\!}\text{\begin{scriptsize}$\phantom{+}2mp\phantom{+}$
\end{scriptsize}}\phantom{\!\!\Bigg|}}
\ar@{-}[r]&
*++[o][F--]{\phantom{\Bigg|\!\!}\text{\begin{scriptsize}$2m(p-1)$
\end{scriptsize}}\phantom{\!\!\Bigg|}}
\ar@{-}[r]&
*++[o][F=]{\phantom{\Bigg|\!\!}\text{\begin{scriptsize}$2m(p-2)$
\end{scriptsize}}\phantom{\!\!\Bigg|}}
\ar@{-}[r]&
*++[o][F--]{\phantom{\Bigg|\!\!}\text{\begin{scriptsize}$2m(p-3)$
\end{scriptsize}}\phantom{\!\!\Bigg|}}
\ar@{-}[r]&\cdots
}
$$
\caption{\label{usposqui}The $SO/U\!Sp$ quiver gauge theory corresponding to the $(SO(2mp+2), \widehat{A}(p,1))$ model \cite{Cecotti:2013lda}. \textsc{Conventions:} A double circle containing an integer $K$  
stands for an $SO(K+2)$ gauge group, a dashed circle for an $USp(K)$ gauge group, and the edges connecting them bi-fundamental \emph{half}-hypermultiplets.}
\end{figure}
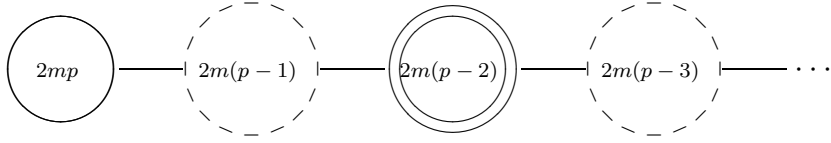

Since the model has no flavor charge,
in view of \S.\ref{phphmot} the interesting auto-equivalences arise from spherical \emph{half}-orbits.

Consider in $\mathscr{D}(D_{mp+1}, \widehat{A}(p,1))=D^b\mathsf{mod}\,\C \vec D_{mp+1}\times \C \widehat{A}(p,1)$ the objects of the form
\be
A_i\otimes R_k
\ee
 where $A_i$ belong to the spherical half-orbit described in \textbf{Fact \ref{Dorbitscirc}}, and $R_k$ ($k\in\Z/p\Z$) is a regular simple in the exceptional tube of period $p$. One has
 \be
 \begin{split}
 &S^{mp}(A_i\otimes R_k)= (\tau^{mp}A_i\otimes \tau^{mp}R_k)\big[2mp\big]=(A_i\otimes R_k)\big[2mp-1\big]\quad\Longleftrightarrow\\
 &\Longleftrightarrow\quad A_i\otimes R_k\ \text{has fractional CY dimension }\frac{a}{b}=\frac{2mp-1}{mp}.
 \end{split}
 \ee
 In particular, $a$ is odd, as required for an half-orbit. If $mp>3$ we get
 the usual $p>1$ simple auto-equivalences
 \be
 \cl_a= H_{A_0\otimes \tau^a R_1}\qquad a\in \Z/p\Z,
 \ee
which for $p>2$ satisfy the same braid relation as before.

 \section{Sub-constituents and dualities}\label{subdualities}
 
Breaking a complicated $\cn=2$ system into simpler sub-constituents is a powerful technique to study its properties. Besides being a useful tool, it provides a clear physical picture of what is going on. 
 
 \subsection{The physical principle}\label{ooooxzq93}
 
Let us start by describing, in rough terms, the physical situation we have in mind.  Often a $\cn=2$ QFT may be decomposed into several distinct physical systems weakly coupled together, the classical example being the decomposition into a Yang-Mills sector and a matter sector. The decomposition into sub-sectors is not unique in general, but depends on a choice of duality frame which specifies which degrees of freedom are weakly coupled in the given physical situation. With respect to such a duality frame, we have a collection of $\cn=2$ theories, $\cf_s$ ($s=1,\cdots, K$), 
coupled together by interaction terms of the form
 \be\label{nav1}
 S=\sum_{s=1}^K S_{\cf_s}+\lambda_{st; ij}\!\int\! d^4x d^4\theta\,\co^{(s)}_i\co^{(t)}_j+O(\lambda_{st;ij}^2),
 \ee
 where $\co^{(s)}_i$ are non-trivial chiral operators in the $s$-th QFT and $\lambda_{st;ij}$ couplings. For example, if $\cf_1$
is a SYM sector with gauge group $G$, and $\cf_2$ is a matter system with a flavor symmetry  $F\supset G$,
 the $\co^{(1)}_i$'s are gauge fields and the $\co^{(2)}_j$'s flavor currents of $G$. 
 
 Suppose that the $s$-th QFT, $\cf_s$, taken in isolation, has a certain duality group 
 $\mathbb{S}(s)$. In general
 $\mathbb{S}(s)$ acts non-trivially on
 the operators $\co^{(s)}_i$,
 \be
 \phi(\co^{(s)}_i)\neq \co^{(s)}_i,\qquad \phi\in\mathbb{S}(s),
 \ee
 and $\phi$ does not extend to a duality of the fully coupled theory $\cf_\text{fully}$.
 For instance, let the matter system $\cf_2$ be $SU(2)$ SQCD with $N_f=4$. This theory has a $SO(8)$ flavor symmetry, and a $SL(2,\Z)$ duality group which acts by $Spin(8)$ triality on the flavor charges \cite{SW2}. We may gauge a subgroup $G\subset SO(8)$.
The $SL(2,\Z)$ duality will not extend to the gauged theory, being broken by the gauge interaction. 
However, the \textit{isotropy \emph{(normal)} subgroup} 
of the interaction
\be\label{nav3}
\ci_\text{int}^{(s)}\equiv\Big\{\phi\in\mathbb{S}(s)\;\Big|\; \phi(\co^{(s)}_i)=\co^{(s)}_i\ \forall\, i\Big\}\subset \mathbb{S}(s)
\ee
is expected, on physical grounds, to extend to a duality of the fully coupled theory $\cf_\text{fully}$. In the example of $SU(2)$ $N_f=4$ SQCD, the isotropy subgroup is the principal congruence subgroup $\Gamma(2)\subset SL(2,\Z)$ \cite{SW2}, and we expect this subgroup of dualities to be preserved by any gauging of the flavor group.

In this example, the isotropy group has \emph{finite index} in the duality group $\mathbb{S}(s)$. This is true for all gauge interactions: we have seen in 
\S.\,\ref{wweeylassdwa} that the action of the $S$-duality group $\mathbb{S}(s)$ on the flavor lattice factorizes through a \emph{finite} reflection group $\mathfrak{W}_s$.
The gauge-interaction isotropy group
$\ci_\text{int}^{(s)}$ is defined by
\be
1\to \ci_\text{int}^{(s)}\to \mathbb{S}(s)\to \mathfrak{W}_s\to 1
\ee
and has finite index.

\subsubsection{UV completeness vs.\! $S$-duality}

We may generalize the argument to arbitrary couplings. A non-trivial interaction \eqref{nav1} between unitary QFTs may be consistent with UV completeness (no Landau poles) only if
\be
\text{dimension}\big(\co^{(s)}_i\co^{(t)}_j\big)\leq 2 \quad\Longrightarrow\quad
\text{dimension}\big(\co^{(s)}_i\big)<2.
\ee
It follows from formulae in \cite{Cecotti:2010fi}, reviewed in appendix \ref{4d2dchiral}, that, in 4d $\cn=2$ theories with a good 2d correspondent, the dimensions $d_4$ of the 4d chiral fields which generate the $\cn=2$ chiral ring $\mathcal{R}_4$ are related to the eigenvalues $e^{2\pi i h}$  of the 2d quantum monodromy $\boldsymbol{H}$ (with $h$ in the range $0\leq h\leq 1/2$) by the formula
\be
d_4=1+\frac{h}{1-\hat c/2},\qquad\text{with }\ 
h\in \left\{\mathrm{Spectrum}\!\left(\frac{1}{2\pi i}\log\boldsymbol{H}\right)\bigg|\; 0\leq h\leq \frac{1}{2}\right\},
\ee
where $\hat c$ is the fractional CY dimension of the 2d brane category $\mathscr{B}$. We may assume $\hat c\geq 1$ without loss.\footnote{\ Otherwise the full $S$-duality group is finite, and all its sub-groups are trivially of finite index.}  The flavor charges $\in \Gamma_\text{flavor}$ correspond to the eigenvalue $1$ of $\boldsymbol{H}$ $(h=0)$, and the corresponding 4d conserved super-currents have dimension $1$.  Only the generators of $\mathcal{R}_4$ may have dimensions $<2$, and they satisfy the bound provided
\be\label{oopibg}
0\leq h< 1-\frac{\hat c}{2}.
\ee
The $S$-duality group $\mathbb{S}(s)$ acts through a finite reflection group on
the $\boldsymbol{H}_{\!s}$ eigenspaces associated to eigenvalues $e^{2\pi i h}$ with $h$ as in \eqref{oopibg}
(compare with eqns.\eqref{veryrelevant},\eqref{veryrelevant2}). 
Then,
$\mathbb{S}(s)/\ci_\text{int}$ is a finite group for all interactions consistent with UV completeness. We are lead to the following
\medskip

\noindent\textbf{``Physical principle''.}
\textit{Modulo commensurability, the $S$-duality group $\mathbb{S}(s)$ of the $s$--th sub-sector is equal to the subgroup of the $S$-duality group $\mathbb{S}_\text{fully}$ of the fully coupled theory which maps the states of the $s$-th sub-sector into themselves.}
\medskip


We now rephrase the above physical picture in the categorial language.
 
\subsection{Constituents and cluster categories}

In general, there is no simple relation between the cluster category of the fully interacting QFT, $\mathscr{C}_\text{fully}$, and the cluster categories of its constituent sub-sectors, $\mathscr{C}(s)$. This is to be expected, since $\mathscr{C}_\text{fully}$ describes the QFT  non-perturbatively in all duality frames and at all couplings, whereas the constituent picture emerges only asymptotically in the limit in which the appropriate coupling $\lambda$ is sent to zero, $\lambda\to 0$, ($\lambda$ being defined with reference to a specific duality frame).
To ``extract'' a constituent sub-sector $\cf_s$ from the interacting cluster category $\mathscr{C}_\text{fully}$, one needs to go through a number of steps which involve extra data and choices:
\begin{itemize}
\item[\textit{i)}] select a cluster-tilting object
$\ct\in \mathscr{C}_\text{fully}$ suitable for the relevant weak coupling limit; 
\item[\textit{ii)}] with respect to the chosen $\ct$, the coupling $\lambda$ (to be sent to zero) is defined by the additional datum of a one-parameter family of  the stability functions
($\cn=2$ central charges) $Z_\lambda$ on the Jacobian module category, i.e.
\be
Z_\lambda\colon K_0\big(\mathscr{C}_\text{fully}\big/\mathsf{add}\,\ct[1]\big)\to \C;
\ee
\item[\textit{iii)}] once specified $\ct$ and $Z_\lambda$ we find the set $\Xi(\lambda)$ of Jacobian modules which are stable for a given value of $\lambda$. A $X\in \Xi(\lambda)$ corresponds to a BPS particle of mass $|Z_\lambda(X)|$.  Modules $X$ which are stable as $\lambda\to0$ such that $|Z_\lambda(X)|$ remains bounded in the limit, then should correspond to the BPS states of the several decoupled sectors. If $\Xi$ is the set of stable objects with bounded mass at $\lambda=0$, we have
$\Xi=\cup_s \Xi_s$ with $\Xi_s$ the set of stable objects describing BPS particles belonging to the $s$--th subsector;
\item[\textit{iv)}] the chosen $Z_\lambda$ produces in the limit $\lambda\to0$ the $s$--th subsector in a particular physical regime (BPS chamber) described by a pair
$(\ct_{(s)},Z_{(s)})$, where  $\ct_{(s)}\in\mathscr{C}(s)$ is a tilting object, and $Z_{(s)}\colon K_0(\mathscr{C}(s)/\mathsf{add}\,\ct_s[1])\to\C$ a stability function. The set of  stable objects of $\mathscr{C}(s)/\mathsf{add}\,\ct_s[1]$ then coincides with $\Xi_s$;
\item[\textit{v)}]  in particular,  the interacting Jacobian module category $\mathscr{J}\equiv\mathscr{C}_\text{fully}/\mathsf{add}(\ct[1])$ should contain all the simples $S_i\in \mathscr{C}(s)/\mathsf{add}\,\ct_s[1]$ (since they are stable for all choices of $Z_{(s)}$)
and
\be
\begin{split}
&\dim \mathrm{Hom}_{\!\mathscr{J}\!}(S_i,S_j)-
\dim \mathrm{Ext}^1_{\!\mathscr{J}\!}(S_i,S_j)-\\
&-\dim \mathrm{Hom}_{\!\mathscr{J}\!}(S_j,S_i)+
\dim \mathrm{Ext}^1_{\!\mathscr{J}\!}(S_j,S_i)\equiv B_{ij}
\end{split}
\ee  
should agree with the exchange matrix $B_{ij}$ of the endo-quiver of $\ct_s$ which belongs to the quiver mutation class of the $s$--th constituent QFT. 
\end{itemize}

Roughly speaking, this procedure sets a (non-intrinsic) correspondence between certain objects of $\mathscr{C}_\text{fully}$ and objects of $\mathscr{C}(s)$, $A \leftrightarrow A_{(s)}$.
Suppose that two corresponding objects,
$A\in\mathscr{C}_\text{fully}$
and $A_{(s)}\in\mathscr{C}(s)$, are both spherical, thus defining two Thomas-Seidel auto-equivalences,
$T_A$ and $T_{A_{(s)}}$, in the respective cluster categories. In this case, it is natural to interpret the duality $T_A$ of the fully interacting theory as arising from the duality $T_{A_{(s)}}$ of the $s$--th constituent sector.

\subsubsection{A simple example in full detail}\label{exindetail}

To illustrate the idea, we study in great detail a simple example which contains all the essential elements of the general case. 
We consider the 4d $\cn=2$ asymptotically free affine theory $\widehat{A}(p,1)$ QFT \cite{Cecotti:2011rv}:
$p=1$ is pure $SU(2)$ SYM, $p=2$ is $SU(2)$ SQCD with $N_f=1$, and in general it is $SU(2)$ SYM coupled to an Argyres-Douglas model of type $D_p$ \cite{Cecotti:2011rv}.
To avoid discussing special cases, we take $p\geq 3$.
We consider the regime in which the Yang-Mills coupling $g_\text{YM}$ is very small; the stable BPS spectrum then consists of the $W$ boson, the BPS states of the matter $D_p$ system, and infinite towers of heavy dyons with magnetic charges $\pm1$ and masses $O(1/g^2_\text{YM})$. As $g_\text{YM}\to0$ the dyons get infinite mass and decouple, and we remain with the $D_p$ matter states plus the $W$.  

The corresponding cluster category $\mathscr{C}(p)$ may be regarded as the category
$\widetilde{\mathsf{coh}}\,\mathbb{X}(p)$ with objects  the coherent sheaves over the weighted projective line $\mathbb{X}(p)$ with a single exceptional point of weight $p$ endowed with the $\Z_2$-graded morphism spaces \cite{BKL}
\be\label{uuuuuqq}
\begin{split}
\mathrm{Hom}_{\widetilde{\mathsf{coh}}}(X,Y)&=\mathrm{Hom}_{\mathsf{coh}}(X,Y)\oplus \mathrm{Ext}^1_{\mathsf{coh}}(X,\tau^- Y)\simeq\\
&\simeq\mathrm{Hom}_{\mathsf{coh}}(X,Y)\oplus D\mathrm{Hom}_{\mathsf{coh}}(Y,\tau^2 X).
\end{split}
\ee 
The magnetic charge of a sheaf is its rank, so the light BPS states with zero magnetic charge are described by finite-length sheaves (i.e.\! skyscrapers with support in a point of $\mathbb{X}(p)$). The $D_p$ matter states (having spins $\leq 1/2$) correspond to \emph{rigid} finite-length sheaves; they belong to the exceptional tube
(i.e.\! they are skyscrapers with support at the exceptional point of weight $p$). The full subcategory $\widetilde{C}_p\subset\mathscr{C}(p)$ over the objects in the exceptional tube is called the \textit{cluster tube} of period $p$ \cite{BKL}. The indecomposable objects of $\widetilde{C}_p$ are the same ones of the usual stable tube $C_p$, but $\widetilde{C}_p$ contains additional \emph{odd} morphisms \eqref{uuuuuqq}. 
The Abelian category $C_p$  is \emph{uniserial}  \cite{bowey}
with $p$ simple objects, $\cs_i$, $i\in\Z/p\Z$, cyclically rotated\footnote{\ We have inverted the numeration of the simple sheaves with respect to ref.\!\cite{shepard} $\cs_i\to \cs_{p-i}$.} by $\tau$:
\be\label{uuuuuqwen}
\tau \cs_i\simeq \cs_{i+1},\qquad \tau^p\cs_i\simeq \cs_i.
\ee
 An indecomposable $\mathscr{E}_{i,\ell}\in C_p$ is  uniquely determined \cite{bowey} by its top $\cs_i$ and length $\ell\in\mathbb{N}$, \begin{gather}
 \mathsf{top}\,\mathscr{E}_{i,\ell}\equiv \mathscr{E}_{i,\ell}\big/\mathrm{rad}\,\mathscr{E}_{i,\ell}\simeq \cs_i,\\ \mathrm{rad}^{\ell-1}\,\mathscr{E}_{i,\ell}\neq0,\quad 
 \mathrm{rad}^{\ell}\,\mathscr{E}_{i,\ell}=0.
 \end{gather}
The periodic tube $C_p$ then is identified with the category $\mathsf{nil}\,\C\widehat{A}(p,0)$ of \emph{nilpotent} finite-dimensional representations of the cyclic quiver $\widehat{A}(p,0)$ (figure \ref{cyccycqq}).

\begin{figure}
$$
\begin{gathered}
\xymatrix{&\bullet\ar[r]^{\phi_1} &\bullet\ar[dr]^{\phi_8}\\
\bullet\ar[ur]^{\phi_2}&&&\bullet\ar[d]^{\phi_7}\\
\bullet\ar[u]^{\phi_3}&&&\bullet\ar[dl]^{\phi_6}\\
& \bullet\ar[ul]^{\phi_4} &\bullet\ar[l]^{\phi_5}}\end{gathered}\qquad \cw=\text{(8-cycle)}=\phi_1\phi_2\phi_3\phi_4\phi_5\phi_6 \phi_7\phi_8
$$
\caption{Example: the cyclic affine quiver $\widehat{A}(8,0)$ and its superpotential.\label{cyccycqq}}
\end{figure}

The rigid bricks of $C_p$ are the indecomposables with length $\ell<p$; the additive closure of the class of rigid bricks in $C_p$  is an Abelian category equivalent to  
the category of modules of the Jacobian algebra of the cyclic quiver
$\widehat{A}(p,0)$ bounded by the ideal generated by the derivatives of the superpotential $\cw=\text{cycle}$. 
The quiver with superpotential
$(\widehat{A}(p,0),\cw=\text{cycle})$ belongs to the mutation class of the Argyres-Douglas SCFT of type $D_p$
\cite{Cecotti:2011rv}. The cluster category of the $D_p$ model is most conveniently written as
\be
\mathscr{C}(D_p)\simeq D^b\mathsf{mod}\,\C D_p/\langle \tau^{-1}[1]\rangle^\Z
\ee
for any Dynkin quiver of type $D_p$; for convenience we orient the Dynkin quiver as
\be
\begin{gathered} 
\xymatrix{&1\\
2\ar[r]& 3\ar[u]\ar[r]& 4\ar[r] &5\ar[r]&\cdots\ar[r]& p-1\ar[r]&p}
\end{gathered}\label{kkkkkawr}
\ee
The indecomposable objects of
$\mathscr{C}(D_p)$ then are the indecomposable modules of $\C D_p$ together with\footnote{\ \label{jjjjjqw12}As always, $P_i$ is the (indecomposable) projective cover of the simple $S_i$ with support at the $i$--th node of the quiver \eqref{kkkkkawr}. Note that the orientation of \eqref{kkkkkawr} differs from the one in figure \ref{refDyqi} by the inversion of the arrow between nodes 2 and 3. Under the isomorphism of the derived categories of the corresponding path algebras, the module $P_2$ for the quiver in figure \ref{refDyqi} becomes the object $P_2[1]$ for the quiver \eqref{kkkkkawr}. Thus the $\Z_2$ automorphism $\theta$ now acts as $P_2[1]\leftrightarrow P_1\equiv S_1$.} $P_i[1]$ ($i=1,\dots,p$).
The rigid bricks in
$\widetilde{C}_p\subset \mathscr{C}(p)$ correpond to the indecomposables of the cluster category $\mathscr{C}(D_p)$ which are not direct summands of $\ct_{D_p}[1]$,
\be\label{7763xv}
\mathscr{C}(D_p)/\mathsf{add}\,\ct_{D_p}[1]\simeq \mathsf{add}\text{(rigid bricks)}\subset \widetilde{C}_p,
\ee
 where the tilting object $\ct_{D_p}$ has the explicit form
\begin{gather}\label{uuurrrq}
T_1\oplus T_2\oplus T_3\oplus \cdots\oplus T_p= P_2\oplus S_1\oplus S_2\oplus P_2/P_4\oplus P_2/P_5\oplus\cdots\oplus P_2/P_p,\\
\mathrm{Hom}_{\mathscr{C}(D_p)}(T_i,T_j[k])\simeq \begin{cases}\C &\text{for }k=0\
\text{and }j\neq i+1\\
0 &\text{otherwise,}
\end{cases}\qquad \text{with }i,j\in\Z/p\Z,
\end{gather}
so that, as it should,
\be
\mathrm{End}_{\mathscr{C}(D_p)}(\ct_{D_p})^\text{op}=\mathsf{Jac}\big(\widehat{A}(p,0),\cw=\text{cycle}\big).
\ee
The AR translation $\tau$ of
$C_p$ induces an auto-equivalence of $\widetilde{C}_p$ which we write as $\widetilde{\tau}$; under the identification \eqref{7763xv}, we have
\be
\widetilde{\tau}\, T_i=T_{i+1}.
\ee
We think of $\widetilde{\tau}$
as a duality of the fully interacting theory which maps matter states into matter states and hence should correspond to an $S$-duality of the matter SCFT of type $D_p$. The $S$-dualities of the Argyres-Douglas theories were computed in \S.\ref{mmmnnbaq}: they are induced by the auto-equivalences of the derived category $\tau$, $\theta$, and $[1]$ (as defined in terms of the Dynkin quiver with the reference orientation of figure \ref{refDyqi}). Using the comment in footnote \ref{jjjjjqw12}, one easily checks that 
\be
T_{i+k}\simeq (\theta\tau)^k T_i,
\ee
where $\simeq$ means equality of their images in the
cluster category $D^b\mathsf{mod}\,\C D_p/\langle \tau^{-1}[1]\rangle^\Z$.  The two sides do not agree in other categories such as the derived or root one: this confirms the idea that the cluster one is the ``right'' physical category.
Thus, as $S$-dualities of the $D_p$ Argyres-Douglas theory,  $\widetilde{\tau}\simeq\theta\tau$.
Note that, for all $p$, we have\footnote{\ Here and below we write $=$ for equivalences which hold in the derived category $D^b\mathsf{mod}\,\C D_p$ and $\simeq$ for equivalences which hold \emph{only} in the cluster (orbit) category $\mathscr{C}(D_p)$.}
\be
(\theta\tau)^{p-1}=\theta [-1]
\simeq (\theta\tau)^{-1}\quad\Rightarrow\quad \widetilde{\tau}^p\simeq \mathrm{Id},
\ee 
in agreement with the fact that $\tau^p=\mathrm{Id}$ in $C_p$. Note that the quantum monodromy $\mathbb{M}$ in the $\widehat{A}(p,1)$ model is $\tau^2$, which induces $\widetilde{\tau}^2$ on
$\widetilde{C}_p$ which is identified with $(\theta\tau)^2=\tau^2$ which is the quantum monodromy for the Argyres-Douglas model of type $D_p$.
\medskip

The matter-sector $S$-duality $\widetilde{\tau}$ is induced by an auto-equivalence of the derived category $D^b\mathsf{mod}\,\C\widehat{A}(p,1)$ of the fully interacting theory which is implemented by the inverse of a telescopic functor $L_{\cs_i}$ (the one point shift \cite{lenzinghandbook}) 
associated to the $\tau$-orbit of an exceptional simple $\cs_i$, $i\in\Z/p\Z$, with $\tau \cs_i=
\cs_{i+1}$. In $C_p$ one has $L_{\cs_1}\simeq \tau^{-1}$ \cite{shepard}. Therefore the
Thomas-Seidel twist $T_{\cs_1}$ of the cluster category $\mathscr{C}(p)$ for the fully interacting theory restricts to the auto-equivalence $\widetilde{\tau}^{-1}$ of the cluster tube $\widetilde{C}_p$.

 All objects in the spherical orbit $\{\tau^k\cs_1\}$
belong to $C_p$ and are rigid,
so they are identified with objects in the Jacobian module category
$\mathscr{C}(D_p)/\mathsf{add}\,\ct_{D_p}[1]$
(cfr.\! eqn.\eqref{7763xv}). In particular, we have the identification
\be
\cs_i\equiv\mathrm{ker}\Big(T_{i+1}\to T_{i}\Big) \longleftrightarrow S_{i-1}\quad \text{for }i=5,\dots,p+1,
\ee
so that the spherical object $\cs_1\in\widetilde{C}_p$  is identified with the simple module $S_p\in\mathsf{mod}\,\C D_p$ whose (half)orbit in $\mathscr{R}(D_p)$ is spherical. This sets a correspondence between the 
twist functor $T_{\cs_1}$ for the fully interacting cluster category $\mathscr{C}(p)$ and the twist functor $T_{S_p}$ for the sub-constituent cluster category $\mathscr{C}(D_p)$
\be
\widetilde{\tau}^{-1}\longleftrightarrow
T_{\cs_1}\longleftrightarrow T_{S_p}.
\ee
This correspondence may be made precise: since $\tilde\tau^{-1}\simeq \tilde\tau^{p-1}\longleftrightarrow T_{\cs_1}^{1-p}$ we compare $\tilde\tau^{-1}$ with the functor $L_{S_p}^{1-p}$.
From eqn.\eqref{Dorbitscirc-1}, \textbf{Fact \ref{Dorbitscirc}}, and \textbf{Remark \ref{Dorbitscirc+1}},
\be
L_{S_p}^{1-p}=\begin{cases}
\tau^{p-1} & p\ \text{odd}\\
\theta^{p-1}\tau^{p-1}& p\ \text{even}
\end{cases}=
\begin{cases}
\theta[-1] & p\ \text{odd}\\
\theta^{-1}[-1] & p\ \text{even}
\end{cases}\simeq
(\theta\tau)^{-1},
\ee
so that the equivalence $\widetilde{\tau}\simeq \theta\tau$ in $\mathscr{C}(D_p)$ may be seen as a correspondence of Thomas-Seidel twists in the interacting and constituent cluster categories.  

\begin{rem} While we have a nice correspondence $T_{\cs_1}\leftrightarrow T_{S_p}$ between auto-equivalences in the cluster categories, the relation between the telescopic functors $L_{\cs_1}$, $L_{S_p}$ is not so good. Indeed, $L_{\cs_1}$ has order $p$ in $D^b C_p/[2]$, while $L_{S_1}$ has order $2p-2$ in 
$\mathscr{R}(D_p)$. In particular, the explicit matrices $\boldsymbol{L}_A$ do not yield (in general) the actual action of the dualities on the constituent charge lattices.   
\end{rem}

 \subsection{$SU$ linear quivers II}\label{againandangain}
 
We return to the $SU$ linear quiver gauge theories of figure \ref{linequiv}.
In \S.\,\ref{aaa15i}
we found in some simple examples a $S$-duality group big enough to accommodate the ``physically expected'' duality group \eqref{iiiqazwe}. Now we discuss how that group is related to the constituent $S$-duality.

The models in figure \ref{linequiv} are special instances of $(G,\widehat{A}(p,1))$ QFTs. They have cluster categories of the form
\be
\mathscr{C}(G,\widehat{A}(p,1))\equiv \mathsf{Hu}_\triangle\!\Big(D^b \mathsf{mod}(\C \vec G\times \C \widehat{A}(p,1))\big/\langle T\rangle^\Z\Big),\qquad T=\tau\otimes\tau.
\ee
Roughly speaking, we may repeat on the second factor,
$\C\widehat{A}(p,1)$ all the constructions we performed in the example of \S.\ref{exindetail}. We then define the \textit{cluster $G$-tube of period $p$} to be
\be\label{uuuuu6547}
\widetilde{C}_p(G)\equiv\mathsf{Hu}_\triangle\!\Big(D^b \mathsf{mod}(\C \vec G\times C_p)\big/\langle T\rangle^\Z\Big)\subset
\mathsf{Hu}_\triangle\!\Big(D^b \mathsf{mod}(\C \vec G\times \C \widehat{A}(p,1))\big/\langle T\rangle^\Z\Big).
\ee
For $G=A_1$ this gives back the usual cluster tube $\widetilde{C}_p$. 
$\widetilde{C}_p(G)$ is related to the $D_p(G)$  SCFT \cite{Cecotti:2013lda} as the cluster tube $\widetilde{C}_p$ is related to the Argyres-Douglas SCFT of type $D_p$. Roughly speaking, $D_p(G)$ is described by the ``relatively rigid'' objects of $\widetilde{C}_p(G)$.

We see the Jacobian algebra
$\mathsf{Jac}(G\boxtimes \widehat{A}(p,1))$ as the completion of the product algebra $\C\vec G\times \C\widehat{A}(p,1)$. The modules of the product algebra are then identified with a class of modules of the Jacobian one (namely the Jacobian modules with vanishing diagonal arrows).   
Let $R_a\in\mathsf{mod}\,\C\widehat{A}(p,1)$ be the regular simples in the exceptional $p$-tube. We write $S_{i,a}$ for the diadic module $S_i\otimes R_a\in\mathsf{mod}\,\C \vec G\times \C\widehat{A}(p,1)$ seen as a module of $\mathsf{Jac}(G\boxtimes \widehat{A}(p,1))$. The skew-symmetric matrix
\be\label{kkkkasqw1}
B_{ia,jb}=\dim \mathrm{Ext}^1(S_{ia},S_{jb})-\dim \mathrm{Ext}^1(S_{jb},S_{ia})
\ee
defines the quiver $Q_{D_p(G)}$ of the constituent $D_p(G)$  sub-sector which should be equipped with the appropriate superpotential $\cw_{D_p(G)}$ \cite{Cecotti:2013lda}.

Now suppose that $\gcd(p,\tilde h(G))>1$. Then $S_i\otimes R_a$ has a spherical (half)orbit in $\mathscr{R}(G,\widehat{A}(p,1))$ which induces a Thomas-Seidel auto-equivalence $T_{S_i\otimes R_a}$ in the cluster category $\mathscr{C}(G,\widehat{A}(p,1))$ of the fully interacting model. $S_i\otimes R_a$ is identified with $S_{i,a}$ which, in turn, gets identified with a simple of the Jacobian algebra of $D_p(G)$, hence with a BPS state of the constituent sub-system $D_p(G)$.
It is clear from eqn.\eqref{uuuuu6547} that $T_{S_i\otimes R_a}$ preserves the cluster $G$-tube
$\widetilde{C}_p(G)$ and it is natural to expect that it sends its ``relatively rigid'' objects into objects of the same kind. In other words, the duality of the fully interacting theory given by the auto-equivalence $T_{S_i\otimes R_a}$ ``restricts'' to a duality of the constituent $D_p(G)$ sector. 

We illustrate this idea in some simple example. 

\subsubsection{Example: $SU(4)\times SU(2)$ with bi-fundamental}\label{lavx4}
We focus on the first example in appendix B of \cite{Cecotti:2012jx}, $m=p=2$, i.e.\! the quiver gauge theory
\be\label{eeexxaqwk}
\begin{gathered}
\xymatrix{ *++[o][F-]{\ \phantom{\Bigg|}\text{\begin{small}$SU(4)$
\end{small}}\phantom{\Bigg|}\ }
\ar@{-}[r] &*++[o][F-]{\ \phantom{\Bigg|}\text{\begin{small}$SU(2)$
\end{small}}\phantom{\Bigg|}\ }}
\end{gathered}
\ee
The $SU(2)$ YM coupling, $g_2$, is exactly marginal, while the $SU(4)$
coupling $g_4$ is asymptotically-free. The quiver
$A_3\boxtimes \widehat{A}(2,1)$ for this model is presented in figure \ref{quia21a3}.

\begin{figure}
\begin{equation*}
\xymatrix{& (1,1) \ar[rrrr]\ar[dddd]\ar[ddl] &&&& (1,2) \ar[rrrr] \ar[dddd]\ar[ddl] &&&& (1,3)\ar[dddd]\ar[ddl]\\
\\
(2,1) \ar[rrrr]\ar[ddr] &&&& (2,2)\ar[rrrr]\ar[ddr]\ar[uulll] &&&& (2,3)
\ar[ddr]\ar[uulll]\\
\\
& (3,1)\ar[rrrr] &&&& (3,2)\ar[llllluu]\ar[lllluuuu]\ar[rrrr] &&&& (3,3)\ar[llllluu]\ar[lllluuuu]}
\end{equation*}
\caption{\label{quia21a3}The quiver
$A_3\boxtimes \widehat{A}(2,1)$.}
\end{figure}
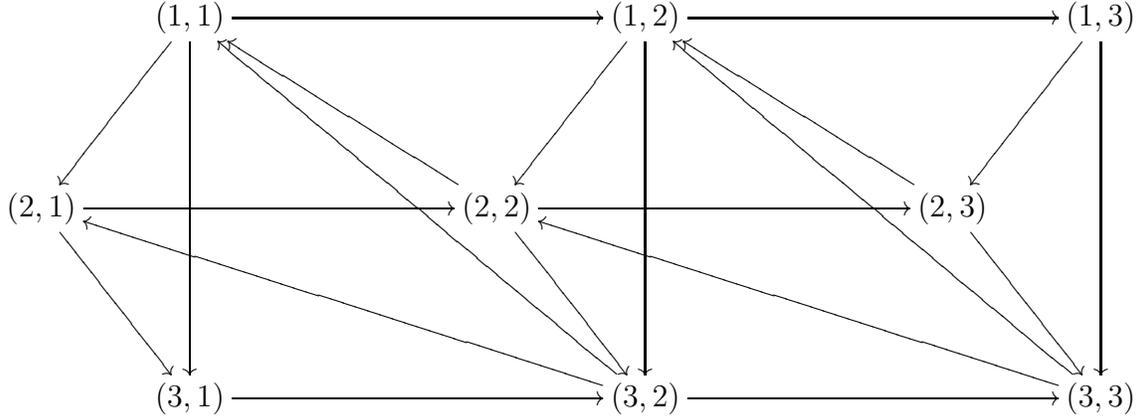

The decoupling limit $g_4\to0$ produce the (free) vectors of $SU(4)$ SYM plus a ``matter'' system which is just $SU(2)$ SQCD with $N_f=4$. 
Considering this matter system in isolation, it has flavor symmetry $SO(8)$, the quarks transforming in the vector representation. The vector representation is fixed by the $\Z_2$ outer automorphism of the $D_4$ graph which interchanges the two 
spinor representations, $s\leftrightarrow c$, and the $SU(4)$ gauge interaction breaks only the group
$\mathfrak{S}_3/\Z_2\simeq \Z/3\Z$ of the flavor triality.
Then the physical arguments of  \S.\ref{ooooxzq93} suggests

\begin{php}\label{tttgbk} The subgroup  $\cs\subset PSL(2,\Z)$ of the $S$-duality of $SU(2)$ with $N_f=4$ which extends to a duality of the fully coupled theory \eqref{eeexxaqwk} has index 3.
\end{php}

We proceed to check this physical statement from the homological side using the strategy outlined around eqn.\eqref{kkkkasqw1}.
The objects $S_{i,a}$ correspond to the diadic objects
\be\label{oiiijl}
S_i\otimes R_a\subset\mathsf{mod}\,\C \vec A_3\times \C\widehat{A}(2,1),
\quad i=1,2,3,\ a=1,2,
\ee
 where $R_2$ is the regular simple with support on the node 2 of the $\C \widehat{A}(2,1)$ quiver
 \be
\begin{gathered}
\xymatrix{& 2\ar[dr]^\rho\\
1\ar[ur]^\phi\ar[rr]_\psi&&3}
\end{gathered}\ee 
and $R_1\equiv \tau R_2$.
Using eqn.\eqref{kkkkasqw1}, we construct the quiver of the ``matter'' sector, see 
 figure \ref{didid}.
This quiver is a well-known member of the mutation class of the $SU(2)$ SQCD $N_f=4$ quiver \cite{Cecotti:2011rv,Alim:2011ae}.
 
 \begin{figure}
 \begin{equation*}
 \xymatrix{&&&S_2\otimes R_2\ar[dl]\ar[rrrd]\\
 S_1\otimes R_2\ar[rrru] && S_1\otimes R_1\ar[dr] && S_3\otimes R_1\ar[ul] && S_3\otimes R_3\ar[llld]\\
 &&& S_2\otimes R_1\ar[lllu]\ar[ru]}
 \end{equation*}
 \caption{\label{didid}The Dirac quiver of the ``matter'' sector of $(A_3,\widehat{A}(2,1))$. It is a quiver in the $SU(2)$ SCQD $N_f=4$ mutation class ($\equiv$ the elliptic Dynkin class $D^{(1,1)}_4$). Notice that the $\Z_2$ Galois automorphism \cite{arnold} of the quiver (which greatly simplifies computation of its BPS spectrum and quantum monodromy \cite{arnold,Cecotti:2013sza,Cecotti:2015lab}) is given by
$S_\alpha\otimes R_2\leftrightarrow S_\alpha\otimes R_1
$
which is the symmetry induced by the charge conjugation $C$ of the fully coupled theory \eqref{eeexxaqwk} (as well as by the involution $1\otimes \tau$).}
 \end{figure}
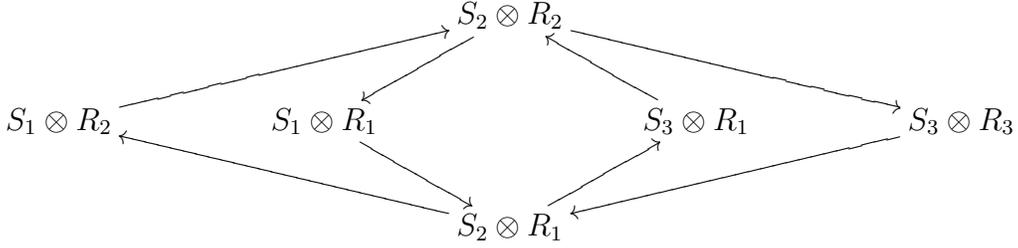
 
The orbits $\{T^k(S_\alpha\otimes R_a)\}$ ($\alpha=1,2,3$, $a=1,2$) are spherical in $\mathscr{R}(A_3,\widehat{A}(2,1)))$;
in facts they are the same spherical orbits we studied in \S.\,\ref{aaa15i}.
The corresponding telescopic functors,
 \be
 \cl_a= L_{S_1\otimes \tau^a R_2},\qquad a\in\Z/2\Z,
 \ee
 then correspond to Thomas-Seidel auto-equivalences of the  cluster category of the $SU(2)$ $N_f=4$ subsector,
 that is, to $S$-dualities of the ``matter system'' which extend to dualities of the fully interacting theory \eqref{eeexxaqwk}. 
 
Since $p=2$, we are in the special case where the usual braiding relations \eqref{uuuju0}\eqref{uuuju} do not apply. However we may read the model-dependent braid relations from the concrete realization of the duality group in terms of the known $9\times 9$ matrices $\boldsymbol{\cl}_a$. One finds the order 4 braid relation
\be
\boldsymbol{\cl}_1\boldsymbol{\cl}_2\boldsymbol{\cl}_1\boldsymbol{\cl}_2=\boldsymbol{\cl}_2\boldsymbol{\cl}_1\boldsymbol{\cl}_2\boldsymbol{\cl}_1,
\ee
which defines the Artin braid group $G_{B_2}$ associated with the Dynkin graph $B_2$ \cite{isobraid,xxx,typeB}
\be
\xymatrix{\bigcirc\!\!\ar@<0.35ex>@{-}[r]^4\ar@<-0.35ex>@{-}[r]&\!\!\bigcirc}
\ee 
For all $n$ one has

\begin{fact}[\!\!\cite{isobraid}\footnote{\ See also \cite{review} \S.\,3.1 \textbf{Example 2}.}] The Artin braid group of type $B_n$, $G_{B_n}$, is an index $n + 1$ subgroup of  $\cb_{n+1}\equiv G_{A_n}$. Indeed, it is the subgroup of braids for which the string beginning in position one also ends in position one.
\end{fact}

In the case of $SU(2)$ with $N_f=4$
the $S$-duality is realized through a 
$\cb_3$ action with a trivial action of its center (recall that $PSL(2,\Z)=\cb_3/Z(\cb_3)$). Then the braid group we find for the sub-constituent,
$G_{B_2}$, has index 3 in the braid group of the isolated $SU(2)$ $N_f=4$ theory, as expected on physical grounds.

\subsubsection{Example: $SU(3)\times SU(2)$ with $(\mathbf{3},\mathbf{2})\oplus (\mathbf{1},\mathbf{2})$}

We consider the case $m=1$, $p=3$, i.e.\! the quiver gauge theory\footnote{\ This is the third example in appendix B of \cite{Cecotti:2012jx}.}
\be\label{eeexxaqwk}
\begin{gathered}
\xymatrix{ *++[o][F-]{\ \phantom{\Bigg|}\text{\begin{small}$SU(3)$
\end{small}}\phantom{\Bigg|}\ }
\ar@{-}[r] &*++[o][F-]{\ \phantom{\Bigg|}\text{\begin{small}$SU(2)$
\end{small}}\phantom{\Bigg|}\ }\ar@{-}[r] &\!\text{\fbox{$\phantom{nn\Bigg|}1\phantom{nn\Bigg|}$}}}
\end{gathered}
\ee
Again the $SU(2)$ YM coupling, $g_2$, is exactly marginal, while the $SU(3)$
coupling $g_3$ is asymptotically-free.
Taking $g_3\to0$ we remain with the free $SU(3)$ gauge vectors plus a matter system which is again  $SU(2)$ SQCD with $N_f=4$. The quiver for the interacting theory is $\vec A_2\boxtimes \widehat{A}(3,1)$. We number the nodes of the affine quiver as in the left figure
\ref{3333qff}.

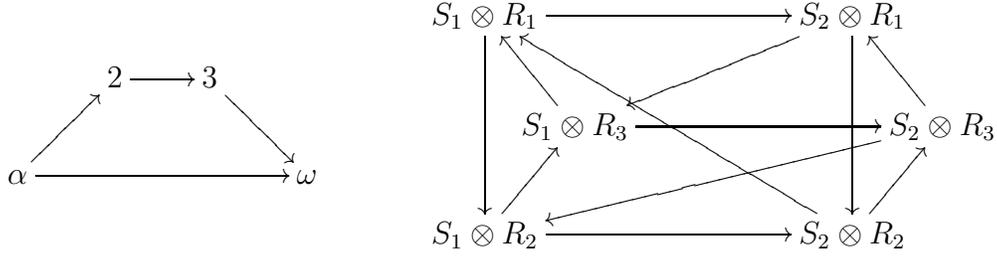
\begin{figure}
$$\begin{gathered}
\xymatrix{&2\ar[r]& 3\ar[rd]\\
\alpha\ar[ur]\ar[rrr]&&& \omega}
\end{gathered}
\qquad\quad
\begin{gathered}
\begin{xy} 0;<1pt,0pt>:<0pt,-1pt>:: 
(0,0) *+{S_1\otimes R_1} ="0",
(0,83) *+{S_1\otimes R_2} ="1",
(34,42) *+{S_1\otimes R_3} ="2",
(139,0) *+{S_2\otimes R_1} ="3",
(139,83) *+{S_2\otimes R_2} ="4",
(173,42) *+{S_2\otimes R_3} ="5",
"0", {\ar"1"},
"2", {\ar"0"},
"0", {\ar"3"},
"4", {\ar"0"},
"1", {\ar"2"},
"1", {\ar"4"},
"5", {\ar"1"},
"3", {\ar"2"},
"2", {\ar"5"},
"3", {\ar"4"},
"5", {\ar"3"},
"4", {\ar"5"},
\end{xy}
\end{gathered}
$$
\caption{\label{3333qff}\textsc{Left:} The $\widehat{A}(3,1)$ affine quiver.
\textsc{Right:} The constituent quiver of $\vec A_2\boxtimes \widehat{A}(3,1)$.}
\end{figure}

The regular simples are the simples with support on the nodes $2$ and $3$, $R_2$ and $R_3$, and the indecomposable $R_1$
of dimension $\alpha+\omega$.
One has $\tau R_a=R_{a+1}$ ($a\in\Z/3\Z$) and $\tau^3 R_a=R_a$.
The simples $S_{i,a}$ correspond to the six diadic objects
$S_i\otimes R_a$ ($i=1,2$) with Dirac pairing\footnote{\ We have changed the sign conventions by an overall $-1$.}
\begin{align}
\langle S_i\otimes R_a,S_i\otimes R_b\rangle&=-\delta^{(3)}_{a,b+1}+\delta^{(3)}_{a,b-1}\\
\langle S_1\otimes R_a,S_2\otimes R_b\rangle&=\delta^{(3)}_{a,b}-\delta^{(3)}_{a,b-1}
\end{align}
leading to the quiver in the right figure \ref{3333qff}. Mutating this quiver at any node one gets the $SU(2)$ $N_f=4$ quiver in the standard form, figure \ref{didid}. 

By the same physical argument
as in \S.\,\ref{lavx4}, we again expect the dualities of the constituent theory which survive the coupling to the $SU(3)$ SYM sector to have index 3 in the $S$-duality group. In facts, we expect the \emph{same} subconstituent duality group.
This time we have $p=3$,
so eqns.\eqref{uuuju0}\eqref{uuuju} apply, and
 we get the cyclic braid group $C\cb_3$. It has correctly index 3 in $\cb_3$. In facts it is the same braid group as in the previous example since: 
 
\begin{fact}[\!\!\cite{xxx}]
One has $C\cb_n\simeq G_{B_n}$.
\end{fact}

\subsubsection{Example: $SU(2N)\times SU(N)$ with bifundamental}

This is the example $p=2$, $m=N$.
The matter SCFT is $SU(N)$ SQCD with $N_f=2N$.
The situation is very similar to the one in \S.\,\ref{lavx4}, except that at $p=2$ the braid relations are model-dependent, and so we expect a different relation between $\boldsymbol{\cl}_1$ and $\boldsymbol{\cl}_2$. Using their explicit matrix realizations
in the Grothendieck group of the root category\footnote{\ We stress again that these matrices do not represent (in general) the action of the dualities on the charges of the constituent theories. } we checked the following relation 
\be\label{uuuuuiiwww}
\overbrace{\boldsymbol{\cl}_1\boldsymbol{\cl}_2\boldsymbol{\cl}_1\boldsymbol{\cl}_2\cdots}^{2N\ \text{factors}}=\overbrace{\boldsymbol{\cl}_2\boldsymbol{\cl}_1\boldsymbol{\cl}_2\boldsymbol{\cl}_1\cdots}^{2N\ \text{factors}}
\ee
 for all $N\leq 6$. Thus we expect that, for all this family of models, the essential part of the $S$-duality group is a quotient of the Artin braid group
 $\cb(I_2(2N))$ associated to the graph $I_2(2N)$
 \be
\xymatrix{\bigcirc\!\!\ar@<0.35ex>@{-}[rr]^{2N}\ar@<-0.35ex>@{-}[rr]&&\!\!\bigcirc}
\ee

\section{A more general framework for $S$-duality}\label{moregeneral}

The analysis of $S$-duality in $(G,L)$ models in  sect.\;\ref{produalities}
exploited the fact that they have a quiver, $G\boxtimes L$,  
of a very convenient form.
The methods of sect.\;\ref{produalities} cannot be applied directly to the models of greatest interest, the $D_p(G)$ and the $D_4^{(1,1)}(G)$, $E_r^{(1,1)}(G)$ SCFTs, since in these cases either no quiver with superpotential is known or
they are not convenient for our purposes. Even for the $SU(2)$ tubular models $D^{(1,1)}_4$, $E_r^{(1,1)}$, which do have nice quivers \cite{Cecotti:2011rv},
the quiver approach is not the best way to study their homological $S$-duality, and in ref.\!\cite{shepard} one used their alternative description in terms of coherent sheaves over weighted projective lines $\mathbb{X}(\boldsymbol{p})$. Here we replace the quiver approach of sect.\;\ref{produalities} with a precise version of the physically motivated idea of \textsf{META}-quivers \cite{Cecotti:2013lda}.
\medskip

Let $\vec G$ be the Dynkin quiver of type $G$ with the reference orientation in figure \ref{refDyqi}. We write $r\equiv r(G)$ for its rank.
 Following Bongartz and Gabriel \cite{Bongartz}, we see $\vec G$ as a bounded $\C$-linear category, whose objects are the nodes $i\in\vec G$ while the Hom-space 
$\mathrm{Hom}(i,j)$ is the vector $\C$-space over the paths in $\vec G$ connecting $i$ to $j$, composition being path concatenation.  

Let $\ch$ be a $\C$-linear Abelian category, with finite-dimensional Hom/Ext-spaces, which is \emph{hereditary} \cite{lenzinghandbook} i.e.
\be
\mathrm{Ext}^k(A,B)=0\ \text{for all }A,B\in \ch,\ \text{and }k\geq 2,
\ee
and has \emph{Serre duality} in the form
\be
\mathrm{Ext}^1(A,B)\simeq D\,\mathrm{Hom}(B,\tau A),
\ee
for a certain functor $\tau\colon\ch\to\ch$ (the AR translation).
Note that we do not ask $\ch$ to have projectives nor injectives; neither we ask the functor $\tau$ to be an equivalence.
We also do not require $\ch$ to have a \emph{tilting object}; instead we
impose the weaker condition that its Grothendieck group $K_0(\ch)$ is a finite-rank lattice.
Examples of categories $\ch$ satisfying the above requirements are:
\begin{itemize}
\item[1)] the category of modules
of a finite-dimensional basic hereditary algebra $\mathsf{mod}\,\C Q$
($Q$ an acyclic quiver). This category has enough projectives and injectives and also tilting objects;
\item[2)] the category of coherent sheaves, $\mathsf{coh}\, \mathbb{X}(\boldsymbol{p})$, over the weighted projective line $\mathbb{X}(\boldsymbol{p})$ of weights $\boldsymbol{p}\equiv (p_1,p_2,\dots, p_s)$. This category has no injectives nor projectives (so $\tau$ is an auto-equivalence), but it has tilting objects;
\item[3)] $C_p$ a stable tube of period $p$, which we can identity with the category $\mathsf{nil}\,\C \widehat{A}(p,0)$, where $\widehat{A}(p,0)$ is the $A^{(1)}_{p-1}$ affine Dynkin graph with the cyclic orientation and $\mathsf{nil}(\cdot)$ stands for the Abelian category of the \emph{nilpotent} (finite-dimensional) modules. This category has no injective, nor surjective, nor tilting objects, but it is uniserial. Moreover $\tau$ is an equivalence satisfying $\tau^p=\mathrm{Id}$.
\end{itemize}
For physical applications to UV complete 4d $\cn=2$ QFTs, we are interested only in a subset of the above examples: 1) with $Q$ Dynkin or affine, 2) with $\chi(\mathbb{X}(\boldsymbol{p}))\equiv 2-\sum_{i=1}^s(1- 1/p_i)\geq 0$, and 3). 
Examples 2) with $\chi(\mathbb{X}(\boldsymbol{p}))>0$ are equivalent to examples 1) with $Q$ affine.

\medskip

We write $\vec G(\ch)$ for the category of linear functors from $\vec G$ to $\ch$
\be
\vec G(\ch)\equiv \mathsf{Funct}(\vec G, \ch).
\ee
A functor $\cx\in \vec G(\ch)$ associates an object $\cx(i)\in\ch$ to each node $i\in \vec G$ and a morphism
\be
\cx(\psi)\in \mathrm{Hom}_\ch\big(\cx(i),\cx(j)\big)
\ee to each arrow $i \xrightarrow{\psi} j$ of $\vec G$.
$\vec G(\ch)$ is a linear Abelian category of global dimension at most $2$, and all
auto-equivalence $\sigma$ of 
$\ch$ induce an auto-equivalence 
$1\otimes \sigma$ of $\vec G(\ch)$ \cite{stillabelian}
\be
1\otimes \sigma\colon \vec G(\ch)\to \vec G(\ch),\qquad  1\otimes \sigma\colon \cx\mapsto\sigma\circ \cx\in\mathsf{Funct}(\vec G,\ch).
\ee

In the special case $\ch=\mathsf{vect}$, the functor category
$\vec G(\mathsf{vect})$ is just the category of (finite-dimensional) modules of the path algebra $\C \vec G$, 
\be
\vec G(\mathsf{vect})=\mathsf{mod}\,\C \vec G.
\ee
More generally, taking $\ch= \mathsf{mod}\,\C G^\prime$ or $\ch=\mathsf{mod}\,\C \widehat{H}$, we get
\be
\vec G(\mathsf{mod}\,\C G^\prime)\simeq \mathsf{mod}(\C G\times \C G^\prime),\qquad 
\vec G(\mathsf{mod}\,\C \widehat{H})
\simeq \mathsf{mod}(\C G\times \C \widehat{H}),
\ee 
so that the functor category $\vec G(\ch)$ may be seen as the natural  generalization of the modules of a product of two hereditary algebras.

We consider the bounded derived category $D^b \vec G(\ch)$. From $\vec G(\mathsf{vect})$
and $\ch$ the derived category inherits the two auto-equivalences
$\tau_G\otimes 1$ and $1\otimes\tau_\ch$ analogous to the ones studied in \cite{kellerper} when $\ch$ is a module category (we use the same notation in the general case). Hence we may define the orbit category
\be\label{rrrrbq}
D^b \vec G(\ch)\Big/\langle \tau_G\otimes \tau_\ch\rangle^\Z,
\ee
the cluster category
\be\label{mmzao}
\mathscr{C}(G,\ch)= \mathsf{Hu}_\triangle\!\Big(D^b \vec G(\ch)\Big/\langle \tau_G\otimes \tau_\ch\rangle^\Z\Big),
\ee
and the root category
\be\label{mmzao}
\mathscr{R}(G,\ch)= \mathsf{Hu}_\triangle\!\Big(D^b \vec G(\ch)\Big/\big[2\Z\big]\Big),
\ee

The correspondence between our 
$\cn=2$ QFTs and cluster categories is then
\begin{equation}\label{sssghjsk}
\begin{minipage}{300pt}
\begin{tabular}{cc}\hline\hline
QFT & cluster category\\\hline
$(G, G^\prime)$ & $\mathscr{C}(G,\mathsf{mod}\,\C G^\prime)\equiv \mathscr{C}(G,G^\prime)$\\
$(G, \widehat{H})$ & 
$\mathscr{C}(G,\mathsf{mod}\,\C \widehat{H})\simeq \mathscr{C}(G,\mathsf{coh}\,\mathbb{X}(p_1,p_2,p_3))$\footnote{ With $\sum_{i=1}^31/p_i>1$.}
\\
$D_p(G)$ & $\mathscr{C}(G, C_p)$\\
$D^{(1,1)}_4(G)$ & $\mathscr{C}(G,\mathsf{coh}\,\mathbb{X}(2,2,2,2))$\\
$E^{(1,1)}_6(G)$ & $\mathscr{C}(G,\mathsf{coh}\,\mathbb{X}(3,3,3))$\\
$E^{(1,1)}_7(G)$ & $\mathscr{C}(G,\mathsf{coh}\,\mathbb{X}(4,4,2))$\\
$E^{(1,1)}_8(G)$ & $\mathscr{C}(G,\mathsf{coh}\,\mathbb{X}(6,3,2))$\\\hline\hline
\end{tabular}
\end{minipage}
\end{equation}
The $S$-duality group
of each $\cn=2$ model (of this class) is then
\be
\mathrm{Aut}\,\mathscr{C}(G,\ch)\big/\mathrm{Aut}\,\mathscr{C}(G,\ch)^0,
\ee 
which is again essentially equal to the
correspondent auto-equivalence group for the root category $\mathscr{R}(G,\ch)$ using the argument in \S\S.\,\ref{comppparing},\;\ref{fffistcaase}.

\subsection{Diadic functors}
Let $X\in\vec G(\mathsf{vect})$ be a Dynkin module; we write $x_i=\dim X(i)$
and choose bases, so that $X(i)$ gets identified with $\C^{x_i}$ and 
$X(\psi)$ with a $x_{t(\psi)}\times x_{s(\psi)}$ complex matrix. 

We define a product $\otimes\colon \vec G(\mathsf{vect})\times \ch\to \vec G(\ch)$ as follows (here $A\in\ch$ and $X\in\vec G(\mathsf{vect})$)
\begin{align}
(X\otimes A)(i)&= \overbrace{A\oplus A\oplus\cdots \oplus A}^{x_i\ \text{summands}}\\
(X\otimes A)(\psi)&= X(\psi)\otimes \mathrm{Id}_A.\label{lllaq12}
\end{align} 
Functors of the form $X\otimes A$ will be called \emph{diadic.}

The category $\vec G(\ch)$ is homologically generated by the 
diadic functors.\footnote{\ For this an other homological assertions on the category $\vec G(\ch)$, see appendix B.}
In facts, the functors of the form
$S_i\otimes A$, where $S_i$ are the simples of $\vec G(\mathsf{vect})$, suffice to generate the full $\vec G(\ch)$. The same statement holds with the $S_i$ replaced by their projective covers $P_i$ (or injective envelopes $I_i$).
The following ``Kunneth formula''
\be\label{kunneth}
\mathrm{Ext}^k_{\vec G(\ch)}(X\otimes A, Y\otimes B) =\bigoplus_{i+j=k}
\mathrm{Ext}^i_{\vec G(\mathsf{vect})}(X,Y)\otimes \mathrm{Ext}^j_\ch(A,B)
\ee
is shown in appendix \ref{homapp}. In particular the Euler form
\be
\chi(\ca,\cb)=\sum_{k\in\Z} (-1)^k\dim \mathrm{Ext}^k(\ca,\cb),
\ee
for diadic objects factorizes
\be
\chi(X\otimes A, Y\otimes B)=\chi(X,Y)\cdot \chi(A,B).
\ee

Suppose that $\ch$ has a tilting object
$T=\bigoplus_{s\in S}T_s$, as it happens in all our examples but $C_p$. Let $P_i$ be the projective cover of $S_i$ in $\vec G(\mathsf{vect})$. We claim that
\be
\ct =\bigoplus_{i=1}^r \bigoplus_{s\in S} P_i\otimes T_s
\ee 
is a tilting object in $\vec G(\ch)$.
Indeed, one has to show that:
\textit{i)} $\mathrm{Ext}^k(\ct,\ct)=0$ for $k\geq 1$, and \textit{ii)} the $P_i\otimes T_s$ generate $\vec G(\ch)$. Statement \textit{ii)} follows from the observations after eqn.\eqref{lllaq12}, while \textit{i)}
is automatic in view of eqn.\eqref{kunneth} and the definition of tilting object $T$. Then
\be
D^b\big(\vec G(\ch)\big)\simeq D^b\big(\mathsf{mod}\,\mathrm{End}_{\vec G(\ch)}(\ct)\big),
\ee 
and 
\be
K_0(\vec G(\ch))\simeq\Z^{r|S|}.
\ee
 The last two equations remain true if $T$ is a tilting object in $D^b(\ch)$ (which is simply the repetitive category of the hereditary category $\ch$ \cite{lenzinghandbook}).
Now we are in a position to check 
the consistency of the identifications \eqref{sssghjsk} with the ones given by the 4d/2d correspondence.
We only need the equivalences \cite{BKL}
\begin{align}
D^b(\mathsf{coh}\,\mathbb{X}(4,4,2))&\simeq \underline{\mathsf{vect}}\,\mathbb{X}(4,4,2)\simeq D^b(\mathsf{mod}(\C A_3\times \C A_3))\\
D^b(\mathsf{coh}\,\mathbb{X}(6,3,2))&\simeq \underline{\mathsf{vect}}\,\mathbb{X}(6,3,2)\simeq D^b(\mathsf{mod}(\C A_5\times \C A_2)),
\end{align}
to realize that the categories of the $E^{(1,1)}_7(A_{n-1})$, $E^{(1,1)}_8(A_{n-1})$ SCFTs are properly identified. All other cases may be seen as straightforward generalizations of these ones.
In appendix B one shows the following

\begin{fact}
Let $X\in \vec G(\mathsf{vect})$ be a rigid brick. Then the linear functor
$J_X\colon \ch\to \vec G(\ch)$,
\be
J_X\colon A\mapsto  X\otimes A,\qquad
J_X\colon \psi\mapsto \mathrm{Id}_X\otimes\psi,
\ee
embeds $\ch$ as a full exact subcategory closed under extensions.
The left (or right) derived functor $J_X$ embeds fully faithfully $D(\ch)$ into
$D(\vec G(\ch))$.\end{fact}

Recall that by Gabriel theorem \cite{assem} all indecomposable objects $X\in\vec G(\mathsf{vect})$ are rigid bricks. Their dimension vectors $[X]\in K_0(\vec G(\mathsf{vect}))$ are the positive roots of 
$G$ under the natural identification of the Grothendieck group and root lattice of $G$,
and an indecomposable $X$ is uniquely identified (up to isomorphism) by its Grothendieck class $[X]$.

\begin{corl} $X\in \vec G(\mathsf{vect})$ indecomposable. Then $J_X\colon \ch\to \vec G(\ch)$ sends bricks into bricks and preserves the spin and $R$-symmetry charges of the corresponding BPS particle (cfr.\! eqn.\eqref{mmmkkklq}).\end{corl}

We shall write $\boldsymbol{J}_X$ for the homomorphism of Grothendieck groups induced by $J_X$
\be
\boldsymbol{J}_X\colon K_0(\ch)\to K_0(\vec G(\ch)),\qquad 
\boldsymbol{J}_X[A]\mapsto [J_X(A)]\equiv [X]\otimes [A]. 
\ee

\subsection{The canonical SYM sector}

As mentioned above,\footnote{\ See ref.\!\cite{Cecotti:2012va} for more details.} the categories $\mathsf{coh}\,\mathbb{X}(\boldsymbol{p})\equiv \vec A_1(\mathsf{coh}\,\mathbb{X}(\boldsymbol{p}))$ with $\chi(\mathbb{X}(\boldsymbol{p}))\geq0$ describe UV complete $SU(2)$ gauge theories coupled to a set $(D_{p_1},\cdots D_{p_s})$ of Argyres-Douglas matter systems. 
The BPS spectrum at weak Yang-Mills coupling consists of the $W$ boson, the BPS states of the $D_{p_i}$ matter systems, and magnetically charged states (all of them being hypermultiplets for $\chi(\mathbb{X}(\boldsymbol{p}))>0$, while for $\chi(\mathbb{X}(\boldsymbol{p}))=0$ we have also dyonic vector multiplets).
In all cases there is a canonical choice of $S$-duality frame in which the $W$ boson corresponds to the $\mathbb{P}^1$-family of skyscraper brick sheaves
$\{\cs_\lambda\}_{\lambda\in\mathbb{P}^1}$ \cite{Cecotti:2012va,shepard}.

Likewise, the functor category $\vec G(\mathsf{coh}\,\mathbb{X}(\boldsymbol{p}))$ describes $\cn=2$ SYM with gauge group $G$ coupled to a set $(D_{p_1}(G),\cdots D_{p_s}(G))$ of SCFTs of type $D_p(G)$ \cite{Cecotti:2013lda}. In the standard duality frame
the $W$ bosons of the $G$ SYM sector are described as follows.
By Gabriel theorem \cite{assem}, the map $\vec {G}(\mathsf{vect})\to K_0(\vec {G}(\mathsf{vect}))$, $X\mapsto \dim X$, yields a one-to-one correspondence between the rigid bricks $X_\alpha\in \vec G(\mathsf{vect})$ and the positive roots of $G$,
\be
\dim X_\alpha =\alpha \in \Delta^+(G).
\ee
The $W$ bosons of $G$, associated to the positive roots correspond to the $\mathbb{P}^1$-families of functors
\be
J_{X_\alpha}(\cs_\lambda)\in \vec G(\mathsf{coh}\,\mathbb{X}(\boldsymbol{p})),\qquad
\alpha\in\Delta^+(G),\ \lambda\in\mathbb{P}^1. 
\ee
This choice of duality frame defines the magnetic charges $m_i(-)$ for the gauge group $G$
defined as
\be
C_{ij}\, m_j(X)=\chi(J_{\alpha_i}(\cs_\lambda),X)-\chi(X,J_{\alpha_i}(\cs_\lambda)), 
\ee 
where $C_{ij}$ is the Cartan matrix of $G$. The states surviving the decoupling limit $g_\text{YM}\to0$ are described by the category controlled by the functions $\{m_i(-)\}$; once eliminated the gauge $W$-bosons, what remains correspond to the matter constituents $(D_{p_1}(G),\cdots D_{p_s}(G))$.

In other words, the electric charges of the $W$-bosons associated to the simple roots $\alpha_i$ of the  gauge group $G$ have the form $\alpha_i\otimes e$, 
with $e\equiv[\cs_\lambda]$ the charge vector of the $SU(2)$ $W$-boson, while
the magnetic charges have the form $\alpha_i^\vee\otimes m$, with $m$ the $SU(2)$ magnetic charge.
Concretely,  a functor $\cx\in \vec G(\mathsf{coh}\,\mathbb{X}(\boldsymbol{p}))$, associates to each object (node) $i\in\vec G$ a coherent sheaf $\cx(i)\in \mathsf{coh}\,\mathbb{X}(\boldsymbol{p})$. The electric and magnetic $G$ charges of $\cx$ are suitable linear combinations of the degrees and ranks of the sheaves $\cx(i)$. The charge vector of the $i$--th simple $W$-boson has support on the $i$--th object and vanishing rank, while the $i$--th magnetic charge of the functor $\cx$ is
\be\label{uuuuuwwwwx}
m_i(\cx)= \mathrm{rank}\,\cx(i).
\ee

 In particular,
for the category $\mathsf{coh}\,\mathbb{X}(\boldsymbol{p})\equiv \vec A_1(\mathsf{coh}\,\mathbb{X}(\boldsymbol{p}))$, the $SU(2)$ magnetic charge coincides with the rank. Hence the controlled category is the Abelian subcategory of \emph{finite length} objects $\mathscr{S}$ 
\cite{lenzinghandbook}
\be\label{uuuuqwazw}
\mathscr{S}=\bigvee_{\lambda\in \mathbb{P}^1} C_\lambda
\ee
with all $C_\lambda$ homogeneous stable tubes but for $s$ exceptional ones which have periods $p_i$.

The subcategory of $\vec G(\mathsf{coh}\,\mathbb{X}(\boldsymbol{p}))$ controlled by the magnetic charges
\eqref{uuuuuwwwwx} is then
\be
\vec G(\mathscr{S})=\bigvee_{\lambda \in \mathbb{P}^1} \vec G(C_\lambda),
\ee
and, in the decoupling limit $g_\text{YM}\to 0$ we may study the controlled category of functors tube by tube. Then we are reduced to the case of $\vec G(C_p)$, with $C_p$ a stable tube of period $p$.

\subsection{Auto-equivalences}

\subsubsection{Induced auto-equivalences}

Let $\sigma$ be an auto-equivalence of the derived category $D^b\ch$. It extends to an equivalence $1\otimes \sigma$ of the derived category $D^b\vec G(\ch)$ and, since it commutes with
$T\equiv \tau_G\otimes \tau_\ch$ as well as with the double shift $[2]$, also to auto-equivalences of the cluster $\mathscr{C}(G,\ch)$ and root  $\mathscr{R}(G,\ch)$ categories.
The same statement holds for the elements $\vartheta$ of the derived Picard group $\mathrm{Aut}(D^b\mathsf{mod}\,\C \vec G)$, 
which induce auto-equivalences $\vartheta\otimes 1$ of all three categories $D^b\vec G(\ch)$,
$\mathscr{C}(G,\ch)$ and $\mathscr{R}(G,\ch)$.

The group homomorphism
\be
1\otimes -\colon \mathrm{Aut}\,D^b\ch\to\mathrm{Aut}\,\mathscr{C}(G,\ch) 
\ee
has kernel $\tau^{h(G)}[-2]$ or $\tau^{h(G)/2}[-1]$, depending on $G$. 
In all cases
\be
\ker(1\otimes -)\subset Z(\mathrm{Aut}\, D^b \ch).
\ee
In particular, the auto-equivalences generated by telescopic functors $L_A$ of $D^b\ch$ which satisfy non-trivial braiding relations (if any), are mapped to non trivial auto-equivalences
\be\label{indueeqac}
\mathscr{L}_A\equiv 1\otimes L_A,
\ee
which satisfy the same braiding relations as the $L_A$.

\subsubsection{Telescopic functors} 

In addition, to the auto-equivalence inherited from $D^b\ch$ (and $D^b\mathsf{mod}\,\C \vec G$), we may have new auto-equivalences produced by the telescopic functors associated to spherical objects in $\mathscr{C}(G,\ch)$. The arguments of \S\S.\,\ref{comppparing},\;\ref{fffistcaase}
extend to this more general case: the spherical objects in 
$\mathscr{C}(G,\ch)$ are in correspondence with the spherical (half)orbits in the root category $\mathscr{R}(G,\ch)$.

The spherical (half)orbits which are easy to describe are the diadic ones. Here we focus on the group of dualities associated to the diadic telescopic functors.
As before, we have

\begin{fact} If $\gcd(p,\tilde h(G))>1$, the tensor product of an object in $A\in D^b\ch$ which belongs to a spherical orbit $\{\tau^k A\}$ times an object in $X\in D^b\mathscr{R}(G)$ which belongs to a spherical (half)orbit $\{\tau^k X\}$ yields a spherical (half)orbit $T^k(X\otimes A)\in \mathscr{R}(G,\ch)$, and hence telescopic autoequivalences $L_{X\otimes \tau^a A}$, where $a\in \Z/\!\gcd(p,\tilde h)\Z$.
If $\gcd(p,\tilde h)>2$, the telescopic auto-equivalences $L_{X\otimes\tau^a A}$ satisfy the braiding relations
of $C\cb_{\gcd(p,\tilde h)}$.
\end{fact}

The induced auto-equivalences $\mathscr{L}_A$ act on the telescopic auto-equivalence $L_{X\otimes B}$ by the adjoint action
\be
\mathscr{L}_A\, L_{X\otimes B}=L_{X\otimes\mathscr{L}_A(B)}\,\mathscr{L}_A,
\ee
thus generating the telescopic functors associated to new spherical orbits.

\subsection{The case $\ch=\mathsf{coh}\,\mathbb{X}(\boldsymbol{p})$ with $\mathbb{X}(\boldsymbol{p})$   tubular}

The essentially new case, where novel phenomena appear is when $\ch$ is the category of coherent sheaves over a tubular line.

\subsubsection{Review of the $G=SU(2)$ case}
The hereditary categories with interesting telescopic autoequivalences are the coherent sheaves of the four tubular weighed projective lines, $\mathsf{coh}\,\mathbb{X}(\boldsymbol{p})$ with $\chi(\boldsymbol{p})=0$, having weights,
\be
\boldsymbol{p}=(2,2,2,2),\quad (3,3,3),\quad (4,4,2),\quad (6,3,2).
\ee
They correspond to the four  $SU(2)$ SCFT with matter in the fundamental. In this sub-section we shall write simply $\mathbb{X}$ for a weighted projective line of tubular type.
We set $p\equiv p_1\equiv \mathrm{lcm}(p_i)$ equal to a maximal weight, that is, respectively, $2$, $3$, $4$, and $6$. One has
\be
\tau^p=\mathrm{Id}\qquad\text{in }D^b\mathsf{coh}\,\mathbb{X}(\boldsymbol{p})\ \ \text{with }\chi(\mathbb{X}(\boldsymbol{p}))=0.
\ee

On $D^b \mathsf{coh}\,\mathbb{X}(\boldsymbol{p})$
act 
two independent telescopic functors which we may take to be associated to the spherical $\tau$-orbits of the structure sheaf $\co$ and of a simple sheaf with support at an exceptional point of maximal period $p$, $\cs_{i,1}$. These two telescopic functors
satisfy the $\cb_3$ relation \cite{lenzingmeltzer,meltzer}
\be
L_\co L_{\cs_{i,1}} L_\co=
L_{\cs_{i,1}} L_\co L_{\cs_{i,1}},
\ee
and in fact generate the full $\cb_3$ group.
The center $Z(\cb_3)$ of $\cb_3$ is the infinite cyclic group generated by 
\be
(L_\co L_{\cs_{i,1}})^3=\begin{cases}\tau^{-3}[1] &p\neq 3\\
\pi_{23}[1] & p=3
\end{cases}
\ee
where $\pi_{23}$ is the permutation which exchanges the last two special points.
The full autoequivalence group $\mathrm{Aut}(D^b\mathsf{coh}\,\mathbb{X})$ is a semi-direct product of $\cb_3$ and the geometric auto-equivalences
\be\label{iiiiggsa}
1\to \mathsf{Pic}(\mathbb{X})^0\ltimes \mathrm{Aut}(\mathbb{X})\to
\mathrm{Aut}(D^b\mathsf{coh}\,\mathbb{X})\to \cb_3\to 1,
\ee 
where $\mathsf{Pic}(\mathbb{X})^0$ is the group of the degree zero line bundles (acting by tensor product $X\mapsto X\otimes \cl$) and $\mathrm{Aut}(\mathbb{X})$ is the group of geometric automorphisms of $\mathbb{X}$ which is (roughly) the group of permutation of the special points having the same weight $p_i$.

One has
\be
SL(2,\Z)=\cb_3/Z(\cb_3)^2
\ee
and this is the quotient of $\mathrm{Aut}(D^b \mathsf{coh}\,\mathbb{X})$ which acts effectively on the $SU(2)$ electromagnetic charges\footnote{\ The $SU(2)$ electric charge $e$ of a coherent sheaf $\ce\in \mathsf{coh}\,\mathbb{X}$ is its degree normalized so that a simple in the tube of largest period has degree $1$ (that is, the generic skyscraper has degree $p$); the $SU(2)$ magnetic charge $m$ of $\ce$ is its rank as a sheaf.}   $(e,m)$,
\be\label{ssssllzzz}
\begin{pmatrix}e\\ m\end{pmatrix}
\longmapsto \begin{pmatrix} a & b\\
c& d\end{pmatrix}\begin{pmatrix}e\\ m\end{pmatrix},\qquad
\begin{pmatrix} a & b\\
c& d\end{pmatrix}\in SL(2,\Z), 
\ee and which is the $S$-duality group in the strict sense of the word. Writing eqn.\eqref{ssssllzzz} we mean that, for all $M\in SL(2,\Z)$, we may find (non uniquely) an auto-equivalence $\sigma_M\in \mathrm{Aut}(D^b\,\mathsf{coh}\, \mathbb{X})$ such that, for all derived coherent sheaves $X$,
\be
\begin{pmatrix}
\mathrm{deg}\,\sigma_M(X)\\
\mathrm{rank}\,\sigma_M(X)
\end{pmatrix}= M 
\begin{pmatrix}
\mathrm{deg}\,X\\
\mathrm{rank}\,X
\end{pmatrix}.
\ee

The subgroup which acts trivially on the other conserved charges (flavor as well as internal electromagnetic charges of the $D_{p_i}$ matter systems)
is the principal congruence subgroup $\Gamma(p)\subset SL(2,\Z)$ \cite{shepard}. Then we have an effective action on the matter charges of the (finite) quotient group $SL(2,\Z/p\Z)$.

\subsubsection{General gauge group $G$}
We consider the
cluster (resp.\! root) category
\be
\mathscr{C}(G,\mathsf{coh}\,\mathbb{X}), \qquad \Big(\text{resp. } \mathscr{R}(G,\mathsf{coh}\,\mathbb{X})\Big),
\ee
where $\mathbb{X}$ is a weighted projective line of tubular type.
The corresponding 4d $\cn=2$ theories are the $D^{(1,1)}_4(G)$, $E^{(1,1)}_r(G)$ SCFTs of ref.\!\cite{DelZotto:2015rca}. 
They have the physical interpretation of SYM with gauge group $G$ coupled to four or three ``matter'' SCFTs of type $D_{p_i}(G)$ such that the $G$ Yang-Mills coupling $g_\text{YM}$ is exactly marginal.
\medskip

$\mathscr{C}(G,\mathsf{coh}\,\mathbb{X})$ has at least two kinds of auto-equivalences:
\begin{itemize}
\item[A)] the ones inherited by $\mathrm{Aut}(D^b\mathsf{coh}\,\mathbb{X})$, eqn.\eqref{iiiiggsa}.
In particular, all DZVX models have a $SL(2,\Z)$ group of dualities generated by 
the two induced functors 
(cfr.\! eqn.\eqref{indueeqac}),
\be\label{oooiqa75}
\mathscr{L}_\co\quad \text{and}\quad \mathscr{L}_{\cs_{i,1}},
\ee
which acts diagonally on all  the electric/magnetic charges of the gauge group $G$
\be
\begin{pmatrix}
\mathrm{deg}(1\otimes\sigma_M\cdot \cx)(i)\\
\mathrm{rank}(1\otimes\sigma_M\cdot\cx)(i)
\end{pmatrix}= M\!
\begin{pmatrix}
\mathrm{deg}\,\cx(i)\\
\mathrm{rank}\,\cx(i)
\end{pmatrix},\qquad M=
\begin{pmatrix}
a& b\\ c & d
\end{pmatrix}\in SL(2,\Z).
\ee
 This is the duality predicted in \cite{DelZotto:2015rca}. Note that the group of induced auto-equivalences $\simeq\mathrm{Aut}(D^b\mathsf{coh}\,\mathbb{X})$, eqn.\eqref{iiiiggsa}, is strictly larger than the expected modular group $SL(2,\Z)$;
\item[B)] for particular pairs (affine star, simply-laced Lie algebra) there is a further enhancement of the $S$-duality group due to the presence of spherical objects in $\mathscr{C}(G,\ch)$ 
(equivalently, of spherical orbits in $\mathscr{R}(G,\ch)$).
\end{itemize}

As in section \ref{produalities}, the spherical objects (orbits) which are easy to 
detect are the diadic ones (but, of course, there may  also be non diadic spherical orbits harder to find).
Again, tensor products of spherical orbits with periods which are not coprime, yield spherical orbits. In the present case we get two different kinds of such orbits (and associated telescopic auto-equivalences):
\begin{itemize}
\item[a)] for each period $p_a\in \boldsymbol{p}$ such that $\gcd(p_a,\tilde h(G))>1$ and each spherical (half)orbit $\{\tau^k X\}\in\mathscr{R}(G)$ we have the spherical (half)orbit $T^k(X\otimes\cs_{j,a})$, $j\in \Z/\!\gcd(p_a,\tilde h(G))\Z$, where $\cs_{j,a}$ are the simple sheaves with support at the $a$--th exceptional point of $\mathbb{X}$. This class of orbits exists independently of the fact that the $G$ Yang-Mills coupling is marginal, and correspond to $S$-dualities of the individual matter constituent of type $D_{p_a}(G)$ which extend to dualities of the fully coupled QFT, as described in section \ref{subdualities};
\item[b)] if the overall Yang-Mills beta-function vanish, i.e.\! for $\mathbb{X}$ of tubular type, we typically have additional spherical orbits (and dualities) of a more interesting kind since they correspond to an \emph{unexpected} enhancement of the $S$-duality group.
If $\gcd(p,\tilde h(G))>1$ we have 
telescopic functors of the form
\be
L_{X\otimes\tau^j \co}\qquad j\in\Z/\!\gcd(p, \tilde h(G))\Z,
\ee
which are additional dualities of the fully coupled theory which do not arise from the dualities of the single constituents.
\end{itemize} 

The telescopic auto-equivalences in items a), b) satisfy a number of braiding relations. If $\gcd(p,\tilde h(G))>2$ we have the universal relations of third order as discussed in section \ref{produalities}.
In addition, we have the model dependent braiding relations
similar to the ones encountered in section \ref{produalities}.
Finally, we have relations between the telescopic auto-equivalences and the induced
auto-equivalences,
in particular the ones induced by the telescopic functors of $D^b\mathsf{coh}\,\mathbb{X}$, eqn.\eqref{oooiqa75}.
 
In the next section we shall discuss a few interesting examples and write down the duality group in detail.

\begin{rem}\label{ttrema}
As discussed around eqn.\eqref{oooqe}, sometimes we have more than one convenient way to write the cluster category of a $\cn=2$ QFT. Some dualities may be manifest in one realization, and other dualities in a different realization. Of course, the physical $S$-duality group should contain all such dualities, and the comparison between different realizations  may help (in some instance) to detect dualities which otherwise would go unnoticed. For instance, we have
\be
E^{(1,1)}_6(A_2)\simeq (D_4,D_4)\simeq E^{(1,1)}_6(A_2)^\prime,
\ee
where $E^{(1,1)}_6(A_2)^\prime$ stands for a second \emph{inequivalent} identification of the $(D_4,D_4)$
SCFT with the $E^{(1,1)}_6(A_2)$ theory obtained by interchanging the role of the two $D_4$'s in the identification.
In facts, using the automorphisms of $D_4$, we get two other inequivalent ``$E^{(1,1)}_6(A_2)$ structures'', 
$E^{(1,1)}_6(A_2)^{\prime\prime}$ and  $E^{(1,1)}_6(A_2)^{\prime\prime\prime}$.
Each identification induces its own auto-equivalences 
\be\{\mathscr{L}_\co, \mathscr{L}_{\cs_{i,1}}\},\ \{\mathscr{L}^\prime_\co, \mathscr{L}^\prime_{\cs_{i,1}}\},\ 
\{\mathscr{L}^{\prime\prime}_\co, \mathscr{L}^{\prime\prime}_{\cs_{i,1}}\},\ \text{and }
\{\mathscr{L}^{\prime\prime\prime}_\co, \mathscr{L}^{\prime\prime\prime}_{\cs_{i,1}}\}
\ee 
all of which belong to the $S$-duality group of the theory.
The corresponding argument says that the $E_6^{(1,1)}(D_4)$ SCFT has 10 inequivalent  ``$E_6^{(1,1)}(D_4)$ structures'', each of which induces a distinct pair
$\mathscr{L}_\co$, $\mathscr{L}_{\cs_{i,1}}$, each pair generating a $SL(2,\Z)$ subgroup of the $S$-duality group (and also a copy of the finite group
$\mathsf{Pic}(\mathbb{X})^0\ltimes \mathrm{Aut}(\mathbb{X})$, cfr.\! 
eqn.\eqref{iiiiggsa}). We shall discuss this more in detail in the next section.
\end{rem}


\section{The  sequence of ``cubic'' $\cn=2$ SCFTs}

We have seen above that for certain special models we have an enhancement of the $S$-duality group which becomes much larger than usual. This enhancement is maximized for a family of $\cn=2$ models which we call the ``cubic'' sequence.\medskip 

We consider the 4d $\cn=2$ SCFTs whose 2d correspondent is a Landau-Ginzburg with a cubic superpotential in $n$ chiral superfields
\be\label{llllq12}
W=\sum_{i=1}^n X_i^3+\sum_{i,j,k\atop\text{all distinct}} a_{ijk}\,X_iX_jX_k+\text{lower order}
\ee
One has $\hat c= n/3$, so the condition $\hat c<2$ yields $n\leq 5$. The five SCFTs in the cubic sequence are listed in table \ref{pqpqrt}. The coupling constants $a_{ijk}$ correspond to exactly marginal deformations of the 4d SCFT
\be
\#\big(\text{marginal deformations}\big)={n \choose 3}.
\ee
The 4d model has flavor symmetry $F$ if and only if $n$ is even. In this case
\be
\mathrm{rank}\,F={n\choose n/2}.
\ee
The last two models in the sequence have a weakly-coupled Lagrangian formulation: they are the generalized quiver gauge theories in figures \ref{ggg123}, \ref{ggg124}.

\begin{table}
\begin{center}
\begin{tabular}{lcllccc}\hline\hline
$n$ & ${n\choose 3}$ & SCFT & a.k.a. &$\mathscr{L}$? & $G$ & $F$\\\hline
1 & 0 & Argyres-Douglas $A_2$ &$(A_2,A_1)$&&& \\
2 & 0 & Argyres-Douglas $D_4$ & $(A_2,A_2)$&&&$SU(3)$\\
3 & 1
& $E^{(1,1)}_6$ i.e.  $E^{(1,1)}_6(A_1)$& $(D_4,A_2)$&&$SU(2)$&\\
4 & 4 
& $E_6^{(1,1)}(A_2)$ & $(D_4,D_4)$ &$\checkmark$ & $SU(3)\times SU(2)^3$ & $U(1)^6$ \\
5 & 10 & $E^{(1,1)}_6(D_4)$ &&
$\checkmark$ &$SO(8)\times SO(5)^3\times SO(3)^6$&\\\hline\hline
\end{tabular}
\caption{\label{pqpqrt} The cubic $\cn=2$ SCFTs. The symbol $\checkmark$ means that the SCFT has a weakly coupled Lagrangian formulation. $G$ is the Yang-Mills gauge group whose coupling constants are exactly marginal. $F$ is the flavor group.}
\end{center}
\end{table}

The $n$-th model in the sequence has a root category
with $3^{n-1}$ obvious spherical orbits
which lead to $3^{n-1}$ telescopic functors. In addition
we have
\be
2{n\choose 3}
\ee
induced auto-equivalences of the
 form $\mathscr{L}_A$, see \textbf{Remark \ref{ttrema}}. 
These generators are related by several generalized braiding relations.

With some abuse of notation, we write the $3^{n-1}$ telescopic functors in the form
\be\label{unnsqw31}
L_{\tau^{k_1}S_1\otimes \tau^{k_2}S_1\otimes\cdots\otimes \tau^{k_{n-1}}S_1\otimes S_1},\qquad k_i\in \Z/3\Z.
\ee
The precise meaning of the abusive notation is that the 
matrices giving the action of the telescopic functors in the Grothendieck group of the root category have the form
\be\label{jjjkkasw}
\boldsymbol{L}_{\tau^{k_1}S_1\otimes \tau^{k_2}S_1\otimes\cdots\otimes \tau^{k_{n-1}}S_1\otimes S_1}=\boldsymbol{1}-\sum_{k=1}^3\left(\bigotimes_{s=1}^n \boldsymbol{H}_s^{k+k_s}\boldsymbol{a}_s\right)\!\otimes\!\left(\bigotimes_{s=1}^n\boldsymbol{a}_s^t\boldsymbol{E}_s\boldsymbol{H}_s^{-k-k_s}\right),
\ee
where $\boldsymbol{H}_s$, $\boldsymbol{E}_s$, and $\boldsymbol{a}_s$ are $n$ copies of (respectively) the monodromy, Euler matrix, and simple charge vector for the (2,2) minimal model of type $A_2$. 
In the \textsc{rhs} of \eqref{jjjkkasw} we set $k_n=0$ as a convention. To simplify the notation, we write $\boldsymbol{L}(k_i)$ for the matrix in eqn.\eqref{jjjkkasw}.
In addition, for each triplet of integers $\{i,j,k\}\subset \{1,2,\dots,n\}$ we have the two induced functors of the several ``$E^{(1,1)}_6(G)$ structures''
see \textbf{Remark \ref{ttrema}}.
Using \textbf{Remark \ref{tttuuuuura}},
we may equivalently write  \emph{three} induced functors per triplet; their Grothedieck  matrix
takes the form ($\alpha=1,2,3$)
\begin{gather}
\boldsymbol{L}_{ijk}(\alpha) 
=\boldsymbol{1}-\sum_{k=1}^3\! \left(\bigotimes_{s=1}^n \boldsymbol{H}_s^{k+\alpha\delta_{s,i}}\right)\!\boldsymbol{K}_{ijk}\!
\left(\bigotimes_{s=1}^n \boldsymbol{H}_s^{-k-\alpha\delta_{s,i}}\right)\\
\boldsymbol{K}_{ijk}=
\boldsymbol{1}\otimes\cdots \otimes \boldsymbol{K}_i\otimes \cdots\otimes \boldsymbol{K}_j\otimes \cdots\otimes \boldsymbol{K}_k\otimes \cdots\otimes \boldsymbol{1}\\
\boldsymbol{K}_s= \boldsymbol{a}_s\otimes \boldsymbol{a}_s^t \boldsymbol{E}_s
\end{gather}
A moment thought suggests that this is not yet the full story.
Let $\Xi\subset \{1,2,\dots,n\}$ be a \emph{non-empty} subset and write
\be\label{tt1}
\boldsymbol{K}_\Xi=\bigotimes_{s=1}^n \boldsymbol{K}_\Xi(s),\qquad 
\boldsymbol{K}_\Xi(s)=\begin{cases}
\boldsymbol{K}_s & \text{if }s\in\Xi,\\
\boldsymbol{1}_s &\text{otherwise}.
\end{cases}
\ee
and
\be\label{tt2}
\boldsymbol{L}_\Xi\Big(\{k_s\}_{s\in\Xi}\Big)=\sum_{k=1}^3\! \left(\bigotimes_{s=1}^n \boldsymbol{H}_s^{k+k_s}\right)\!
\!\left(\frac{\boldsymbol{1}}{3}+(-1)^{|\Xi|}\boldsymbol{K}_\Xi\right)\!\!
\left(\bigotimes_{s=1}^n \boldsymbol{H}_s^{-k-k_s}\right)
\ee
Note that the \textsc{rhs} depends on the $k_s$ mod 
$k_s\to k_s+k$, so that for each
$\Xi$ we get $3^{|\Xi|-1}$ generators. In total
\be
\frac{1}{3}\sum_{\emptyset\neq\Xi\subset\{1,\dots,n\}} 3^{|\Xi|}=
\frac{1}{3}\left(\sum_{\alpha_i\in \{0,1\}^n}3^{\sum_i \alpha_i}-1
\right)=\frac{4^n-1}{3}=\sum_{k=1}^n{n\choose k}3^{k-1}.
\ee
For instance, for $n=5$ we have 341 braid generators. Of course, the same matrix group is generated by
one generator $\boldsymbol{L}_\Xi$ per subset $\Xi$
together with the $n$ matrices
$\boldsymbol{H}_s$; since $\boldsymbol{H}_s^{-1}=\boldsymbol{L}_{\{s\}}$, we remain with only $2^n-1$
generators, but in this economical presentation the generalized braid group structure is less manifest.

\begin{figure}
$$
\xymatrix{ &&\text{\fbox{1}}\ar@{-}[d]\\
&&*++[o][F-]{2}
\ar@{-}[d]\\
\text{\fbox{1}}\ar@{-}[r] &*++[o][F-]{2}\ar@{-}[r]&*++[o][F-]{3}\ar@{-}[r]& *++[o][F-]{2}\ar@{-}[r]&\text{\fbox{1}}}
$$
\caption{\label{ggg123}The gauge theory $E^{(1,1)}(A_2)$. A circled $K$
stands for a $SU(K)$ gauge group, \fbox{$f$} for $f$ fundamentals, and edges for bi-fundamental hypermultiplets.
}
\end{figure}
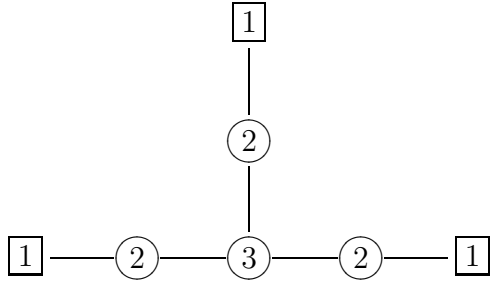


\begin{figure}
$$
\xymatrix{ &*++[o][F=]{3}\ar@{-}[r]&\bullet\ar@{-}[d]\ar@{-}[r]&*++[o][F=]{3}\\
*++[o][F=]{3}\ar@{-}[d]&&*++[o][F=]{5}
\ar@{~}[d]&&*++[o][F=]{3}\ar@{-}[d]\\
\bullet\ar@{-}[r] &*++[o][F=]{5}\ar@{~}[r]&*++[o][F=]{8}\ar@{~}[r]& *++[o][F=]{5}\ar@{-}[r]&\bullet\\
*++[o][F=]{3}\ar@{-}[u]&&&&*++[o][F=]{3}\ar@{-}[u]}
$$
\caption{\label{ggg124}The gauge theory $E^{(1,1)}(D_4)$. A doubly-circled $K$
stands for a $SO(K)$ gauge group, a curly edge for a $\tfrac{1}{2}(\boldsymbol{8},\boldsymbol{4})$ half bi-spinor,
and a $\bullet$ for a $\tfrac{1}{2}(\boldsymbol{4},\boldsymbol{2},\boldsymbol{2})$ half tri-spinor. 
}
\end{figure}
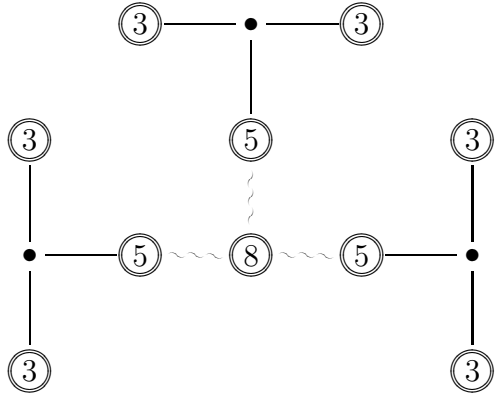

We also have a ``geometric''
$\mathfrak{S}_n$ auto-equivalence which generalizes the auto-equivalence $\mathrm{Aut}(\mathbb{X})\equiv\mathfrak{S}_3$ for the tubular projective line of weights $(3,3,3)$ (which describes the $n=3$ cubic SCFT). Up to the action of $\mathfrak{S}_n$
we may reduce to one
subset per cardinality, say
to the sub-sets $[k]\equiv\{1,\dots,k\}\subset \{1,\dots,n\}$.
Then
\be
\boldsymbol{L}_{[k]}= \boldsymbol{M}_{[k]}\otimes \overbrace{\boldsymbol{1}\otimes \cdots\otimes\boldsymbol{1}}^{n-k\ \text{factors}}
\ee
where $\boldsymbol{1}$ stands for the $2\times 2$ identity matrix and $\boldsymbol{M}_{[k]}$ is the
$2^k\times 2^k$ matrix such that
\be
\begin{aligned}
\big(\boldsymbol{M}_{[k]}\big)_{1,i}&=\big(v_{[k]}\big)_i\\
\big(\boldsymbol{M}_{[k]}\big)_{i,1}&=-1\\
\big(\boldsymbol{M}_{[k]}\big)_{i,j}&=\delta_{i,j}-\delta_{i,2^k}\,\delta_{j,2^k}\qquad\text{for } i,j\neq 1
\end{aligned}
\ee
with $v_{[k]}$ the vector in $ \Z^{2^k}$
defined by the recursion
\be
\left\{\begin{aligned}
& v_{[0]}=-1\\
&v_{[k+1]}=\big(v_{[k]},-v_{[k]}\big).
\end{aligned}\right.
\ee
\be
(1-\boldsymbol{L}_{[k]})(\boldsymbol{L}_{[k]}+\boldsymbol{H}_{[k]}^{-1})=0.
\ee
The minimal equations for the
matrices $\boldsymbol{L}_\Xi$ are written in table \ref{minhhh}. We stress that these matrices have no direct bearing on the action of the dualities on the charges. However the braiding relation between the $\boldsymbol{L}_\Xi$ are mapped into braiding relations between the corresponding Thomas-Seidel twists $T_\Xi$ for the cluster category. This yields a presentation of the subgroup $\mathrm{Tel}(\mathscr{C})$ of the $S$-duality group.

\begin{table}
$$\begin{array}{c}\hline\hline
1+\boldsymbol{L}_{\{s\}}+\boldsymbol{L}_{\{s\}}^2=0\\
(1-\boldsymbol{L}_{\{s,t\}}+\boldsymbol{L}_{\{s,t\}}^2)(\boldsymbol{L}_{\{s,t\}}^2-1)=0\\ 
\big(\boldsymbol{L}_{\{s,t,u\}}^3-1\big)\big(\boldsymbol{L}_{\{s,t,u\}}-1\big)=0 \\
 (1-\boldsymbol{L}_{\{s,t\}}+\boldsymbol{L}_{\{s,t,u,v\}}^2)(\boldsymbol{L}_{\{s,t,u,v\}}^2-1)=0\\
\big(\boldsymbol{L}_{\{1,2,3,4,5\}}^3-1\big)\big(\boldsymbol{L}_{\{1,2,3,4,5\}}-1\big)=0\\\hline\hline
\end{array}$$
\caption{\label{minhhh} Minimal equations for the matrices $\boldsymbol{L}_\Xi$.}
\end{table}

The generalized braid groups have two kinds of relations,  \emph{binary} and \emph{ternary}
\cite{stbraid}:
\begin{itemize}
\item binary of order $n$: they involve two letters $s$ and $t$
\be
\overbrace{ststst\cdots}^{n\ \text{factors}}=\overbrace{tststs\cdots}^{n\ \text{factors}} 
\ee
\item ternary of order $n$: they involve three letters $s$, $t$
and $u$
\be
\overbrace{stustu\cdots}^{n\ \text{factors}}=\overbrace{tustus\cdots}^{n\ \text{factors}}
\ee 
which may or not be cyclic 
in the three letters $s$,$t$,$u$.
\end{itemize}

For instance, consider the $n=5$ cubic model and focus on the subgroup of the
$S$-duality group generated by the 81 telescopic functors of 
the form \eqref{unnsqw31}.
A search by Mathematica of generalized braid relations between the corresponding 81
matrices \eqref{jjjkkasw}
produced relations of the following kinds:
\begin{itemize}
\item binary relations of order 2 and 3;
\item cyclic ternary relations of order 4 and 12.
\end{itemize}

\begin{rem}
Pragmatically, the way we determined the duality group used only properties of the hypersurface singularity $W=0$ with  $W$ as in eqn.\eqref{llllq12}. For 4d $\cn=2$ SCFT defined by a quasi-homogeneous hypersurface singularity $W$, we can use the same explicit formulae.
Most such models are already covered by the analysis in the previous sections, but there are some special 4d SCFT defined by a quasi-homogeneous singularity which have no simple description in terms of a category of the form $\vec G(\ch)$. Examples are the model defined by the singularity
\be\label{555tri}
W=X^5+Y^5+Z^5
\ee
with fractional CY dimension $\hat c=9/5$, or some higher Arnold singularities \cite{arnold,DelZotto:2011an}.
The brane category is well known \cite{orlov1}, and we easily write down
the matrices $\boldsymbol{L}_A$ which represent the telescopic auto-equivalences on its Grothendieck group: 
they are written in terms of simple data of the singularity
(Stokes matrix, monodromy, and roots). It is tempting to conjecture that the $S$-duality group is again commensurable to the matrix group constructed out of the geometric data of the singularity.
As an example for $(5,5,5)$ triangle singularity in eqn.\eqref{555tri}, the relevant matrices are given by eqns.\eqref{tt1}\eqref{tt2} $n=3$ and the order of $\boldsymbol{H}_s$ for the $A_1$ minimal model, 3, replaced by the order of the monodromy of the $A_4$ minimal model, $5$.

\end{rem}

\appendix

\section{4d/2d correspondence and 4d chiral operators}\label{4d2dchiral}
%
%
%
%
%

The 4d/2d correspondence \cite{Cecotti:2010fi} states that --- for a certain class of 4d $\cn=2$ models --- the exchange matrices $B_{ij}$ of their quivers arise as the BPS counting matrices of 2d $(1,1)$ models with $\hat c<2$. More precisely, for a $\cn=2$ QFT in this class there is a 2d $(1,1)$ theory with $n$ supersymmetric vacua and $0\leq \hat c <2$ such that
\begin{equation}
 B=S^t-S
\end{equation}
where $S$ is the unipotent integral $tt^*$ Stokes matrix of the 2d model \cite{Cecotti:1992rm}. In a suitable basis the matrix $S$ is upper triangular with 1's along the diagonal, and the off--diagonal (generically) integral entries count the 2d BPS states as in \cite{Cecotti:1992rm}. Quiver mutations correspond to 2d wall--crossing.
The matrix $\boldsymbol{H}=(S^t)^{-1}S$ is the 2d quantum monodromy with eigenvalues
\begin{equation}\label{ssss}
 \Big\{\exp\!\big(2\pi i (q_a-\hat c/2)\big)\ \colon\ q_a\equiv \text{UV $U(1)_R$ charges of 2d chiral primaries}\Big\}.
\end{equation}
In particular, only Stokes matrices such that the eigenvalues of $\boldsymbol{H}$ are roots of unity may correspond to unitary 2d QFT.

It follows from 2d PCT that the set $\{q_a\}$ is symmetric under \begin{equation}\label{rrrr}q_a\longleftrightarrow \hat c-q_a.\end{equation} The 2d theory has always an operator with $q_a=0$, namely the identity, so $\exp(\pm 2\pi i \hat c/2)$ are always eigenvalues of $\boldsymbol{H}$. \emph{A priori} this fixes $\hat c/2$ only mod 1, but since $0\leq \hat c/2 <1$, the value of $\hat c$ is uniquely fixed once we know \emph{which}  
eigenvalue of $\boldsymbol{H}$ is to be identified with $\exp(2\pi i \hat c/2)$.  Only eigenvalues consistent with the symmetry \eqref{rrrr} may be identified with $\hat c$. $\hat c$ is also determined as the fractional CY dimension of the corresponding derived brane category.

4d flavor charges correspond to zero--eigenvectors of $B$, $(S-S^t)\psi=0$. Now
\begin{equation}
 S\psi=S^t\psi \quad\Longleftrightarrow\quad \boldsymbol{H}\psi\equiv (S^t)^{-1}S\psi=\psi,
\end{equation}
so flavor charges correspond to eigenvectors of the 2d quantum monodromy associated to the eigenvalue $+1$, that is,
comparing with eqn.\eqref{ssss}, to 2d chiral primaries of dimension $\hat c/2\mod 1$. Since 2d unitarity imples $q_a\leq \hat c<\hat c/2+1$, the dimension of the 2d `flavor' operators $\co_f$ is precisely $\hat c/2$. The dual parameters $m_f$ in the 2d action
\begin{equation}
 S_0+\sum_f\left(\int d^2z\, d^2\theta\: m_f\,\co_f+\mathrm{H.c.}\right)
\end{equation}
have a 2d $U(1)_R$ charge $1-\hat c/2$. From the 4d viewpoint the $m_f$'s, being dual to conserved flavor charges, have the dimension of masses; so
\begin{equation}
 \text{4d dimension}= \frac{\text{2d $U(1)_R$ charge}}{1-\hat c/2}.
\end{equation}
In particular, the dimensions of the operators parametrizing the Coulomb branch are given by the $k$ numbers
\begin{equation}
\Big\{\Delta_1,\Delta_2,\cdots,\Delta_k\Big\}\equiv \left\{ \frac{1-q_a}{1-\hat c/2}\ \ \text{such that }q_a < \hat c/2\right\}
\end{equation}
which are determined by $\boldsymbol{H}$ and the identification of which eigenvalue is identified with $\exp(2\pi i \hat c/2)$, up to a few mod 1 ambiguities.
Note that for an interacting theory $\Delta_\ell >1$, as required by 4d unitarity.
Since the minimal $q_a$ is always zero, the \textit{largest} dimension of a Coulomb branch operator is given by
\begin{equation}\label{pppa}
 \Delta_k= \frac{1}{1-\hat c/2}.
\end{equation}

\section{Some homological results}\label{homapp}

\begin{prop} 
The category $\vec G(\mathcal H)$ is homologically generated by the diadic objects of the form $S_i\otimes A$ (resp.\! $P_i\otimes A$) where $A\in\ch$ and the $S_i$ are the simples of $\vec G(\mathsf{vect})$ (resp.\! $P_i$ are the projective covers of the $S_i$).
\end{prop}
Before giving the proof, we define the concept of \emph{linear orientation} of a Dynkin quiver:
\begin{defn}
A Dynkin quiver with $n$ vertices numbered as $1,2,\dots, n$ is \emph{linearly oriented} iff there exists only one source, labelled by ``1'', and for all arrows $i \to j$, one has $j>i$.
\end{defn}
\begin{proof}[Proof of the proposition]
Let $\vec G$ be a \emph{linearly oriented} Dynkin quiver. 
Let $\vec G_{[i]}$ be the full subquiver of $\vec G$ over the nodes $\{1,2,....,i\}$. The node $i+1$ is connected to $\vec G_{[i]}$ by a single arrow $\psi_{i+1}$ with target $i+1$.
Let $\cx \in \vec G(\mathcal H)$ be a functor;
we write $\cx_{[i]}$ for its 
restriction to the full $\C$-linear subcategory $\vec G_{[i]}\subset \vec G$. Clearly, we may see $\cx_{[i]}$ as an object of $\vec G(\ch)$  extending it by zero. In the Abelian category $\vec G(\ch)$ we have the exact sequence
\be
0\to S_i\otimes \cx(i)\to \cx_{[i]}\to\cx_{[i-1]}\to 0.
\ee
Since $\cx_{[1]}=\cx(1)\otimes S_1$, and $\cx\equiv\cx_{[n]}$, by recursion on $i$ we get the first claim. The second one follows from the exact sequence
\be
0\to \bigoplus_{j\in J(i)}P_j\otimes A\to P_i\otimes A \to S_i\otimes A\to 0.
\ee
\end{proof}

It is clear that
\begin{equation}
\mathrm{Hom}_{\vec G(\mathcal H)}(X_1 \otimes A_1,X_2 \otimes A_2)\simeq \mathrm{Hom}_{\vec G(\mathsf{vect})}(X_1,X_2) \otimes \mathrm{Hom}_{\ch}(A_1,A_2).
\label{eq:homhom}
\end{equation}
More generally, for diadic functors one has the ``Kunneth formula''
\begin{prop}
For $A,B \in \mathcal H$ and $X,Y \in \vec G(\mathsf{vect})$, one has
\begin{equation}
\mathrm{Ext}^k_{\vec G(\ch)}(X \otimes A,Y \otimes B)=\bigoplus_{i+j=k}\mathrm{Ext}^i_{\vec G(\mathsf{vect})}(X,Y)\otimes \mathrm{Ext}^j_{\ch}(A,B),\ \ \forall k
\label{eq:exts}
\end{equation}
In particular, the global dimension of $\vec G(\ch)$ is one more the global dimension of $\ch$. Since we assume $\ch$ to be hereditary, $\mathrm{gl.dim.}\,\vec G(\ch)=2$. The Serre functor
of $\vec G(\ch)$ is the product of the Serre functors for $\vec G(\mathsf{vect})$ and $\ch$.
\be
S=\tau[1]\otimes \tau[1]\equiv \tau\otimes\tau [2].
\ee
\end{prop}

The proof proceeds by induction on the lengths of $X$ and $Y$. The expression is true for simples; then one shows that if it is true for $X$ (resp.\! $Y$) it is true also for an extension of $X$ (resp.\! $Y$) by a simple.

One has a simple formula for the Euler form
\be
\chi_{\vec G(\ch)}(X, Y ) :=\sum_{ s
=0}^2
(-1)^k \mathrm{Ext}^k_{\vec G(\ch)}(X, Y )
\ee
which depends only on the Grothedieck classes $[X]$, $[Y ]$ of $X$, $Y$. One has
\be
\chi_{\vec G(\mathcal H )}(X_1 \otimes A_1; X_2 \otimes A_2) = \chi_{\vec G(\mathsf{vect})} (X_1, X_2) \chi_{\ch}(A_1, A_2).\ee

\subsection{Embedding of autoequivalences}

\begin{prop}
$J_R:\mathcal H \to \mathsf{Funct}(\vec G, \mathcal H), \ A \mapsto R \otimes A$, with $R\in\vec G(\mathsf{vect})$ a rigid brick, is a full exact embedding
whose image is closed under extension.
\end{prop} 
\begin{proof}
We have to show that the functor is fully faithful:
\be
\begin{split}
&\mathrm{Hom}_{\vec G(\mathcal H)}(J_R(A),J_R(B)) \cong  \mathrm{Hom}_{\vec G(\mathcal H)}(R\otimes A,R\otimes B ) \cong \\
&\cong\mathrm{Hom}_{\vec G(\mathsf{vect})}(R,R) \otimes \mathrm{Hom}_{\ch}(A,B) \cong \mathrm{Hom}_{\mathcal H}(A,B),
\end{split}\ee
where in the last equality we  used that $R$ is a brick. 
To see that the image is closed under extensions notice that
for a rigid brick $R$
\be
\mathrm{Ext}^1_{\vec G(\mathcal H)}(R \otimes A, R \otimes B)\simeq
\mathrm{Ext}^1_\ch(A,B).
\ee
The elements of the \textsc{rhs} are the classes of exact sequences of the form $0\to B\xrightarrow{\beta} C\xrightarrow{\alpha} A\to 0$,
which correspond in the \textsc{lhs} to the classes of exact sequences of the form   \be
0\to R\otimes B\xrightarrow{\mathrm{Id}_R\otimes \beta} R\otimes C\xrightarrow{\mathrm{Id}_R\otimes \alpha} R\otimes A\to 0.\qedhere\ee
\end{proof}

\end{document}